\documentclass[openright,twoside]{iitbthesis}

\usepackage{epsfig}
\usepackage{amssymb}
\usepackage{amsthm}
\usepackage{amsmath}
\input comment.sty

\newtheorem{theorem}{Theorem}[chapter]

\newtheorem{claim}[theorem]{Claim}

\newtheorem{cor}[theorem]{Corollary}
\newtheorem{definition}[theorem]{Definition}

\newtheorem{lemma}[theorem]{Lemma}

\newcommand{\sort}[1]{\mbox{Sort}(#1)}
\newcommand{\scan}[1]{\mbox{Scan}(#1)}
\newcommand{\perm}[1]{\mbox{Perm}(#1)}
\newcommand{\msf}[1]{\mbox{MSF}(#1)}
\newcommand{\ssf}[1]{\mbox{SF}(#1)}
\newcommand{\var}[1]{\mbox{Var}(#1)}
\newcommand{\colour}[1]{\mbox{Col}(#1)}

\begin{document}


\clearpage\pagenumbering{roman}  

\title{Efficient Algorithms and Data Structures for Massive Data Sets}
\author{Alka}
\date{March 2010.}


\iitbdegree{Doctor of Philosophy}

\thesis

\department{Department of Computer Science and Engineering}

\maketitle

\begin{certificate}
It is certified that the work contained in the thesis entitled
       {\sf \bf Efficient Algorithms and Data Structures for Massive Data Sets}
       by {\bf Alka},
       has been carried out under my supervision and that this
       work has not been submitted elsewhere for a degree.

\vspace*{25mm}

\noindent
        Dr. Sajith Gopalan \\
        Associate Professor \\
        Department of Computer Science and Engineering\\
        Indian Institute of Technology, Guwahati.\\

\vspace*{10mm}
\noindent
        March 2010.

\end{certificate}


\begin{abstract}
  For many algorithmic problems,
traditional algorithms that optimise on the number of instructions executed
prove expensive on I/Os. Novel and very different design techniques,
when applied to these problems, can produce algorithms that are I/O efficient. 
This thesis adds to the growing chorus of such results.
The computational models we use are the external memory model and the W-Stream model.

On the external memory model, we obtain the following results.
(1) An I/O efficient algorithm for computing minimum spanning trees of graphs
that improves on the performance of the best known algorithm.
(2) The first external memory version of soft heap, an approximate meldable priority queue.
(3) Hard heap, the first meldable external memory priority queue that matches
the amortised I/O performance of the known external memory priority queues, while
allowing a meld operation at the same amortised cost.
(4) I/O efficient exact, approximate and randomised algorithms for the minimum cut problem, 
which has not been explored before on the external memory model.
(5)~Some lower and upper bounds on I/Os for interval graphs.

On the W-Stream model, we obtain the following results.
(1) Algorithms for various tree problems and list ranking that match the performance
of the best known algorithms and are easier to implement than them.
(2) Pass efficient algorithms for sorting, and the maximal independent set problems,
that improve on the best known algorithms.
(3)~Pass efficient algorithms for the graphs problems of finding vertex-colouring, 
approximate single source shortest paths, maximal matching, 
and approximate weighted vertex cover.
(4)~Lower bounds on passes for list ranking and maximal matching.

We propose two variants of the W-Stream model, and design
algorithms for the maximal independent set, vertex-colouring, and 
planar graph single source shortest paths problems on those models.

\end{abstract}

\begin{acknowledgement}
\singlespace
First and foremost, I would like to thank my guide for all his support and
encouragement during the course of my PhD and Masters.
I am also grateful to him for always patiently listening to my ideas and doubts
even when they were trivial or ``silly".
Indeed, I have been inspired by his hard-working attitude,
intellectual orientation and a thoroughly professional outlook towards research.
The research training that I have acquired while working with him
will drive my efforts in future.

I would also like to thank the members of my doctoral committee, in particular
Profs. S. V. Rao, Pinaki Mitra and J. S. Sahambi,
for their feedback and encouragement
during the course of my PhD work. I am also grateful to the entire Computer
Science faculty for their tremendous support and affection during my stay at IIT
Guwahati.
I would also like to thank the anonymous referees who have commented at various fora,
on the I/O efficient minimum spanning trees algorithm presented in this thesis.
Their comments have been specifically valuable in improving the rigour and presentation 
of this piece of work.

I also gratefully acknowledge MHRD, Govt of India and Philips Research, India
for supporting my research at IIT Guwahati.

I take this opportunity to express my heartfelt thanks to my friends and
colleagues who made my stay at IIT
Guwahati an enjoyable experience. Knowing and interacting with friends such as
Godfrey, Lipika, Minaxi, Mili, and Thoi has been a great experience. These
friends have always been by my side whenever I have needed them.

It goes without saying that this journey would not have been possible without
the tremendous support and encouragement I have got from my
sister Chutti and my bhaiya Lucky,  and my
parents. Indeed no words can express my gratefulness towards my parents
who have unconditionally and whole-heartedly supported me in all my endeavours.
I am also grateful to my in-laws for their understanding, patience and
tremendous support. I am also thankful to my sister-in-law Tinku for being a great
friend and her understanding during difficult times.

Last but not the least, I thank my husband Mani for his constant
encouragement, love, support and infinite patience. Without him,
this journey could not have been completed so smoothly.

\vspace*{25mm}
{\Large Date: \rule{4cm}{1sp}\hfill Alka \rule{1cm}{0pt}}

\end{acknowledgement}

\tableofcontents
\listoftables
\listoffigures




\cleardoublepage\pagenumbering{arabic} 

\chapter{Introduction}
\label{intro:chapt}

Over the years, computers have been used to solve larger and larger
problems. Today we have several applications of
computing that often deal with massive
data sets that are terabytes or even petabytes in size.
Examples of such applications can be found in
databases \cite{KaRaDa+96, RaSu94}, geographic information systems \cite{ArVeVi+07, HaCo91}, 
VLSI verification, computerised medical treatment, astrophysics, geophysics, 
constraint logic programming, computational biology, 
computer graphics, virtual reality, 3D simulation and modeling
\cite{Arthesis}, 
analysis of telephone calls in a voice network \cite{BuGoWe03, CoFiPr+00, FKM+05a, HuDaMa00}, 
and transactions in a credit card network \cite{ChFaPr99}, to name a few.

\section{Memory Models}

In traditional algorithm design,
it is assumed that the main memory is infinite in size and
allows random uniform access to all its locations.
This enables the designer to assume that all the data
fits in the main memory. (Thus, traditional algorithms
are often called ``in-core''.) Under these assumptions,
the performance of an algorithm is decided by the number of
instructions executed, and therefore, the design goal is to
optimise it.

These assumptions may not be valid while dealing with massive data sets,
because in reality, the main memory is limited, and so the bulk of the
data may have to be stored in inexpensive but slow secondary memory.
The number of instructions executed would no longer be a reliable
performance metric; but a measure of the time taken in input/output (I/O) communication would be.
I/O communication tends to be slow because of large access times of 
	secondary memories.

Computer hardware has seen significant advances in the last few decades:
machines have become a lot faster, and the amount of main memory they 
have has grown.
But the issue of the main memory being limited has only become
more relevant, because applications have grown even faster in size.
Also, small computing devices (e.g., sensors, smart phones)
with limited memories have found several uses.

Caching and prefetching methods are typically designed to be 
	general-purpose, and cannot take full advantage of the
	locality present in the computation of a problem. 
Designing of I/O efficient algorithms has, therefore, been an actively researched
area in the last twenty years. A host of algorithms have been designed with
an intent of minimising the time taken in I/O.  
For many a problem,
it has been shown that while the traditionally efficient algorithm is expensive on I/Os,
novel and very different design techniques can be used to produce an
algorithm that is not. This thesis makes contributions of a similar
vein.

We now describe some models of computations that have been used to design
I/O efficient algorithms.

\subsection{The External Memory Model}
 
In this memory model, introduced by Aggarwal and Vitter
\cite{AgVi88}, it is assumed that the bulk of the data is kept in the
secondary memory which is a permanent storage. The
secondary memory is divided into blocks. An input/output (I/O) is defined as 
	the transfer of a block of data between the secondary memory and a volatile main memory.
The processor's clock period and the main memory access time are
negligible when compared to the secondary memory access time.
The measure of the performance of an algorithm is the number of I/Os it performs.
Algorithms designed on this model are referred to as external memory algorithms.
The model defines the following parameters: the size of the problem input ($N$),
the size of the main memory ($M$), and the size of a disk block ($B$).

It has been shown that, on this model, the number of I/Os needed to read (write)
$N$ contiguous items from (to) the disk is $\scan{N} = \Theta (N/B)$, and that
the number of I/Os required to sort $N$ items is $\sort{N} = \Theta((N/B)\log_{M/B} (N/B))$ \cite{AgVi88}.
For all realistic values of $N$, $B$, and $M$, $\scan{N} < \sort{N} \ll N$.

\subsection{The Streaming Model}

This model \cite{AlMaSz99, HRR99, MP80} allows the input data to be accessed sequentially, but not randomly.
Algorithms on this model are constrained to access the data sequentially
in a few passes from the read-only input tape, using only a small amount
of working memory typically much smaller than the input size.
Algorithms designed on this
model are called streaming algorithms.

The streaming model defines the following parameters:
the size of the problem input ($N$), and the size of the working memory ($M\log N$).
At any one time $\Theta(M)$ elements can fit into the working
memory; $\Theta(\log N)$ bits are needed to represent each input element.
In a pass, the input is read sequentially by loading $O(M)$ elements
into the working memory at a time.
The measure of the performance of an algorithm on this model is the
number of passes it requires.

The read-only streaming model is very restrictive.
Several problems, such as graph problems, are inherently difficult on the model.
Therefore many extensions to this model have been proposed.
The W-Stream model \cite{R03} is one such extension.
In this model, an algorithm is allowed to write an intermediate stream as it reads the
input stream.
This intermediate stream, which can be a constant factor larger than the original
stream, can be used as input for the next pass.
Algorithms designed on this
model are called W-Stream algorithms.

\section{Graphs}

Working with massive data sets often require analysing massive graphs.
An example is the ``call graph" on a telephone network, where a vertex corresponds to a telephone
number and an edge to a call between two numbers during
some specified time interval. Another example is the ``web graph", where
vertices are web pages and edges are links between web pages \cite{FKM+05a}.
Moreover, many problems from widely varied domains can be reduced to graph problems.
One example is the shortest path problem on geometric domains which
is often solved by computing shortest paths in graphs, where
a vertex is a discrete location, an edge is a connection between two vertices,
and the weight of an edge is the geometric distance between the locations
that correspond to its end-vertices \cite{zeh}.

A graph $G=(V,E)$ consists of a finite nonempty set
$V$ of vertices and a set $E\subseteq V \times V$
    of unordered pairs, called edges, of distinct vertices \cite{H69}.
We shall denote $\vert V \vert$ and $\vert E \vert$ by $V$ and $E$ respectively.
Two vertices $u,v\in V$ are said to be adjacent to (or neighbours of) each
    other iff $(u,v)\in E$.
The neighbourhood of vertex $u$ (the set of $u$'s neighbours) is denoted by $N(u)$.
The number of neighbours of a vertex is its degree. In a directed
graph every edge is an ordered pair; if
    $(u,v)\in E$, we say $v$ is an out-neighbour of $u$; we will
    also be using the terms in-neighbour, out-degree and in-degree,
    which can be similarly defined.

Most of the problems that we consider in this thesis are on graphs.

\section{Background}

In this section, we discuss some of the known upper and lower bound results that are relevant to 
our work, and are on the external memory model and variants of the streaming model.

\subsection{On the External Memory Model}

Detailed surveys of algorithmic results on the external memory model
can be found in \cite{MSS03, V01, Vi08}.
Here we mention some results relevant to us. 

The number of I/Os needed to read (write) $N$ contiguous items from (to) the disk 
	is $\scan{N}=\Theta (N/B)$, and that the number of I/Os required to sort 
	$N$ items is $\sort{N}=\Theta((N/B)\log_{M/B} (N/B))$ \cite{AgVi88}.

The list ranking and $3$ colouring problems on $N$ node linked lists 
	have I/O complexities that are $O(\sort{N})$ \cite{ChGoGr+95}.
For many problems on trees (e.g., Euler Tour, expression tree evaluation)
the I/O complexity is $O(\sort{N})$ \cite{ChGoGr+95}.
For list ranking and many tree problems the lower bound is $\Omega(\perm{N})$ on
	I/Os, where $\perm{N}$ is the number of I/Os required to permute $N$ elements. 

The connected components problem requires $\Omega((E/V)\sort{V})$ I/Os \cite{MR99}.
This lower bound extends also to related problems like minimum spanning tree,
biconnected components and ear decompositions \cite{MR99}.
The best known upper bound on I/Os for all these problems is
	$O(\frac{E}{V}\sort{V}\log \log \frac{VB}{E})$ \cite{ArBrTo04, MR99}.

Results on breadth first search can be found in \cite{MeMe02, MR99}.
The best known upper bound is $O(\sqrt{(VE)/B} + \sort{E} + \ssf{V,E})$ on I/Os
	\cite{MeMe02}, where $\ssf{V,E}$ is the number of I/Os required to 
	compute the spanning forests of a graph $G = (V,E)$.
Depth first search on an undirected graph can be computed 
	in $O(V \log V + E/B \log (E/B))$ I/Os \cite{KS96}.

Some restricted classes of sparse graphs, for example 
planar graphs, outerplanar graphs, grid graphs and bounded tree width graphs, have
been considered in designing I/O efficient 
algorithms for single source shortest paths, depth first search, and breadth first search. 
Exploitation of the structural properties of these classes of sparse graphs has led to algorithms 
for them that
	perform faster than the algorithms for a general graphs.
Most of these algorithms require $O(\sort{V+E})$ I/Os \cite{MSS03, V01, Vi08}.  

\subsection{On the Streaming Model and Its Variants}

Many data sketching and statistics problems have approximate solutions
	that execute in $O(1)$ passes and use polylogarithmic working space 
on the streaming model. 
See \cite{Mu05} for a survey.
Graph problems are hard to solve on the streaming model.
Therefore, streaming solutions for them have tended to be on 
	variants of the streaming model such as the semi-streaming  
\cite{FKM+05a}, 
	W-Stream \cite{DFR06, R03}, 
	and streaming and sorting models \cite{AgDaRa+04, R03}.  
Most of the graph problems have a lower bound of $\Omega(N/(M \log N))$ on passes 
	on the W-Stream model \cite{DEM+07, DFR06, HRR99, R03}.
Connected components and minimum spanning trees, each can be computed in $O(N/M)$ passes 
	\cite{DFR06}.
With high probability, a maximal independent set can be computed in $O((N \log N)/M)$ passes \cite{DEM+07}.
For the single source shortest paths problem, there is a monte-carlo algorithm
	that executes in $O((C N \log N)/ \sqrt{M})$ passes \cite{DFR06}, where
	$C$ is the maximum weight of an edge.

\section{Summary of the Thesis}

The thesis is divided into two parts.
In the first part are presented some lower and upper bounds results,
and data structures on the external memory model.
The second part of the thesis deals with the W-Stream model and its variants.

\subsection{Part I: Algorithms and Data Structures on the External Memory Model}
\label{EMGA:DS}
\subsubsection{Minimum Spanning Trees}

The minimum spanning tree (MST) problem on
an input undirected graph $G = (V,E)$,
where each edge is assigned a
real-valued weight,
is to compute a spanning forest (a spanning tree for each connected 
	component) of $G$
so that the total weight of the edges in the spanning forest
is a minimum.
We assume that the edge weights are unique. This assumption is without
loss of generality, because if the edge weights are not unique, then
tagging each edge weight with the label of the corresponding edge will
do to make all weights unique. One scan of the edgelist is sufficient for
such a tagging. 

For this problem, a lower bound of $\Omega(\frac{E}{V} \sort{V})$ on
I/Os is known \cite{MR99}. We present an I/O efficient
algorithm that computes a minimum spanning tree of an
undirected graph $G = (V,E)$ in $O(\sort{E}$ $\log \log_{E/V} B)$ I/Os. 
The current best known upper bound on I/Os for this
problem is $O(\sort{E}$ $\log\log (VB/E))$ \cite{ArBrTo04}. Our
algorithm performs better than that for practically all values of
$V$, $E$ and $B$, when $B \gg 16$ and $(B)^{\frac{1 - \sqrt{1 - \frac{4}{\log B}}}{2}} \leq E/V \leq (B)^{\frac{1 + \sqrt{1 - \frac{4}{\log B}}}{2}}$. Our Algorithm matches the lowerbound when $E/V \geq B^{\epsilon}$ for a constant $\epsilon > 0$. In particular, when
$E/V=B^{\epsilon}$, for a constant $0 < \epsilon < 1$, our algorithm, in
addition to matching the lower bound, is asymptotically faster
than the one by Arge et\ al.\ \cite{ArBrTo04} by a factor of $\log\log B$.

Graph problems, because of their inherent lack of data localities,
are not very amenable to efficient external memory solutions.
Therefore, even a modest $\log\log B$ factor of improvement is
significant.

In addition to computing a minimum spanning tree, our algorithm
also computes connected components. Therefore, our algorithm
improves on the best known upper bound for the connected
components problem too.

\subsubsection{External Memory Soft Heap, and Hard Heap, a Meldable Priority Queue}
A priority queue is a data structure that allows {\tt Insert}, {\tt Findmin}, and {\tt Deletemin} operations to execute efficiently. External
memory priority queues that perform each of these operations in
$O((1/B) \log_{M/B}(N/B))$ amortized I/Os are known \cite{Arge03, KS96}.

We present an external memory version of soft heap \cite{Ch00a}
	that we call ``External Memory Soft Heap'' (EMSH for short). 
It supports {\tt Insert}, {\tt Findmin}, {\tt Deletemin} and {\tt Meld} operations.
An EMSH may, as in its in-core version, and at its discretion,
corrupt the keys of some elements in it, by revising them upwards.
But the EMSH guarantees that the number of corrupt elements in it is never more than
$\epsilon N$, where $N$ is the total number of items inserted in it, and $\epsilon$ is
a parameter of it called the error-rate.
The amortised I/O complexity of an {\tt Insert} is 
	$O(\frac{1}{B} \log_{M/B}\frac{1}{\epsilon})$.
{\tt Findmin}, {\tt Deletemin} and {\tt Meld} all have non-positive amortised I/O complexities.

This data structure is useful for finding exact and approximate medians, and 
	for approximate sorting the same way it is in its in-core version \cite{Ch00a}.
Each can be computed in $O(N/B)$ I/Os.

When we choose an error rate $\epsilon<1/N$, EMSH stays devoid of corrupt nodes, 
and thus becomes a meldable priority queue that we call ``hard heap''.  
The amortised I/O complexity of an {\tt Insert}, in this case, is $O(\frac{1}{B} \log_{M/B}\frac{N
}{B})$, over a sequence of operations involving $N$ {\tt insert}s.
{\tt Findmin}, {\tt Deletemin} and {\tt Meld} all have non-positive amortised I/O complexities.
If the inserted keys are all unique, a {\tt Delete} (by key) operation can
        also be performed at an amortised I/O complexity of $O(\frac{1}{B} \log_{M/B}\frac{N}{B})$.
A balancing operation performed once in a while on a hard heap 
ensures that the number of I/Os performed by a sequence of $S$ operations on it is
$O(\frac{S}{B}+\frac{1}{B}\sum_{i = 1}^{S} \log_{M/B}\frac{N_i}{B})$, where
$N_i$ is the number of elements in the heap before the $i$th operation.

\subsubsection{The Minimum Cut Problem}
The minimum cut problem on an undirected unweighted graph is to partition the vertices into two
        sets while minimising the number of edges from one side of the partition to the other.
It is an important combinatorial optimisation problem.
Efficient in-core and parallel algorithms for the problem are known.
For a recent survey see \cite{Br07, Ka00, KaMo97, NaIb08}.
This problem has not been explored much from the perspective of massive data sets.
However, it is shown in \cite{AgDaRa+04, R03} that a minimum cut can be computed in
        a polylogarithmic number of passes using only a polylogarithmic sized main
        memory on the streaming and sorting model.

We show that any minimum cut algorithm requires
	$\Omega(\frac{E}{V} \sort{V})$ I/Os.
A minimum cut algorithm is proposed that runs in $O(c  (\msf{V,E}\log E + \frac{V}{B} \sort{V}))$ I/Os, and
        performs better on dense graphs than the algorithm of \cite{Ga95}, which requires $O(E + c^2 V \log(V/c))$
        I/Os, where $\msf{V,E}$ is the number of I/Os required to compute a minimum 
	spanning tree, and $c$ is the value of minimum cut.
Furthermore, we  use this algorithm to construct
        a data structure that stores all $\alpha$-mincuts (which are cuts of size at most $\alpha$ times the
        size of the minimum cut), where $\alpha < 3/2$.
The construction of the data structure requires an additional $O(\sort{k})$ I/Os, where $k$ is
        the total number of $\alpha$-mincuts.
The data structure answers a query of the following form in $O(V/B)$ I/Os: A cut $X$ (defined by a vertex
        partition) is given; find whether $X$ is an $\alpha$-mincut.

Next, we show that the minimum cut problem can be computed with high
probability in $O(c \cdot \msf{V,E} \log E  + \sort{E} \log^2 V +  \frac{V}{B} \sort{V} \log V)$ I/Os.
We also present a $(2 + \epsilon)$-approximate minimum cut
algorithm that requires $O((E/V) \msf{V,E})$ I/Os and 
        performs better on sparse graphs than both our exact minimum cut algorithm, and
	the in-core algorithm of \cite{Ga95} run on the external
	memory model in to-to.

\subsubsection{Some Lower and Upper bound results on Interval Graphs}
A graph $G=(V,E)$ is called an interval graph,
    if for some set $\Im$ of intervals of a linearly ordered set,
    there is a bijection $f:V\rightarrow\Im$
    so that two vertices $u$ and $v$ are adjacent in $G$ iff $f(u)$ and
    $f(v)$ overlap.
Every interval graph has an interval representation in which endpoints
    are all distinct \cite{G80}.

We show that
finding the connected components in a collection of disjoint
monotonic doubly linked lists (MLCC) of size $V$ is equivalent to 
the problem (IGC) of minimally colouring an interval
graph whose interval representation is
	given.
The number of I/Os needed for both are shown to be $\Omega(\frac{V}{B} \log_{M/B} \frac{\chi}{B})$, 
where $\chi$ is the chromatic
    number of an interval graph, or the total number of disjoint monotonic doubly linked lists, as is relevant.
We also show that the $3$ colouring of a doubly linked list ($3$LC) of size $V$ is
reducible to the $2$ colouring of a set of disjoint monotonic doubly linked lists (2MLC)  in 
$O(\scan{V} + \sort{\chi})$ I/Os.
It is also shown that 2MLC and 3LC of sizes $V$ each have lower bounds of 
$\Omega(\frac{V}{B} \log_{M/B} \frac{\chi}{B})$ on
I/Os, where $\chi$ is the number of disjoint 
	monotonic  doubly linked lists, and the total number of forward and backward stretches
	in the doubly linked list respectively.

We
present an SSSP algorithm that requires $O(\sort{V})$
    I/Os, and a BFS tree computation algorithm that requires $O(\scan{V})$ I/Os, and
	a DFS tree computation algorithm that requires $O(\frac{V}{\chi} \sort{\chi})$ I/Os.
The input graph is assumed to be represented as a set of intervals in sorted
	order.
We show that IGC can be computed in an optimal
$O(\frac{V}{B}\log_{M/B}\frac{\chi}{B})$ I/Os, if
the input graph is represented as a set of intervals in sorted
	order. 
Optimal algorithms are given for 3LC, 2MLC, MLCC problems.

\subsection{Part II: Algorithms on the W-Stream Model and its Variants}
\label{WSA}
\subsubsection{Some Algorithms on the W-Stream Model}
The following results are presented.
\begin{itemize}
\item Lower bounds of $\Omega(N/(M \log N))$ on passes for list ranking and maximal matching.
A lower bound for list ranking also applies to expression tree evaluation, finding the depth of
every node of a tree, and finding the number of descendants of every node in a tree.

\item An algorithm that sorts $N$ elements in $O(N /M)$ passes
while performing $O(N\log M+N^2/M)$ comparisons and $O(N^2/M)$ elemental reads.
Our algorithm does not use a simulation, and is easier to implement than the
earlier algorithms.

\item Algorithms for list ranking, and tree problems
such as Euler Tour, rooting of trees, labelling of rooted trees and
expression tree evaluation that use $O(N /M)$ passes each.
Unlike the previous algorithms, our algorithms are easy to implement as they do not use
simulations.

\item Algorithms for finding a maximal independent set and a $\Delta + 1$ colouring of graphs.
We show that when the input graph is presented in an adjacency list representation,
each can be found deterministically in $O(V / M)$ passes.
We also show that when the input is presented as an unordered edge list,
each can be found deterministically in
$O(V /x)$ passes, where
$x = O(\min\{M, \sqrt{M\log V} \})$ for MIS, and
$x = O(\min\{M, \sqrt{M\log V}, \frac{M\log V}{\Delta\log\Delta} \})$
for $\Delta + 1$ colouring.
\item  Algorithms for maximal matching and $2$-approximate weighted vertex cover that are
deterministic and require $O(V / M )$ passes.  The vertex cover algorithm assumes that
the weight of each vertex is $V^{O(1)}$.
The input here is assumed to be an unordered edge list.
The lower bound for the maximal matching problem is shown to be $\Omega(V/(M \log V))$ on passes.

\item An algorithm that, for all vertices $v \in V$, computes with high probability
	 an $\epsilon$-approximate shortest path from a given source vertex $s$ to $v$ 
	in $O(\frac{V \log V \log W}{\sqrt{M}})$ passes, where $W$ is the sum of the weights of the edges.
We assume that $\log W = O(\log V)$.
If $C$ is the maximum weight of an edge, then $W \leq VC$, and
our algorithm improves on the previous bound by a factor of $C/\log({VC})$ at the cost a
small error in accuracy.
Here again, we assume the input to be  given as an unordered edge list.

\end{itemize}
\subsubsection{Two Variants of the W-Stream Model and Some Algorithms on Them}
We propose two models which are variants of the
        W-Stream model. We give
        the following algorithms that run on two of those models: an
        $O(V/M)$ passes maximal independent set algorithm and an
        $O(V/x)$ passes $(\Delta+1)$-colouring algorithm,where $x = O(\min \{M, \sqrt{M\log V}\})$, both for general graphs,
        and an $O((\sqrt{V}+ \frac{V}{M})\log V+\frac{V}{\sqrt{M}})$ passes single source shortest paths algorithm 
        and an $O(\frac{V^2}{M})$ passes all pairs shortest paths algorithm, both
for planar graphs.

\section{Organisation of the Thesis}
The rest of the thesis is organised as follows. In chapters~\ref{mst:chapt}-\ref{ig:chapt} we present
    our results on external memory graph algorithms and data structures.
Chapter~\ref{mst:chapt} describes the improved result on minimum spanning tree problem.
Chapter~\ref{emsh:chapt} describes external memory soft heap and hard heap. 
Chapter~\ref{cut:chapt} describes the results on the minimum cut problem.
Chapter~\ref{ig:chapt} presents the lower bound results on IGC, $3$LC, $2$MLC, and MLCC.
It also presents upper bound results for $3$LC, IGC, SSSP, BFS and DFS.

In chapters~\ref{wstream:chapt}, and \ref{modstream:chapt} we present some lower and upper bounds 
	on the W-Stream model and its variants. 
Chapter~\ref{wstream:chapt} describes our results on sorting, tree problems, maximal
    independent set, graph colouring, SSSP, maximal matching, and weighted vertex cover,
all on the W-Stream model. 
Chapter~\ref{modstream:chapt} introduces two new variants of the W-Stream model, and then
presents algorithms for the maximal independent set, $\Delta+1$-colouring (both on general graphs), 
	SSSP and APSP (both on planar graphs) problems, on those models.
\part{Algorithms and Data Structures on the External Memory Model}
   \chapter{Minimum Spanning Trees}
\label{mst:chapt}
\section{Introduction}
The computing of minimum spanning trees is a fundamental problem in
Graph Theory and has been studied on various models of
computation. 
For this problem, a lower bound of $\Omega(\frac{E}{V}  \sort{V})$ on I/Os is known. 
In this chapter, we give an I/O efficient algorithm that computes a minimum spanning
tree of an undirected graph $G = (V,E)$ in $O(\sort{E}$ $\log \log_{E/V} B)$ I/Os. 
The current best known upper bound on I/Os for this problem 
	is $O(\sort{E}$ $ \log\log (VB/E))$ \cite{ArBrTo04}. 
Our algorithm performs better than that for practically all values of $V$, $E$ 
	and $B$, when $B\gg 16$ and $(B)^{\frac{1 - \sqrt{1 - \frac{4}{\log B}}}{2}} \leq E/V \leq (B)^{\frac{1 + \sqrt{1 - \frac{4}{\log B}}}{2}}$. 
Our Algorithm matches the lowerbound when $E/V \geq B^{\epsilon}$ for a constant $\epsilon > 0$. In particular, when
$E/V=B^{\epsilon}$, for a constant $0 < \epsilon < 1$, our algorithm, in
addition to matching the lower bound, is asymptotically faster
than the one by Arge et al. by a factor of log log B.

Graph problems, because of their inherent lack of data localities,
are not very amenable to efficient external memory solutions.
Therefore, even a modest $\log\log B$ factor of improvement
is significant.

In addition to computing a minimum spanning tree, our algorithm
computes connected components also. Therefore, our algorithm
improves on the best known upper bound for the connected components problem too.

\subsection{Problem Definition}
The \emph{minimum spanning tree (MST) problem} on
an input undirected graph $G = (V,E)$,
where $V$ and $E$ are the vertex set and edge set
respectively, and each edge is assigned a
real-valued weight,
is to compute a spanning forest (a spanning tree for each connected component) of $G$
so that the total weight of the edges in the spanning forest
is a minimum.

\subsection{Previous Results}
For the MST problem, a lower bound of
$\Omega(\frac{E}{V} \sort{V})$ on the number of I/Os is known
\cite{MR99}, while the currently best known algorithm, by Arge et\
al.\ \cite{ArBrTo04}, executes in $O(\sort{E}$ $\log \log (VB/E))$
I/Os. This algorithm matches the lower bound only when $E = \Omega(VB)$.
The problems of computing connected components (CC) and minimum
spanning forests are related. The best known upper bound for the
connected components problem \cite{MR99} is also $O(\sort{E} \log \log (VB/E))$. Both the connected components and minimum
spanning tree problems can be solved by a randomised algorithm in
$(\frac{E}{V} \sort{V})$ I/Os \cite{AbBuWe02, ChGoGr+95}. Some of the
best known upper bounds for the MST and CC problems are summarised
in the Table~\ref{prev:result:mst}.
\begin{table}
\begin{center}

            \begin{tabular}{|l|l|l|}
               \cline{1-3}
               \hline
              {\bf Problem} & {\bf I/O complexity} & {\bf Reference} \\
               \hline \hline
               MST and CC & $O(\sort{E} \log (V/M))$ & \cite{ChGoGr+95}  \\

               & $O(\sort{E} \log B + \scan{E} \log{V})$ &  \cite{KS96}  \\
              
               & $O(V + \sort{E})$ & \cite{ArBrTo04} \\
               
               & $O(\sort{E} \log \log (VB/E))$ &  \cite{ArBrTo04} \cite{MR99}  \\
               
               & $O(\frac{E}{V} \sort{E})$ (randomised) & \cite{AbBuWe02} \cite{ChGoGr+95} \\
               \hline

            \end{tabular}
\caption{Previous upper bounds for the MST and CC problems}
\label{prev:result:mst}
\end{center}
       \end{table}

\subsection{Our Results}
We propose a new MSF algorithm for undirected graphs. The I/O
complexity of our algorithm is $O(\sort{E}  \log \log_{E/V}B
)$. The complexity can be improved to $O(\frac{E}{V} \sort{V}$ $\log \log_{E/V} B)$,
    if a sparsification technique \cite{EGI+97} is applied to our algorithm \cite{Vi08}.
Therefore, our algorithm matches the lower bound for the
problem when $E \geq VB^{\epsilon}$, where $\epsilon>0$ is some constant.
The function $\log\log(VB/E)$ is $O(\log\log_{E/V}B)$ only if
$\log B - \log(E/V)<(\log B)/(\log(E/V))$, which is only if
$\log(E/V)>(\log\sqrt{B})(1+\sqrt{1-4/(\log B)})$ or
$\log(E/V)$ $< (\log\sqrt{B})$$(1-\sqrt{1-4/(\log B)})$.
If $B\gg 16$, for practically all values of $V$, $E$ and $B$ such that $(B)^{\frac{1 - \sqrt{1 - \frac{4}{\log B}}}{2}} \leq E/V \leq (B)^{\frac{1 + \sqrt{1 - \frac{4}{\log B}}}{2}}$,
$\log\log(VB/E)>\log\log_{E/V}B$, and therefore
our algorithm performs better than that of \cite{ArBrTo04}.
Our Algorithm matches the lowerbound when $E/V \geq B^{\epsilon}$ for a constant $\epsilon > 0$. In particular, when
$E/V=B^{\epsilon}$, for a constant $0 < \epsilon < 1$, our algorithm, in
addition to matching the lower bound, is asymptotically faster
than the one by Arge et\ al.\ \cite{ArBrTo04} by a factor of $\log\log B$.

In addition to computing a minimum spanning tree, our algorithm
computes connected components also.
Therefore, our algorithm
improves the upper bound for the connected components problem too
\cite{MR99}.

\subsection{Organisation of This Chapter}
Section~\ref{sec:prem:mst} is a preliminary section in which we discuss some
        existing algorithms whose idea will be used later in our algorithm.
In Section~\ref{mst:algo}, we describe our algorithm. 

\section{Preliminaries}
\label{sec:prem:mst}
The MSF algorithms that are based on edge contraction typically
proceed in a number of Bor\r{u}vka phases \cite{B26}. In each phase the
lightest edge adjacent to each vertex $v$ is selected and output
as part of the MSF. Then the selected edges are contracted; that
is, each set of vertices connected by the selected edges is fused
into a supervertex. Proofs of correctness of this approach can be
found in \cite{B26,ChGoGr+95,CLC82, CV91, KS96,MR99}.

Let the size of a
supervertex be the number of vertices it contains from the
original graph. An edge in the original graph between
two vertices that belong to the same supervertex is an
internal edge of the supervertex. Edges $(u,v)$ and $(u',v')$, where
$u$ and $u'$ end up in the same supervertex,
and
so do $v$ and $v'$, become multiple edges. At the end of
a Bor\r{u}vka phase, the algorithms typically remove the internal edges
and for each set of multiple edges retain only one of the lightest.

After the $i$-th phase, the size of every supervertex is at least
$2^i$, and thus after $O(\log (V/M))$ phases, the  vertices in the
contracted graph fit in the main memory. Once the vertices can fit
into the main memory, the MSF can be computed in one scan of the
sorted edge set using the disjoint set data structure and
Kruskal's algorithm \cite{CLR90}. Each Bor\r{u}vka phase can be
performed in $O(\sort{E})$ I/Os \cite{ArBrTo04, ChGoGr+95, KS96, MR99}; this
results in an $O(\sort{E} \log (V/M))$ algorithm. Kumar and
Schwabe \cite{KS96} improved on this when they obtained an
$O(\sort{E}  \log B + \scan{E}  \log V)$ I/Os algorithm;
they use the fact that after $\Theta (\log B)$ phases, with the
number of vertices decreased to $O(V/B)$, a contraction phase can
be performed more efficiently. Recently, Arge et\ al.\
\cite{ArBrTo04} obtained an improved
    $O(\sort{E}  \log \log (VB/E))$ I/Os algorithm.
They use the fact
    that after $\Theta(\log (VB/E))$ phases, with number of vertices decreased to $E/B$, a
modified version of Prim's internal memory algorithm \cite{CLR90}
can be used to construct an MSF in the remaining graph.

This modified version of Prim's algorithm is also given in
\cite{ArBrTo04}, and works as follows: An external memory priority
queue (EMPQ) is initialised with the edges that are incident to
the source vertex. In each step, as in Prim's algorithm, we add to
the present MSF the lightest {\em border edge} (an edge that
connects a vertex in the present MSF to one that is not), and add
to the EMPQ all the edges incident to the newly captured vertex.
The EMPQ stores all the current border edges and some internal
edges (edges between vertices in the present MSF). A
delete-minimum operation on the EMPQ produces the lightest edge in
it; it is an internal edge iff there are two copies of it in the
EMPQ; discard it if it is an internal edge. The algorithm performs
$\Theta(E)$ EMPQ operations which can be performed in
$O(\sort{E})$ I/Os \cite{Arge03, KS96,FJKT99} and also needs one
I/O per vertex. Thus, its I/O complexity is $O(V+\sort{E})$. When
$V=E/B$, this is $O(\sort{E})$. Note that the algorithm can also
be used for computing connected  components
    when an edge $(u,v)$ is inserted in EMPQ with key value $k = \min \{u,v\}$.

The algorithm of Arge et\ al.\  \cite{ArBrTo04} performs $\Theta(\log
(VB/E))$ phases in a total of $O(\sort{E}$ $ \log \log (VB/E))$
I/Os. They divide the $\Theta(\log (VB/E))$ phases into $\Theta(
\log \log (VB/E))$ \emph{superphases} requiring $O(\sort{E})$ I/Os
each and obtain the following the result:

{\em The minimum spanning tree of an undirected weighted graph $G
= (V,E)$ can be reduced to the same problem on a graph with
$O(E/B)$ vertices and $O(E)$ edges in $O(\sort{E}$ $ \log \log
(VB/E))$ I/Os. }

The $i$-th superphase of their algorithm consists of $\lceil \log
\sqrt N_i \rceil$ phases, where $N_i = 2^{(3/2)^i}$. To be
efficient, the phases in superphase $i$ work only on a subset
$E_i$ of edges. This subset contains the $\lceil \sqrt N_i \rceil$
lightest edges incident with each vertex $v$. These edges are
sufficient to perform $\lceil \log \sqrt N_i \rceil$ phases as
proved in \cite{ArBrTo04,CHIL03, MR99}. Using this subset of edges,
each superphase is performed in $O(\sort{E_i})=O(\sort{E})$ I/Os.

Combining the above with the modified Prim's algorithm,
\cite{ArBrTo04} presents an $O(\sort{E}$ $\log \log (VB/E))$ minimum spanning tree algorithm.

Our algorithm changes the scheduling in the initial part of the
algorithm of Arge et\ al.\ \cite{ArBrTo04}. Some of ideas
from Chong et\ al.\ \cite{CHIL03} too have been used.

\section{Our Algorithm}
\label{mst:algo}
The structure of our MSF algorithm is similar to that of
\cite{ArBrTo04}. The first part of the algorithm reduces the number
of vertices from $V$ to $E/B$, while the second part applies the I/O
efficient version of Prim's algorithm (see the previous section
for a description) to the resultant graph in
$O((E/B) + \sort{E})= O(\sort{E})$ I/Os.

The first part of our algorithm differs from that of \cite{ArBrTo04},
but the second part of the two algorithms are identical.
While we also make use of $\log (VB/E)$ phases,
we schedule them quite differently, thereby achieving a different
I/O complexity that is an improvement over the one of \cite{ArBrTo04}
for most values of $V$, $E$ and $B$.
The rest of this section describes our scheduling of the
phases, each of which reduces the number of vertices by a factor of at
least two, for an overall reduction by a factor of at least $VB/E$:
from $V$ to at most $E/B$.

We assume that the edge weights are unique. This assumption is without
loss of generality, because if the edge weights are not unique, then
tagging each edge weight with the label of the corresponding edge will
do to make all weights unique. One scan of the edgelist is sufficient for
such a tagging.
 
We assume that $E>V$ in the input graph. This assumption is
without loss of generality, because if $E<V$ in $G$, we could add
a vertex to the graph and make it adjacent to every other vertex
with the weights of all the new edges set to $\infty$. This can be
done in $O(\scan{V})$ I/Os. An MST $T$ of the resultant graph
$G'=(V',E')$ can be converted into an MSF for the original graph
by deleting all the new edges in $T$. Since $V'=V+1$ and $E'=E+V$,
we have that $E'>V'$, if $E>1$.

If $E<V$ in $G$, alternatively, we could run Bor\r{u}vka phases
until the number of edges in the resultant graph becomes zero or
exceeds the number of vertices. (The latter is possible if
isolated supervertices are removed from the graph as soon as they
form.) The first of those phases would run in $O(\sort{V})$ I/Os.
(See \cite{ArBrTo04}. Also see the discussion below on the first two
steps of a phase (Hook and Contract) of our algorithm.) The I/O
cost of the subsequent
    phases will fall geometrically, as the number of vertices would at least
halve in each phase. Thus, the total I/O cost too will be
$O(\sort{V})$.

As mentioned before, our algorithm executes $\log(VB/E)$
Bor\r{u}vka phases to reduce the number of vertices from $V$ to $E/B$.
These phases are executed in a number of stages.
The $j$-th stage, $j\geq 0$, executes $2^j \log(E/V)$ Bor\r{u}vka phases.
(Note that $\log(E/V)>0$.)

Therefore, the number of Bor\r{u}vka phases executed prior to the
$j$-th stage is $(2^j-1)\log(E/V)$, and the number of
supervertices at the beginning of the $j$-th stage is at most
$V/2^{(2^j-1)\log(E/V)}=E/(E/V)^{2^j}$.  Thus, $\log \log_{E/V} B$
stages are required to reduce the number of vertices to $E/B$. We
shall show that each stage can be executed in $O(\sort{E})$ I/Os.
Thus, the algorithm would compute the MSF in $O(\sort{E} \log \log_{E/V} B)$ I/Os.

For $j>0$, the $j$-th stage takes as input the graph $G_j=(V_j,E_j)$ output by the
previous stage. $G_{0}=(V_{0},E_{0})$, the input to the $0$-th stage, is the input graph
$G=(V,E)$. Let $g(j)=\log \log (E/V)^{2^j}$.
Thus the $j$-th stage has $2^{g(j)}$ phases.
From $E_j$ we construct $g(j)+2$ buckets that are numbered from $0$ to $g(j)+1$.
Each bucket is a set of edges, and is maintained as a sorted array with
the composite $\langle \mbox{source vertex}, \mbox{edgeweight} \rangle$ as the sort key.
In bucket $B_k$, $0 \leq k \leq
g(j)$, we store, for each vertex $v \in V_j$, the $(2^{2^k}-1)$ lightest
edges incident with $v$. Clearly, $B_{k}$ is a subset of $B_{k+1}$.
Bucket $B_{g(j)+1}$ holds all the edges of $G_j$.

For $0 \leq k \leq g(j)$,
Bucket $B_k$
is of size $\leq V_j (2^{2^k} - 1)$.
Since many of the vertices might be of degrees smaller than $(2^{2^k} - 1)$, the
actual size of the bucket could be much smaller than
$V_j (2^{2^k} - 1)$; but this is a gross upper bound.
The total space used by all the buckets of the $j$-th stage is, therefore, at most
\[\left(\sum^{g(j)}_{k=0} |B_k| \right) +
    |B_{g(j)+ 1}|
\leq \left(\sum^{g(j)}_{k=0} V_j (2^{2^k} - 1)\right) + O(E_j)\]
\[ \
= O(V_j (E/V)^{2^j}) + O(E_j) = O(E)
\]
because the set of supervertices $V_j$ number at most
$E/(E/V)^{2^j}$, as argued earlier.

Bucket $B_{g(j)+1}$ can be formed by sorting the list of edges on
the composite key $\langle \mbox{source vertex}, \mbox{edgeweight}
\rangle$. This requires $O(\sort{E_j})$ I/Os. (Note that for each
edge $\{u,v\}$, the list contains two entries $(u,v)$ and $(v,u)$
sourced at $u$ and $v$ respectively.) Next we form buckets
$B_{g(j)},\ldots,B_{0}$ in that order. For $g(j)\geq k\geq 0$,
bucket $B_{k}$ can be formed by scanning $B_{k+1}$ and choosing
for each vertex $v \in V_j$, the $(2^{2^k}-1)$ lightest
edges incident with $v$. Clearly, this involves scanning each
bucket twice, once for write and once for read. We do not attempt
to align bucket and block boundaries. As soon as the last record
of bucket $B_{k+1}$ is written, we start the writing of $B_k$; but
this requires us to start reading from the beginning of $B_{k+1}$;
if we retain a copy of the block that contains the beginning of
$B_{k+1}$ in the main memory, we can do this without performing an
additional I/O; thus the total I/O cost is proportional to the
total space (in blocks) needed for all the buckets of the $j$-th
stage, which is $O(E_{j}/B)$. By choosing not to align bucket and
block boundaries, we save the one-per-bucket overhead on I/Os we
would have incurred otherwise.

In the $j$-th stage, after constructing the buckets from $E_j$, the algorithm
performs $2^j \log (E/V)$ Bor\r{u}vka phases. These phases in our algorithm include,
in addition to the hook and contract operations, the clean up of
an appropriate bucket and some bookkeeping operations.
For $2\leq i \leq 2^j \log (E/V)$,
let $G_{j,i}=(V_{j,i},E_{j,i})$ be the graph
output by the $(i-1)$-st phase; this is the input
for the $i$-th phase.
Let $G_{j,1}=(V_{j,1},E_{j,1})= G_j = (V_j, E_j)$ be the input for the first phase.

A brief description of the $j$-th stage follows.
A detailed description is given later.

\vspace{0.2in}
\hrule

\begin{itemize}
\item Construct buckets $B_{g(j)+1}\ldots B_{0}$ as described earlier.

\item for $i=1$ to $2^j \log (E/V) -1 $,
\begin{enumerate}
\item (Hook) For each vertex $u \in V_{j,i}$, if $B_{0}$ contains an edge
incident with $u$
(there can be at most one, because $2^{2^{0}}-1=1$),
select that edge. Let $S$ be the set of the selected edges.

\item (Contract) Compute the connected components in
the graph $(V_{j,i}, S)$ and select one representative vertex for each
connected component. These representatives form the set $V_{j,i+1}$.
Construct a star graph for each connected component, with the representative
as the root, and the other vertices of the component pointing to it.

\item (Cleanup one bucket) Let $f(i)=1+$ the number of trailing $0$'s in the
binary representation of $i$. Clean up $B_{f(i)}$ of internal and multiple edges.

\item (Fill some buckets)
For $k = f(i)-1$ to $0$, for each supervertex $v$, store in bucket $B_k$
some of the lightest        edges ($(2^{2^k}-1)$ of them, if possible)
incident with $v$.

\item (Store the stars)
Delete all the edges from bucket $B_{f(i)}$ and instead store in
it the star graphs computed from (i) the star graphs obtained in step
2 of the present phase and (ii) the star graphs available in
buckets of lower indices.
\end{enumerate}

\item Perform the last ($2^{g(j)}$-th) Bor\r{u}vka phase of the stage.
In this phase, Bucket $B_{g(j)+1}$ is processed.
But this bucket contains all the edges in the graph,
and is therefore treated differently from the other buckets.
Using the star graphs available from
the previous phases, we clean up $B_{g(j)+1}$ of
internal and multiple edges. This leaves us with
a clean graph
$G_{j+1} = (V_{j+1}, E_{j+1})$
with which to begin the next stage.
\end{itemize}
\vspace{0.2in}
\hrule
\vspace{0.2in}

A discussion on the structure of a stage is now in order: The
buckets are repeatedly filled and emptied over the phases.
We call a bucket ``full'' when it contains edges, and
``empty'' when it contains star graphs, that is, information
regarding the clusterings of the past phases. (We use the word
``clustering'' to denote the hook and contract operations
performed during the course of the algorithm.) Let $f(i)=1+$ the
number of trailing $0$'s in the binary representation of $i$. The
$i$-th phase of the $j$-th stage, $1\leq i \leq 2^j \log(E/V)-1$,
(i)~uses up the edges in bucket $B_0$ for hooking, thereby
emptying $B_0$, (ii)~fills $B_{k}$, for all $k<f(i)$, from
$B_{f(i)}$,  and finally (iii)~empties $B_{f(i)}$. A simple
induction, therefore, shows that
\begin{quote}
for $i\geq 1$ and $k\geq 1$, bucket $B_{k}$ is full at the end of the $i$-th phase,
if and only if the $k$-th bit
(with the least significant bit counted as the first)
in the binary representation of $i$ is $0$; $0$ represents ``full'' and
$1$ represents ``empty''.
\end{quote}
The first phase forms the basis of the induction.
Prior to the first phase, at the end of a hypothetical $0$-th phase, all
buckets are full. The first phase uses up $B_0$, fills it from $B_1$ and
then empties $B_1$.
Emptying of bucket $B_{f(i)}$ and filling of all buckets of smaller indices
is analogous to summing $1$ and $(i-1)$ using binary representations;
this forms the induction step.

Definition: For $0\leq k\leq g(j)$, a $k$-interval is an interval $[i,i+2^{k})$
on the real line, where $i$ is a multiple of $2^{k}$. For $i\geq 1$,
the $f(i)$-interval $[i-2^{f(i)-1},i+2^{f(i)-1})$ is called
``the interval of $i$'' and is denoted by $I(i)$. Note that $i$
is the midpoint of $I(i)$. If we map the $i$-th phase to the interval $[i-1,i)$,
then $I(i)$ corresponds to phases numbered $i-2^{f(i)-1}+1,\ldots,i+2^{f(i)-1}$.

The binary representation of $i$ has exactly $f(i)-1$ trailing $0$'s.
Therefore, $f(i-2^{f(i)-1})>f(i)$.
In the $(i-2^{f(i)-1})$-th phase, we empty $B_k$ for some $k>f(i)$,
and in particular, fill $B_{f(i)}$.
All numbers between $(i-2^{f(i)-1})$ and $i$ have less
than $f(i)-1$ trailing $0$'s.
That means $B_{f(i)}$ is filled in the $(i-2^{f(i)-1})$-th phase,
and is never accessed again until the $i$-th phase, when it
is cleaned and emptied.
All numbers between $i$ and $(i+2^{f(i)-1})$ have less
than $f(i)-1$ trailing $0$'s.
Also, $f(i+2^{f(i)-1})>f(i)$.
So, in the $(i+2^{f(i)-1})$-th phase, we empty $B_k$ for some $k>f(i)$,
and in particular, fill $B_{f(i)}$ with edges taken
from it. To summarise, $B_{f(i)}$ is filled in the
$(i-2^{f(i)-1})$-th phase, emptied in the $i$-th phase and
filled again in the $(i+2^{f(i)-1})$-th phase.
That is, for $1\leq i\leq 2^{g(j)}-1$, $I(i)$ is the interval
between two successive fillings of $B_{f(i)}$.

In other words, for $1\leq k\leq g(j)$, $B_{k}$ is alternately
filled and emptied at intervals of $2^{k-1}$ phases.
In particular, $B_1$ is filled and emptied in alternate
phases. $B_0$ fulfills the role of a buffer that holds the lightest
edge incident with every vertex that is not overgrown.
It is filled in every step.

Now we discuss the five steps of a phase in detail.

{\bf Step 1:}
(Hook) For each vertex $u \in V_{j,i}$, if $B_{0}$ contains an edge
incident with $u$
(there can be at most one, because $2^{2^{0}}-1=1$),
select that edge. Let $S$ be the set of the selected edges.

{\bf Remarks on Step 1:} For each (super)vertex $u \in V_{j,i}$,
as we shall show, if $u$ is not sufficiently grown for phase $i$
in stage $j$ (that is, $u$ has not grown to a size of
$(E/V)^{2^{j-1}}\cdot 2^i$), then $B_{0}$ contains
an edge incident with $u$ in phase $i$, and this
is the lightest external edge incident with
$u$.

{\bf Step 2:}
(Contract) Compute the connected components in
the graph $H=(V_{j,i}, S)$ and select one representative vertex for each
connected component. These representatives form the set $V_{j,i+1}$.
Construct a star graph for each connected component, with the representative
as the root, and the other vertices of the component pointing to it.

{\bf Remarks on Step 2:}
Each $u \in V_{j,i}$ has exactly one outgoing edge in
$H=(V_{j,i}, S)$. Also, for any two consecutive edges
$(u,v)$ and $(v,w)$ in $H$, $\mbox{wt}(u,v)\geq\mbox{wt}(v,w)$.
Therefore, any cycle in $H$ must have a size of two. Moreover,
each component of $H$ is a pseudo tree, a connected directed graph
with exactly one cycle.

Make a copy $S'$ of $S$ and sort it on destination.
Recall that $S$ is sorted on source, and that
for each $u\in V_{j,i}$, there is at most one edge in $S$
with $u$ as the source. For $(u,v)\in S$, let us call $v$ the
parent $p(u)$ of $u$. Read $S$ and $S'$ concurrently.
For each $(v,w)\in S$ and each $(u,v)\in S'$, add $(u,w)$
into $S''$. For each $u$ such that $(u,u)\in S''$ and
$u<p(u)$, delete $(u,p(u))$ from $H$ and
    mark $u$ as a root in $H$.
Now $H$ is a forest.

Let $H'$ be the underlying undirected tree of $H$. Form an array
$T$ as follows: For $i\geq 1$, if the $i$-th element of array $S$
is $(u,v)$, then add $(u,v,2i+1)$ and $(v,u,2i)$ to $T$ as its
$(2i)$-th and $(2i+1)$-st elements. Clearly, $T$ is the edge list
of $H'$. Also an edge $(u,v)$ and its twin $(v,u)$ point to each
other in $T$. Sort $T$ on source and destination without violating
the twin pointers. We have obtained an adjacency list
representation of $H'$ with twin pointers which can be used in
computing an Euler tour.

Find an Euler tour $U$ of $H'$ by simulating the $O(1)$ time
EREW PRAM algorithm \cite{Ja92} on the external memory model \cite{ChGoGr+95}.
For each root node $r$ of $H$, delete
from $U$ the first edge with $r$ as the source. Now $U$ is a collection
of disjoint linked lists. For each element $(u,r)$ without a successor in $U$,
set $\mbox{rank}(u,r)=r$, and for every other element $(u,v)$, set
$\mbox{rank}(u,v)=0$. Now invoke list ranking on $U$, but with addition
of ranks replaced by bitwise OR. The result gives us the connected components
of $U$, and therefore of $H$. Each edge and therefore each vertex of $H$
now holds a pointer to the root of the tree to which it belongs.

Each connected component of $H$ forms a supervertex. Its root
shall be its representative. Thus, we have a star graph for
each supervertex.

The purpose of steps 3, 4 and 5 is to make sure that for each supervertex $v$ in
the remaining graph, if $v$ is not sufficiently grown, the
lightest external edge of $v$ is included
in bucket $B_{0}$ before the next phase.

As we mentioned earlier, bucket $B_{f(i)}$ is filled in the $(i-2^{f(i)-1})$-th phase,
and is never accessed again until the $i$-th phase.
To clean up $B_{f(i)}$, we need information regarding
all the clustering that
the algorithm has performed since then.
The following lemma is helpful:

\begin{lemma}
\label{step3} For every $i$, $1 \leq i < 2^{g(j)}$, and every $l$
such that $i\leq l< i+2^{f(i)-1}$ (that is, $l$ in the second half
of $I(i)$), at the end of the $l$-th phase, the clustering
information obtained from phases numbered
$i-2^{f(i)-1}+1,\ldots,i$ (that is, phases in the first half of
$I(i)$) is available in bucket $B_{f(i)}$.
\end{lemma}

{\bf Proof:} This can be proved by induction as follows.
The case for $i=1$ forms the basis: $f(1)=1$.
$1\leq l< 1+2^0=2$
implies that $l=1$. The first phase uses up $B_0$, fills it from $B_1$,
empties $B_1$ and then stores the newly found star graphs in
$B_{1}$. At the end of the first phase, therefore,
the clustering information obtained from
the first phase is indeed available in $B_{1}$. Hence the basis.

Hypothesise that for every $i$ ($1\leq i\leq p-1$) and every $l$ such that
$i\leq l<i+2^{f(i)-1}$,
at the end of the $l$-th phase,
the clustering information obtained from
$i-2^{f(i)-1}+1$ through $i$ is available in bucket $B_{f(i)}$.

The $p$-the phase provides the induction step.

If $f(p)=1$ (and therefore $p$ is odd), then
$p\leq l<p+2^0=p+1$
implies that $l=p$. The $p$-th phase uses up $B_0$, fills it from $B_1$,
empties $B_1$ and then stores the newly found star graphs in
$B_{1}$. At the end of the $p$-th phase, therefore,
the clustering information obtained from
the $p$-th phase is indeed available in $B_{1}$.

Suppose $f(p)>1$ (and therefore $p$ is even).
For $k$ with $1\leq k<f(p)$, let $r=p-2^{k-1}$.
Then $f(r)=k$. Since $r<p$, the hypothesis applies to $r$.
Also, $p-2^{k-1}=r\leq p-1<r+2^{f(r)-1}=p$. 
Therefore,
at the end of the $(p-1)$-st phase,
the clustering information obtained from
phases
$p-2^k+1$ through $p-2^{k-1}$
is available in bucket $B_{k}$, for $f(p)-1\geq k \geq 1$.
That is, all the clustering done since the last time $B_{f(p)}$ was filled,
which was in the $(p-2^{f(p)-1})$-th phase, till the $(p-1)$-st phase
is summarised in buckets $B_{f(p)-1}$ through $B_{1}$. This, along with
the star graphs formed in the $p$-th phase, is all the information that is needed
to clean up $B_{f(p)}$ in the $p$-th phase.

In the $p$-th phase $B_{f(p)}$ is cleaned up and the resultant star graphs are stored in it.
These summarise all the clustering done in phases numbered
$(p-2^{f(p)-1}+1)$ through $p$. After the $p$-th phase, $B_{f(p)}$ is accessed
only in the $(p+2^{f(p)-1})$-th phase.
Hence the induction holds.
\hfill $\Box$

We now introduce the notion of thresholds.
At the beginning of a stage, the threshold value
$h_k(v)$ for supervertex $v$ and bucket $B_{k}$ is defined as follows:
For a supervertex $v$ with at least $2^{2^k}$ external edges
let $e_{k}(v)$ denote the $2^{2^k}$-th lightest external edge of $v$.
Then,
\[h_k(v) =  \left\{
              \begin{array}{ll}
                 \mbox{wt}(e_{k}(v)) &
                              \mbox{~~if $v$ has at least $2^{2^k}$ external edges}\\
                 \infty & \mbox{~~otherwise}
              \end{array}
            \right. \]
Note that $h_{k}(v)$, when it is finite, is the weight of the lightest external edge of $v$
not in $B_k$; $h_{k}(v)=\infty$ signifies that every external edge of
$v$ is in $B_k$.
We store $h_k(v)$ at the end of $v$'s edgelist in
$B_k$. Note that the adding of the threshold values to the end of
the lists does not cause an asymptotic increase in the I/O complexity
of bucket formation.

Definition: A supervertex $v$ is $i$-maximal
if none of the supervertices formed in or before the $i$-th phase
is a proper superset of $v$.

Note that each supervertex at the beginning of stage $j$ is
$0$-maximal.

Definition: For a $t$-maximal supervertex $v$
(where $t\in I(i)$ and $1\leq i\leq 2^{g(j)}-1$)
and an $(i-2^{f(i)-1})$-maximal supervertex $x$,
$x$ is called an {\em $(i-2^{f(i)-1})$-seed} of $v$ if $x$ has the smallest
$h_{f(i)}$ threshold among all the
$(i-2^{f(i)-1})$-maximal supervertices that participate in $v$.
(We say that $x$ participates in $v$ when $x\subseteq v$.)

We use $r_{i}(v)$ to denote $h_{f(i)}(x)$, where $x$ is an $(i-2^{f(i)-1})$-seed of $v$.

{\bf Step 3:} Let $f(i)=1+$ the number of trailing $0$'s in the
binary representation of $i$.
Clean the edge lists of bucket $B_{f(i)}$ of
internal and multiple edges.
Let set $X_i=\phi$.
For each $i$-maximal supervertex $v$,
copy from $B_{f(i)}$ into $X_{i}$ all the external edges of $v$
with weight less than $r_{i}(v)$.

\noindent
{\bf Remark:} The cleanup operation makes use of all the clustering done since
the last time bucket $B_{f(i)}$ was clean, which was at the end of the
$(i-2^{f(i)-1})$-th phase. All the necessary information regarding past clustering
is available in buckets of smaller indices.

As described earlier in Lemma~\ref{step3}, if $f(i) > 1$
    bucket $B_{f(i)-1}$ stores the star graphs formed by MSF edges
    discovered in phases $(i - 2^{f(i)-1}) + 1$ through $(i - 2^{f(i) - 2})$.
Similarly if $f(i) > 2$, bucket $B_{f(i)-2}$ stores the
    star graphs formed in phases $(i - 2^{f(i)-2}) + 1$ through
    $(i - 2^{f(i)-3})$.
The leaves of the stars in $B_{f(i)-2}$ are the
roots of the stars in $B_{f(i)-1}$.
The leaves of the stars in $B_{f(i)-3}$ are the
roots of the stars in $B_{f(i)-2}$.
Continue like this and we have that bucket $B_1$ stores the star graphs formed
    in phase $(i-1)$ and the leaves of the stars formed in phase $i$
    are the roots of the stars in $B_1$.
Thus, the stars in $B_{f(i)-1}, \ldots, B_{1}$
along with the stars from the MSF edges computed in Step $2$ of the $i$-th phase,
form a forest $F$.

Reduce each connected component of forest $F$ to a star graph.
Let $F'$ denote the resultant forest.
This is done similarly to the contraction of Step 2.
Suppose $F'$ is expressed as a list of edges $(u,u_r)$, where $u_r$ is the
representative of the supervertex in which $u$ participates.

Rename the edges of $B_{f(i)}$ as follows: Sort $F'$ on source. Sort $B_{f(i)}$ on source.
Read $F'$ and $B_{f(i)}$ concurrently. Replace each $(u,v)\in B_{f(i)}$
with $(u_r,v)$. Sort $B_{f(i)}$ on destination.
Read $F'$ and $B_{f(i)}$ concurrently. Replace each $(u,v)\in B_{f(i)}$
with $(u,v_r)$. Sort $B_{f(i)}$ on the composite $\langle$source, destination, weight$\rangle$.

Remove from $B_{f(i)}$ edges of the form $(u,u)$; these are internal edges.
Also, for each $(u,v)$, if $B_{f(i)}$ has multiple copies of $(u,v)$,
then retain in $B_{f(i)}$ only the copy with the smallest weight.
This completes the clean up of $B_{f(i)}$. This way of clean-up
is well known in literature. (See \cite{ArBrTo04}, for example.)

The stars in $F'$ summarise all the clustering the algorithm performed
    in phases $(i - 2^{f(i) - 1} + 1)$ through $i$. 

Definition: We say that a set $P$ of edges is a minset of supervertex $v$,
if for any two external edges $e_1$ and $e_2$ of $v$ in $G$ with
$\mbox{wt}(e_1)<\mbox{wt}(e_2)$,  $e_{2}\in P$ implies that $e_{1}\in P$.

We now begin an inductive argument, which will be closed out in the remark after Step 4.

Inductively assume that at the end of the
$(i-2^{f(i)-1})$-th phase,
for each $(i-2^{f(i)-1})$-maximal supervertex $x$,
$B_{f(i)}$ is a minset
and the threshold $h_{f(i)}(x)$
is a lower bound on the weight of the lightest external edge of $x$ not in $B_{f(i)}$.
We claim that for each $i$-maximal supervertex $v$,
$X_i$ forms a minset.

We prove the claim by contradiction. Suppose $X_{i}$ is not a
minset of $v$. Then, among the edges of $G$ there must exist an
external edge $e$ of $v$ such that $\mbox{wt}(e)<r_{i}(v)$ and
$e\not\in X_{i}$. Say $e$ is an external edge of an
$(i-2^{f(i)-1})$-maximal supervertex $y$ that participates in $v$.
Clearly, $\mbox{wt}(e)<r_{i}(v)\leq h_{f(i)}(y)$. Of the external
edges of $y$ in $B_{f(i)}$, exactly those of weight less than
$r_{i}(v)$ are copied into $X_i$. Therefore, if $e\not\in X_{i}$,
then $e\not\in B_{f(i)}$. But then $B_{f(i)}$ is a minset of $y$.
That is, $h_{f(i)}(y)<\mbox{wt}(e)$.  Contradiction. Therefore,
such an $e$ does not exist.

Note that
the threshold $r_{i}(v)$
is a lower bound on the weight of the lightest external edge of $v$ not in $X_{i}$.

{\bf Step 4:} for $k = f(i)-1$ to $0$, and for each supervertex $v$,

- if $B_{f(i)}$ has at least $2^{2^k}$ edges with $v$ as the source, 
        then
               copy into $B_k$ the $(2^{2^k}-1)$ lightest of them,
               and set $h_{k}(v)$ to the weight of the $2^{2^k}$-th lightest of them;

-        else
               copy into $B_k$ all the edges in $B_{f(i)}$ with $v$ as the source,
               and set $h_{k}(v)$ to $r_{i}(v)$.

\noindent
{\bf Remark:}
Clearly, for each $i$-maximal supervertex $v$,
each $B_{k}$ ($k<f(i)$) formed in this step is a
minset.
Also,
the threshold $h_{k}(v)$ is a lower bound on
the weight of the lightest external edge of $v$ not in $B_{k}$.
When we note that at the beginning
of a stage, for each vertex $v$, and for each bucket $B_{k}$ that forms,
$B_{k}$ is a minset of $v$, and
$h_{k}(v)$ is the weight of the lightest external edge of $v$ not in $B_{k}$,
the induction we started in the remarks after Step 3 closes out.
Thus, we have the following lemma.

\begin{lemma}
Each freshly filled bucket is
a minset for each supervertex that is maximal at the time
of the filling.
Also, for each maximal supervertex $v$,
the threshold $h_{k}(v)$
is a lower bound on the weight of the lightest external edge of $v$ not in $B_{k}$.
\end{lemma}

\begin{cor}
At the beginning of a phase,
for each maximal supervertex $v$,
if $B_{0}$ contains an external edge $e$
of $v$, then $e$
is the lightest       external edge of $v$ in $G$.
\end{cor}

{\bf Step 5:}
Delete all the edges from bucket $B_{f(i)}$ and instead store in it $F'$, the set of star
graphs formed during the clean up.
$F'$ is maintained as an edgelist.

\subsection{Correctness of the Algorithm}
Corollary 1 above proves the correctness partly. It now remains to show that
every supervertex that is not ``overgrown'', will find its lightest external edge
in $B_{0}$ at the beginning of the next phase.

For each supervertex $v$ that has not grown to its full size
and bucket $B_k$, we define
the guarantee $Z_{k}(v)$ given by $B_{k}$ on $v$ as follows:
At the beginning of the $j$-th stage, if the component in $G_{j}$ that
contains $v$ has $s_v$ vertices, then
\[Z_k(v) =  \left\{
              \begin{array}{ll}
                 2^{k} &
                              \mbox{~~if $v$ has at least $2^{2^k}-1$ external edges}\\
                 \log s_v & \mbox{~~otherwise}
              \end{array}
            \right. \]

\noindent
For $i\geq 1$, at the end of the $i$-th phase, for each $i$-maximal supervertex $v$,
and for some $(i-2^{f(i)-1})$-seed $x$ of $v$, we set $Z_{f(i)}(v)=Z_{f(i)}(x)$.
For each $k<f(i)$, if $B_{f(i)}$ has at least $(2^{2^k}-1)$ edges with $v$ as the source,
then we set $Z_{k}(v)=i+2^{k}$; else we set $Z_{k}(v)=Z_{f(i)}(v)$.

\begin{lemma}
For $i\geq 1$, at the end of the $i$-th phase,
for each $k < f(i)$ and each $i$-maximal supervertex $v$ that has not grown to
its full size,
$Z_k(v) \geq i + 2^k$.
\end{lemma}

\proof The proof is by a strong induction on $i$.

The first phase forms the basis. When $i=1$, $f(i)=f(1)=1$ and
$(i - 2^{f(i)-1})=0$. In the first phase $B_1$ is emptied.
Let $v$ be a $1$-maximal supervertex that has not grown to
its full size. Then $s_{v}\geq 4$. (If $s_v=1$ then $v$ is isolated
in $G$, and so does not participate in the algorithm. If $s_v$ is $2$ or $3$,
then $v$ will be grown to its full size after the first hooking.) Let $x$ be a $0$-seed of $v$.
Then $Z_1(v) = Z_1(x)\geq 2$.
Thus, $Z_0(v)$ is set to either $Z_1(v)$ or $1 + 2^0 = 2$. Either way, $Z_{0}(v)\geq 2$,
and hence the basis.

Now consider the $i$-th phase. There are two cases.

Case 1: $f(i) < f(i-1)$. Clearly, $i$ is odd, and hence $f(i) = 1$
and $(i - 2^{f(i)-1})=i-1$. This case is similar to
the basis. $B_1$ is the bucket emptied.
Let $v$ be an $i$-maximal supervertex that has not grown to its full size, and let $x$ be an $(i-1)$-seed of $v$.
Then $Z_1(v) = Z_1(x)$. As $x$ is an $(i-1)$-maximal supervertex, and $1<f(i-1)$,
we hypothesise that $Z_1(x)\geq (i-1)+2^1 = i+1$.
Thus, $Z_0(v)$ is set to either $Z_1(v)$ or $i + 2^0 = i+1$. Either way, $Z_{0}(v)\geq i+1$,
and the induction holds.

Case 2: $f(i) > f(i-1)$, and hence $i$ is even.
We hypothesise that at the end of the $(i - 2^{f(i)-1})$-th phase,
for each $k < f(i - 2^{f(i)-1})$ and each $(i - 2^{f(i)-1})$-maximal supervertex $x$
not grown to its full size,
$Z_k(x) \geq i - 2^{f(i)-1} + 2^k$.
Since, $f(i)<f(i - 2^{f(i)-1})$, in particular,
$Z_{f(i)}(x) \geq i - 2^{f(i)-1} + 2^{f(i)} = i + 2^{f(i) - 1}$.
Let $v$ be an $i$-maximal supervertex not grown to its full size,
and let $x$ be an $(i-2^{f(i)-1})$-seed of $v$.
Then $Z_{f(i)}(v) = Z_{f(i)}(x) \geq i + 2^{f(i)-1}$.
Thus, for each $k<f(i)$, $Z_k(v)$ is set to either
$Z_{f(i)}(x) \geq i + 2^{f(i) - 1}$ or $i + 2^k$. Either way, $Z_k(v)$ is $\geq i+2^k$,
and the induction holds. \hfill $\Box$

Definition:
For $s\geq 0$, let $t_{s} = \lfloor \frac{t}{2^{s}} \rfloor 2^{s}$.
For $0\leq t\leq 2^{g(j)}-1$ and $0\leq k\leq g(j)$,
a $t$-maximal supervertex $v$ is $k$-stagnant if
the weight of the lightest external edge of $v$ is at least
the $h_{k}(x)$,  where $x$ is the $t_{k}$-seed of $v$.

Note that $[t_{k},t_{k}+2^k)$ is a $k$-interval. That means bucket $B_k$ is
filled in the $t_{k}$-th phase, and is not filled again till
phase $t_{k}+2^k >t$. If immediately after the formation of $B_{k}$ in
the $t_k$-th phase, we were to cull all the external
edges of $v$ in it into an edgelist sorted on weights, and then were to prune
this list using $h_{k}(x)$, then the resultant
list would be empty. Note also that
every maximal supervertex is $0$-stagnant.

Phase $0$ is a hypothetical phase prior to the first phase.
Therefore, a $0$-maximal vertex is a vertex of $G$, and so is a
$0$-seed.

\begin{lemma}
For $0\leq k\leq g(j)$ and $0\leq t< 2^{k}$,
if $v$ is a $t$-maximal $k$-stagnant supervertex such that
$x$ is a $0$-seed of $v$, then $\mbox{size}(v) \geq 2^{Z_k(x)}$.
\end{lemma}

\proof Let $v$ be a $t$-maximal $k$-stagnant supervertex such that
$x$ is a $0$-seed of $v$,
for $0\leq k\leq g(j)$ and
$0\leq t< 2^{k}$. Here $t_{k}=0$.
That is, the lightest external edge of
$v$ has a weight not less than $h_{k}(x)$.

If $x$ has at least
$2^{2^k}-1$ outgoing edges given to it in $B_{k}$ at the beginning
of the stage, then all those edges are internal to $v$; so $x$ and
its neighbours ensure that $v$ has a size of at least
$2^{2^k}=2^{Z_k(x)}$. If $x$ has less than $2^{2^k}-1$ outgoing
edges in $B_{k}$, then $h_{k}(x)=\infty$, and $h_{k}(x')=\infty$
for every $0$-maximal $x'$ participating in $v$. Thus, $v$ has no
external edge. Then $v$ cannot grow anymore, and so has a size of
$2^{Z_k(x)}=s_v$. \hfill $\Box$

In the $i$-th phase, when we put $x$ up into $B_{k}$ with a guarantee of $Z_{k}(v)$
we would like to say that if ever a $k$-stagnant supervertex $v$ forms with $x$ as its
$i$-seed, then $v$ would have a size of at least $2^{Z_{k}(v)}$.
That is what the next lemma does.

We defined $f(i)$ as $1+$ the number of trailing $0$'s in the
binary representation of $i$, for $i\geq 1$. Let $f(0)$ be
defined as $g(j)+1$.

\begin{lemma}
For each $i$, $0\leq i\leq 2^{g(j)}$,

(1) for each $i$-maximal supervertex $v$,
    if $v$ has not grown to its final
    size in $G$, then $\mbox{size}(v) \geq 2^i$, and

(2) for $0\leq k<f(i)$, and $i\leq t< i+2^{k}$, if $v$ is a $t$-maximal
    $k$-stagnant supervertex such that
$x$ is an $i$-seed of $v$, then $\mbox{size}(v) \geq 2^{Z_k(x)}$.
\end{lemma}

\proof
The proof is by induction on $i$. The case of $i=0$ forms the basis for (2),
whereas the case of $i=1$ forms the basis for (1).
Every $0$-maximal supervertex $v$ has a size of one. Lemma 4 proves
(2) for $i=0$. For every vertex $v$ in $G_j$ that is not isolated,
$B_0$ holds the lightest external edge of $v$ in $G$ at the beginning
of stage $j$, and so $v$ is certain to hook in the first phase.
Thus every $1$-maximal supervertex that has not grown to its final
size in $G$ will have a size of at least two.

Now consider $p>0$. Hypothesise that (1) is true for all $i\leq p$, and
 (2) is true for all $i\leq p-1$.

\noindent
{\bf The Induction steps:}
Consider the $p$-th phase. After the hook and contract,
$B_{f(p)}$ is emptied and for all $k < f(p)$, bucket $B_k$ is filled.
Suppose for $t \in [p, p + 2^k )=I(p+2^{k-1})$, a $t$-maximal $k$-stagnant supervertex $v$ forms.
Let $x$ be a $p$-seed of $v$. In the $p$-th phase,
$x$ is put up into $B_{k}$ with a
guarantee of $Z_k(x)$. Let $y$ be a $(p-2^{f(p)-1})$-seed of $x$.
There are two cases.

\noindent Case 1: $Z_k(x)$ = $Z_{f(p)}(y)$. In this case,
supervertex $x$ gets all the edges from $B_{f(p)}$, and $h_k(x)$ =
$h_{f(p)}(y)$. So, $y$ is a $(p-2^{f(p)-1})$-seed of $v$ too.
The lightest external edge of supervertex $v$ has a weight not less
than $h_{f(p)}(y)=h_{k}(x)$, because otherwise $v$ would not be
$k$-stagnant as assumed.
So supervertex $v$ is $f(p)$-stagnant.
Also supervertex $v$ is $t$-maximal for some
$t \in [ p , p + 2^k ) \subseteq [p-2^{f(p)-1}, p + 2^{f(p)-1})$.
So $v$ qualifies for an application of hypothesis (2) with
$i=p-2^{f(p)-1}$ and $k=f(p)$.
Therefore, $\mbox{size}(v) \geq 2^{Z_{f(p)}(y)}$.
But $Z_{f(p)}(y) = Z_k(x)$, so $\mbox{size}(v) \geq 2^{{Z_k}(x)}$.

\noindent Case 2: If $Z_k(x) \neq Z_{f(p)}(y)$, then $Z_k(x) = p +
2^k$. In this case supervertex $v$ may not be $l$-stagnant for $l
> k$. Supervertex $x$ comes into $B_k$ with $2^{2^k} - 1$ edges.
All those edges are internal to $v$. So $v$ has a
size $\geq 2^{2^k}$ in terms of $p$-maximal supervertices. Since
each $p$-maximal supervertices has size $\geq 2^p$ (by hypothesise (1)),
$\mbox{size}(v) \geq 2^{2^k} \cdot 2^p =
2^{p + 2^k} = 2^{Z_k(x)}$.

Thus (2) is proved for $i=p$.
Note that hypothesis (2) pertains to $I(p)$ and the induction step
argues that (2) holds for $I(p+2^{k-1})$ for each $k<f(p)$.

\noindent
Now we prove statement (1) of the lemma for the $(p+1)$-st phase.
Suppose $x$ is a $(p+1)$-maximal supervertex that does not grow in the
$(p+1)$-st phase. That means $x$ is $0$-stagnant at the beginning of the $(p+1)$-st phase.
That is possible only if $B_{f(p)}$ has no edge to put up for $x$ in $B_0$ in the
$p$-th phase.
So $x$ is $f(p)$-stagnant, and
qualifies for an application of hypothesis (2) with
$i=p-2^{f(p)-1}$ and $k=f(p)$.
Suppose $y$ is a $(p - 2^{f(p) - 1})$-seed of $x$. Then $\mbox{size}(x) \geq 2^{Z_{f(p)}(y)}$.

\noindent
From Lemma 3, at $p - 2^{f(p) - 1}$, supervertex $y$ was put up into $B_{f(p)}$ with
    $Z_{f(p)}(y) \geq p - 2^{f(p) - 1} + 2^{f(p)} \geq p+1$.
So $\mbox{size}(x) \geq 2^{p+1}$.

If $x$ is a $(p+1)$-maximal supervertex that does grow in the
$(p+1)$-st phase, then it has at least two $p$-maximal supervertices (each of size at least
$2^p$ by hypothesis (1)) participating in it.
So it has a size
of at least $2^{p+1}$.

We conclude that every $(p+1)$-maximal supervertex has a size of at least
$2^{p+1}$. Hence the induction.
\hfill $\Box$

This guarantees that every supervertex that has not grown to a size of $2^{i}$
will indeed find an edge in $B_{0}$ in the $i$-th phase.
This completes the proof of correctness.

\subsection{The I/O Complexity of the Algorithm}
We discuss below the number of I/Os taken by each step of phase $i$ in stage $j$:

\noindent
{\bf I/Os in Step 1:} Since array $B_0$ holds the edges sorted on source, $S$ is
merely a copy of $B_0$, and so can be written out in
$O(\sort{V_{ji}})$ I/Os.

\noindent {\bf I/Os in Step 2:} One step of a Parallel Random
Access Machine (PRAM) that uses $N$ processors and $O(N)$ space
can be simulated on the external memory model in $O(\sort{N})$
I/Os \cite{ChGoGr+95}. The Euler tour of a tree of $N$ vertices
given in adjacency list representation with twin pointers can be
formed in $O(1)$ time with $O(N)$ processors on an Exclusive Read
Exclusive Write (EREW) PRAM \cite{Ja92}.

If $Y$ is a permutation of an array $X$ of $n$ elements,
and if each element in $X$ knows its position in $Y$,
then any $O(1)$ amount of information that each element
of $X$ holds can be copied into the corresponding element
of $Y$ and vice versa in $O(1)$ time using $n$ processors on an EREW PRAM.
Therefore, if each element of $X$ holds a pointer to another element of $X$, then
these pointers can be replicated in $Y$
in $O(1)$ time using $n$ processors on an EREW PRAM,
and hence in $O(\sort{N})$ I/Os  on the external memory model.

The list ranking algorithm of \cite{ChGoGr+95} when invoked on a list
of size $n$ takes $O(\sort{n})$ I/Os.

Putting it all together, the I/O complexity of Step 2 is therefore, $O(\sort{V_{j,i}})$.

\noindent
{\bf I/Os in Step 3:}
Let $b_{k,j,i}$ denote an upper bound on the size of bucket $B_{k}$ in phase $i$
of stage $j$.
(We define $b_{k,j,i}$ as the product of the number of supervertices that remain in
phase $i$, and $2^{2^k}-1$, the maximum degree a vertex can have in bucket $k$.)

For fixed $i$ and $j$, $b_{k,j,i}$ varies as $2^{2^k}-1$ with $k$.
Clearly, for any $k$, $b_{k,j,i}\geq\Sigma_{l=0}^{k-1} b_{l,j,i}$.

Also, $b_{l,j,i}>2b_{l,j,i+1}$, because the number supervertices that have
external edges halves in every phase.
Thus $F'$ that summarises the clustering
information in all of $B_{f(i)-1},\ldots,B_{0}$ has a size of
at most $b_{f(i),j,i-2^{f(i)-1}}$ space. $B_{f(i)}$
has a size of at most $b_{f(i),j,i}$. Therefore, the clean up of
$B_{f(i)}$
can be carried out in $\sort{b_{f(i),j,i-2^{f(i)-1}}}$ I/Os.

\noindent
{\bf I/Os in Step 4:}
Once $B_{f(i)}$ has been cleaned up,
for $f(i)> k\geq 0$, bucket $B_{k}$
can be formed by scanning $B_{k+1}$ and choosing for
each vertex
$v \in V_j$, the $(2^{2^k}-1)$ lightest
edges incident with $v$.
Clearly, this involves
scanning each bucket twice, once for write and once for read.
Since the bucket sizes fall superexponentially as $k$ falls,
this can be done in $O(b_{f(i),j,i}/B)$ I/Os.

\noindent {\bf I/Os in Step 5:} The clustering information computed earlier in Step 3
is stored in $B_{f(i)}$ in this step, The cost of this step is clearly dominated
by that of Step 3.

The total number of I/Os executed by the $i$-th phase of the $j$-th stage is therefore
$O(\sort{V_{j,i}}+ \scan{b_{f(i),j,i}} + \sort{b_{f(i),j,i-2^{f(i)-1}}})=
O(\sort{V_{j,i}}+ \sort{b_{f(i),j,i-2^{f(i)-1}}})$.

In the final (i.e., $2^{g(j)}$-th) Bor\r{u}vka phase, the cleanup of bucket $B_{g(j)
+ 1}$ uses the clustering information from previous
phases which is available in the lower numbered buckets. This
clean up is the same as the one in step $3$ and can be executed in
$O(\sort{E_j})$ I/Os.

Therefore, the total I/O cost of the $j$-th stage is
\[O(\sort{E_j})+ \sum^{2^{g(j)}-1}_{i=1} O(\sort{V_{j,i}} + \sort{b_{f(i),j,i-2^{f(i)-1}}})\]
\[= O(\sort{E_j}) + \sum^{g(j)}_{k=0} \sum^{(g(j)/2^k)-1}_{r=0} O(\sort{b_{k,j,r.2^k}})\]
\[= O(\sort{E}) + \sum^{g(j)}_{k=0} O(\sort{b_{kj0}}) =
O(\sort{E})+ \sum^{g(j)}_{k=0} O(\sort{V_{j} (2^{2^{k}}-1))}\]
which is $O(\sort{E})$.
(See the discussion at the beginning of Section 3.)

\begin{lemma}
Our algorithm performs $\log (VB/E)$ phases in $O(\sort{E} \log
\log_{E/V} B)$ I/Os.
\end{lemma}
\proof As argued above, the total I/O cost of the $j$-th stage is
$O(\sort{E})$. If the total number of stages is $r$, the total
number of phases executed is: $\Sigma_{j=0}^{r-1} 2^{g(j)} =
(2^r-1)\log(E/V)$. If this is to be $\log (VB/E)$, then $r$ must
be $\log \log _{E/V} B$. \hfill $\Box$

Thus, we reduce the minimum spanning tree problem of an undirected
graph $G = (V,E)$ to the same problem on a graph with $O(E/B)$
vertices and $O(E)$ edges. On this new graph, the external memory
version of Prim's algorithm can compute a minimum spanning tree in
$O(E/B + \sort{E})$ I/Os. From the MSF of the reduced graph, an
MSF of the original graph can be constructed; the I/O complexity
of this will be dominated by the one of the reduction.

Putting everything together, therefore,

\begin{theorem}
Our algorithm computes a minimum spanning forest of an undirected
graph $G = (V,E)$ in $O(\sort{E} \log \log_{E/V} B)$ I/Os.
\end{theorem}

\section{Conclusions from this Chapter}
\label{conc:mst}
In this chapter we present an improved external memory algorithm for the computing of minimum
	spanning forests of graphs.
Our algorithm uses a novel scheduling technique 
	on the $\log \frac{VB}{E}$ Bor\r{u}vka phases needed to reduce the graph size.
In our scheduling technique, each bucket is emptied and filled at regular intervals. 
This might be restrictive, because at the point of emptying, a bucket may contain
many more edges than are necessary to fill the lower numbered buckets, but all
those are discarded. This slack in our scheduling could be exploited to design
a faster algorithm. We have not yet succedded.
   \chapter{External Memory Soft Heap and Hard Heap, a Meldable Priority Queue}
\label{emsh:chapt}
\section{Introduction}
An external memory version of soft heap that we call ``External Memory Soft Heap''
(EMSH for short) is presented. 
It supports {\tt Insert}, {\tt Findmin}, {\tt Deletemin} and {\tt Meld} operations.
An EMSH may, as in its in-core version, and at its discretion, 
corrupt the keys of some elements in it, by revising them upwards.
But the EMSH guarantees that the number of corrupt elements in it is never more than
$\epsilon N$, where $N$ is the total number of items inserted in it, and $\epsilon$ is
a parameter of it called the error-rate.
The amortised I/O complexity of an {\tt Insert} is 
$O(\frac{1}{B} \log_{m}\frac{1}{\epsilon})$, where $m = \frac{M}{B}$.
{\tt Findmin}, {\tt Deletemin} and {\tt Meld} all have non-positive amortised I/O complexities.

When we choose an error rate $\epsilon<1/N$, EMSH stays devoid of corrupt nodes, 
and thus becomes a meldable priority queue that we call ``hard heap''.  
The amortised I/O complexity of an {\tt Insert}, in this case, is $O(\frac{1}{B} \log_{m}\frac{N}{B})$,
over a sequence of operations involving $N$ {\tt insert}s.
{\tt Findmin}, {\tt Deletemin} and {\tt Meld} all have non-positive amortised I/O complexities.
If the inserted keys are all unique, a {\tt Delete} (by key) operation can
	also be performed at an amortised I/O complexity of $O(\frac{1}{B} \log_{m}\frac{N}{B})$.
A balancing operation performed once in a while on a hard heap 
ensures that the number of I/Os performed by a sequence of $S$ operations on it is
$O(\frac{S}{B}+\frac{1}{B}\sum_{i = 1}^{S} \log_{m}\frac{N_i}{B})$, where
$N_i$ is the number of elements in the heap before the $i$th operation.

\subsection{Definitions}
A priority queue is a data structure used for maintaining a set $S$ of elements, 
	where each element has a key drawn from a linearly ordered set.
A priority queue typically supports the following operations:
\begin{enumerate}
\item {\tt Insert}$(S,x)$: Insert element $x$ into $S$.
\item {\tt Findmin}$(S)$: Return the element with the smallest key in $S$
\item {\tt Deletemin}$(S)$: Return the element with the smallest key in $S$ and remove it from $S$.
\item {\tt Delete}$(S,x)$: Delete element $x$ from $S$
\item {\tt Delete}$(S,k)$: Delete the element with key $k$ from $S$
\end{enumerate}
Algorithmic applications of priority queues abound \cite{AHU74,CLR90}.

Soft heap is an approximate meldable priority queue devised
	by Chazelle \cite{Ch00a}, and supports
	{\tt Insert}, {\tt Findmin}, {\tt Deletemin}, {\tt Delete}, and
	{\tt Meld} operations. 
A soft heap, may at its discretion, corrupt the keys of some elements in it, by revising them upwards.
A {\tt Findmin} returns the element with the smallest current key, which may or may not be corrupt. 
A soft heap guarantees that the number of corrupt elements in it is never more than
$\epsilon N$, where $N$ is the total number of items inserted in it, and $\epsilon$ is
a parameter of it called the error-rate.   
A {\tt Meld} operation merges two soft heaps into one new soft heap.

\subsection{Previous Results}
I/O efficient priority queues have been reported before \cite{Arge03,BK98,FJKT99,KS96,MZ03}.
An I/O efficient priority queue, called \emph{buffer tree}  was introduced by 
        Arge \cite{Arge03}.
A buffer tree is an $(a,b)$ tree, where $a = M/4$ and $b = M$, and each
        node contains a buffer of size $\Theta(M)$.
Buffer tree supports the {\tt Insert}, {\tt  Deletemin}, {\tt Delete} (by key), 
	and offline search operations.
The amortised complexity of each of these operations on buffer tree is
        $O(\frac{1}{B} \log_{m}\frac{N}{B})$
        I/Os, over a sequence of operations of length $N$.
External memory versions of heap are presented in \cite{KS96} and \cite{FJKT99}. 
The heap in \cite{KS96} is a $\sqrt{m}$-way tree, each node of which contains a buffer 
of size $\Theta(M)$; it supports {\tt Insert}, {\tt Deletemin}, and {\tt Delete} (by key) operations.
The amortised cost of each operation on this heap is
        $O(\frac{1}{B} \log_{m}\frac{N}{B})$
        I/Os, where $N$ is the total number of
        elements in the heap.
The heap in \cite{FJKT99} is an $m$-way tree and it does not contain a buffer at each node.
It supports {\tt Insert}, and {\tt Deletemin} operations.
For this heap, the total number of I/Os performed by a sequence of $S$
        operations is $O(\frac{S}{B}+\frac{1}{B}\sum_{i = 1}^{S} \log_{m}\frac{N_i}{B})$, where$N_i$ is the number
        of elements in the heap before the $i$th operation.

An external memory version of tournament tree that supports 
        the {\tt Deletemin}, {\tt Delete}, and {\tt Update}
        operations is also presented in \cite{KS96};
this is a complete binary tree.
An {\tt Update}$(x,k)$ operation changes the key of element $x$ to $k$
        if and only if $k$ is smaller than the present key of $x$.
The amortised cost of each operation on this data structure is $O(\frac{1}{B} \log_2 \frac{N}{B})$ I/Os \cite{KS96}.

The priority queue of \cite{BK98} maintains a hierarchy of sorted lists in secondary memory.
An integer priority is presented in \cite{MZ03}.
See \cite{BCF+00} for an experimental study on some of these priority queues.
Numerous applications of these data structure have also been reported:
graph problems, computational geometry problems and sorting to name a few
\cite{Arge03,FJKT99,KS96}.

See Table \ref{comparison:emsh} for a comparison of hard heap with the
        priority queues of \cite{Arge03,FJKT99,KS96}.

Soft heap is an approximate meldable priority queue devised
        by Chazelle \cite{Ch00a}, and supports
        {\tt Insert}, {\tt Findmin}, {\tt Deletemin}, {\tt Delete}, and
        {\tt Meld} operations. 
This data structure is used in computing minimum spanning
        trees \cite{Ch00b} in the fastest known in-core algorithm for the problem.
Soft heap has also been used for finding exact and approximate
    medians, and for approximate sorting \cite{Ch00a}.
An alternative simpler implementation of soft heap 
is given by Kaplan and Zwick \cite{KaZw09}.

\begin{table}
\begin{center}
            \begin{tabular}{|l|l|l|l|l|l|}
               \cline{1-6}
               \hline
{\bf Properties}  &   {\bf buffer tree}     & {\bf heap} \cite{KS96}      & {\bf heap} \cite{FJKT99}    & {\bf tourn. tree}
   & {\bf EMSH} \\
\hline \hline
type of tree &  $(M/4,M)$   & $\sqrt{m}$-way & $m$-way & complete & set of
                                                                      $\sqrt{m}$- \\
             &  tree        & tree           & tree    & binary tree & way trees \\
\hline
size of a     & $\Theta(M)$              & $\Theta(\sqrt{m}B)$  & $\Theta(M)$        & $\Theta(M
)$    & $\Theta(\sqrt{m}B)$ \\
node          &                  &                &          &        &        \\
\hline
buffer of size &    yes          &   yes          & no       & yes    & no  \\
$M$ at each    &                 &                &          &        &   \\
node            &                 &                &          &        &   \\
\hline
extra space    & $\Theta(N)$     &  $\Theta(\sqrt{m}N)$ & 0   & $\Theta(N)$ & 0 \\
\hline
operations     & {\tt Insert} & {\tt Insert}    & {\tt Insert}    & {\tt Delete}    & {\tt Insert} \\
               & {\tt Delete} & {\tt Delete}    & {\tt Deletemin} & {\tt Deletemin} & {\tt Delete} \\
            & {\tt Deletemin} & {\tt Deletemin} & {\tt Findmin}   & {\tt Findmin}   & {\tt Deletemin} \\
            & {\tt Findmin}   & {\tt Findmin}   &                 & {\tt Update}& {\tt Findmin}, {\tt Meld} \\
\hline
\end{tabular}
\caption{Comparison with some known priority queues }
\label{comparison:emsh}
\end{center}
\end{table}
\subsection{Our Results}
In this chapter, we present an external memory version of soft heap that permits
batched operations. We call our data structure ``External Memory Soft Heap''
(EMSH for short). 
As far as we know, this is the first implementation of soft heap
on an external memory model. 
When we choose an error rate $\epsilon<1/N$, EMSH stays devoid of corrupt nodes, 
and thus becomes a meldable priority queue that we call ``hard heap''.

EMSH is an adaptation of soft heap for this model.
It supports {\tt Insert}, {\tt Findmin}, {\tt Deletemin} and {\tt Meld} operations.
An EMSH may, as in its in-core version, and at its discretion, 
corrupt the keys of some elements in it, by revising them upwards.
But it guarantees that the number of corrupt elements in it is never more than
$\epsilon N$, where $N$ is the total number of items inserted in it, and $\epsilon$ is
a parameter of it called the error-rate.
The amortised I/O complexity of an {\tt Insert} is $O(\frac{1}{B} \log_{m}\frac{1}{\epsilon})$.
{\tt Findmin}, {\tt Deletemin} and {\tt Meld} all have non-positive amortised I/O complexities.

A hard heap (an EMSH with $\epsilon < 1/N$) does not have any corrupt element.
Therefore, it is an exact meldable priority queue.
The amortised I/O complexity of an {\tt Insert}, in this case, is $O(\frac{1}{B} \log_{m}\frac{N}{B})$.
{\tt Findmin}, {\tt Deletemin} and {\tt Meld} all have non-positive amortised I/O complexities.
If the inserted keys are all unique, a {\tt Delete} (by key) operation can
	also be performed at an amortised I/O complexity of $O(\frac{1}{B} \log_{m}\frac{N}{B})$.

\subsection{Organisation of This Chapter}

This chapter is organised as follows: Section~\ref{datastr:emsh} describes the data structure.
The correctness of the algorithm is proved in Section~\ref{correct:emsh}.
The amortised I/O analysis of the algorithm is presented in Section~\ref{analy:emsh}.
EMSH with $\epsilon<1/N$ is discussed in Section~\ref{hard:emsh}.
Some of its applications are shown in Section~\ref{appl:emsh}
\section{The Data Structure}
\label{datastr:emsh}
An EMSH consists of a set of trees on disk.
The nodes of the trees are classified as follows.  
A node without a child is a leaf. 
A node without a parent is a root.
A node is internal, if it is neither a leaf, nor a root. 

Every non-leaf in the tree has at most $\sqrt{m}$ children.
Nodes hold pointers to their children. 

Every node has a rank associated with it at the time of its creation. 
The rank of a node never changes.
All children of a node of rank $k$ are of rank $k-1$. 
The rank of a tree $T$ is the rank of $T$'s root. 
The rank of a heap $H$ is $\max\{\mbox{rank}(T) \ |\ T\in H\}$. 
An EMSH can have at most $\sqrt{m}-1$ trees of any particular rank.

Each element held in the data structure has a key drawn from a 
	linearly ordered set. 
We will treat an element and its key as indistinguishable.

Each instance of EMSH has an associated error-rate $\epsilon>0$. 
Define $r = \log_{\sqrt{m}}{1/\epsilon}$.
Nodes of the EMSH with a rank of at most $r$ are called {\em pnodes} (for ``pure nodes''), and 
	nodes with rank greater than $r$ are called {\em cnodes} (for ``corrupt nodes'').
Each {\em pnode} holds an array that contains elements in sorted order.
A tree is a ptree if its rank is at most $r$, and a ctree otherwise.

We say that a pnode $p$ satisfies the \emph{pnode invariant} (PNI), if 
\begin{quote}
	$p$ is a non-leaf and the array in $p$ contains at most 
	$B\sqrt{m}$ and at least $B\sqrt{m}/2$ elements, or \\
	$p$ is a leaf and the array in $p$ contains at most 
	$B\sqrt{m}$ elements.
\end{quote}
Note that a pnode that satisfies PNI may contain less than $B\sqrt{m}/2$ elements,
	if it is a leaf.

Every cnode has an associated doubly linked list of {\em listnodes}. 
A cnode holds pointers to the first and last listnodes of its list. 
The size of a list is the number of listnodes in it. 
Each listnode holds pointers to the next and previous
	listnodes of its list; the next (resp., previous)
        pointer of a listnode $l$ is null if $l$  is the last (resp., first)
        of its list.
Each listnode contains at most $B\sqrt{m}$, and unless it is the last of a list,
	at least $B\sqrt{m}/2$  elements. 
The last listnode of a list may contain less than $B\sqrt{m}/2$ elements.

Let $s_k$ be defined as follows:
\begin{equation*}
s_k =
 \begin{cases}
 0 & \text{if $k \leq r$} \\
 2 & \text{if $k = r+1$} \\
 \lceil \frac{3}{2} s_{k-1} \rceil & \text{if $k > r+1$}
 \end{cases}
\end{equation*}

We say that a cnode $c$ that is a non-leaf  
	satisfies the \emph{cnode invariant} (CNI), if 
	the list of $c$ has a size of at least $\lfloor s_k/2 \rfloor + 1$. 
A leaf cnode always satisfies CNI.

\begin{table}
\begin{center}
            \begin{tabular}{|l|l|l|l|}
               \cline{1-4} \hline
 {\bf Type of node} & {\bf Number of} & {\bf Number of elements in a} & {\bf Size of the list of a} \\
              & {\bf children} & {\bf pnode of this type} & {\bf cnode of this type} \\
              &          & {\bf if it satisfies PNI} & {\bf if it satisfies CNI} \\
\hline \hline
               leaf node & $0$ & $\leq B\sqrt{m}$ & $\geq 1$  \\
\hline
               root node & $\leq\sqrt{m}$ & $\geq B\sqrt{m}/2$; $\leq B\sqrt{m}$  & $\geq \lfloor s_k/2 + 1\rfloor$ \\
\hline
               internal node & $\leq\sqrt{m}$ & $\geq B\sqrt{m}/2$; $\leq B\sqrt{m}$ & $\geq \lfloor s_k/2 + 1\rfloor$  \\
\hline
            \end{tabular}
\caption{Types of nodes in the data structure, and the invariants on them}
\label{tab:data:structure:emsh}
\end{center}
\end{table}
 
Table \ref{tab:data:structure:emsh} summarizes the different types of nodes in an EMSH, 
the number of children each can have, and the PNI and CNI constraints on each.

Every cnode has a {\em ckey}.
For an element $e$ belonging to the list of a cnode $v$,
	the ckey of $e$ is the same as the ckey of $v$; 
	$e$ is called corrupt if its ckey
	is greater than its key.

An EMSH is said to satisfy the heap property if the following
	conditions are met:
For every cnode $v$ of rank greater than $r+1$, the ckey of $v$ is smaller than 
	the ckey of each of $v$'s children.
For every cnode $v$ of rank $r+1$, the ckey of $v$ is smaller than
	each key in each of $v$'s children.
For every pnode $v$, each key in $v$ is smaller than each key in each of $v$'s
	children.

For each rank $i$, we maintain a bucket $B_{i}$ for the roots of rank $i$. 
We store the following information in $B_i$:
\begin{enumerate}
\item the number of roots of rank $i$ in the EMSH; there are at most $\sqrt{m}-1$ such roots.
\item pointers to the roots of rank $i$ in the EMSH.
\item if $i>r$ and $k=\min\{\mbox{ckey}(y)\ |\ y \mbox{ is a root of rank } i\}$
	then a listnode of the list associated with the root of rank $i$, whose ckey value is $k$;
	this listnode will not be the last of the list, unless the list has only one listnode. 
\item if $i \leq r$ then the $n$ smallest of all elements in the roots of rank $i$, 
	for some $n\leq B\sqrt{m}/2$
\item a pointer suffixmin$[i]$
\end{enumerate}

We define the minkey of a tree as follows:
for a ptree $T$, the minkey of $T$ is defined as the 
	smallest key in the root of $T$;
for a ctree $T$, the minkey of $T$ is the ckey of the root of $T$.
The minkey of a bucket $B_{i}$ is the smallest of the minkeys of the trees of rank $i$ in the EMSH;
	$B_{i}$ holds pointers to the roots of these trees.  
The suffixmin pointer of $B_{i}$ points to the bucket with the smallest minkey
	among $\{B_{x}\ |\ x\geq i\}$.  

For each bucket, we keep the items in 1, 2 and 5 above, and at most a block of the elements 
	(3 or 4 above) in the main memory.
When all elements of the block are deleted by {\tt Deletemin}s, the next 
	block is brought in.
The amount of main memory needed for a bucket is, thus, $O(B+\sqrt{m})$. 
As we shall show later, the maximum rank in the data structure, and so the number
	of buckets is $O(\log_{\sqrt{m}} (N/B))$. Therefore, if 
	$N = O(B m^{M/2(B + \sqrt{m})})$ 
	the main memory suffices for all the buckets.
(See Subsection~\ref{memory:requirement}). 

We do not keep duplicates of elements. 
All elements and listnodes that are taken into the buckets would be 
	physically removed from the respective roots.
But these elements and listnodes would still be thought of as belonging to their
	original positions.
For example, the above definition of minkeys assumes this.

A bucket $B_i$ becomes {\em empty}, irrespective of the value of $i$,
when all the elements in it have been deleted.


\subsection{The Operations}

In this section we discuss the {\tt Insert}, {\tt Deletemin}, {\tt Findmin}, {\tt Meld}, {\tt Sift} and {\tt Fill-Up} 
    operations on EMSH. The first four are the basic operations. 
The last two are auxiliary.
The {\tt Sift} operation is invoked only on non-leaf nodes that fail to satisfy the 
	pnode-invariant (PNI) or cnode-invariant (CNI), whichever is relevant.
When the invocation returns, the node will satisfy the invariant.
Note that PNI applies to pnodes and CNI applies to cnodes. 
{\tt Fill-Up} is invoked by the other operations on a bucket when they find it empty.

\subsubsection{{\tt Insert}}
For each heap, a buffer of size $B\sqrt{m}$ is maintained in the main memory. 
If an element $e$ is to be inserted into heap $H$, store it in the buffer of $H$. 
The buffer stores its elements in sorted order of key values. 
If the buffer is full (that is, $e$ is the $B\sqrt{m}$-th  element of the buffer), 
	create a new node $x$ of rank $0$, and copy all elements in the buffer into it.
The buffer is now empty.
Create a tree $T$ of rank $0$ with $x$ as its only node. 
Clearly, $x$ is a root as well as a leaf. 
Construct a new heap $H'$ with $T$ as its sole tree. 
Create a bucket $B_0$ for $H'$, set the number of trees in it to $1$, and include a pointer to $T$ in it.
Invoke {\tt Meld} on $H$ and $H'$.

\subsubsection{{\tt Deletemin}}
A {\tt Deletemin} operation is to delete and return an element with the smallest key in the EMSH.
The pointer suffixmin$[0]$ points to the bucket with the smallest minkey.
A {\tt Deletemin} proceeds as in Figure~\ref{fig:deletemin}.

\begin{figure}
\vspace{0.2in}
\hrule
\begin{tabbing}
aaaa \= aaaa \= aaaa \= aaaa \= aaaa \= aaaa \kill \\
Let $B_i$ be the bucket pointed by suffixmin$[0]$; \\
let $e$ be the smallest element in the insert buffer; \\
if the key of $e$ is smaller than the minkey of $B_i$ then \\
	\> delete $e$ from the insert buffer, and return e; \\
\\
if $i\leq r$ then \\
	\> the element with the smallest key in the EMSH is in bucket $B_{i}$; \\
	\> let it be $e$; delete $e$ from $B_{i}$; \\
	\> if $B_{i}$ is not empty, then \\ 
	\>	\> update its minkey value; \\

else \\
	\> let $x$ be the root of the tree $T$ that lends its minkey to $B_{i}$; the ckey of $x$ is smaller \\
	\> than all keys in the pnodes and all ckeys; $B_{i}$ holds elements from a listnode $l$ of $x$;\\
	\> let $e$ be an element from $l$;  delete $e$ from $l$; \\
\\
if $B_{i}$ is empty then \\
	\> fill it up with an invocation to {\tt Fill-Up}(), and update $B_i$'s minkey value; \\
\\
update the suffixmin pointers of buckets $B_{i}, \ldots, B_{0}$; \\ 
return $e$; 
\end{tabbing}
\vspace{0.1in}
\hrule
\vspace{0.2in}
\caption{Deletemin}
\label{fig:deletemin}
\end{figure}

Note that if after a {\tt Deletemin}, a root fails to satisfy the relevant invariant,
then a {\tt Sift} is not called immediately. We wait till the next {\tt Fill-Up}.
(While deletions happen in buckets, they are counted against the root from which
the deleted elements were taken. Therefore, a deletion can cause the corresponding root
to fail the relevant invariant.)
 
Recall that we keep at most a block of $B[i]$'s elements in the main memory.
When all elements of the block are deleted by {\tt Deletemin}s, the next block is
	brought in.

\subsubsection{{\tt Findmin}}
A {\tt Findmin} return the same element that a {\tt Deletemin} would.
But the element is not deleted from the EMSH. 
Therefore, a {\tt Findmin} does not need to perform any of the updations
that a {\tt Deletemin} has to perform on the data structure.

As it is an in-core operation, a {\tt Findmin} does not incur any I/O. 

\subsubsection{{\tt Meld}}

In the {\tt Meld} operation, two heaps $H_1$ and $H_2$ are to be merged
        into a new heap $H$.
It is assumed that the buckets of the two heaps remain in the main memory.

Combine the input buffers of $H_1$ and $H_2$.
If the total number of elements exceeds $B\sqrt{m}$, then
create a new node $x$ of rank $0$, move $B\sqrt{m}$ elements from the buffer into it
leaving the rest behind, create a tree $T$ of rank $0$ with $x$ as its only node, 
create a bucket $B_0'$, set the number of trees in it to $1$, and include a pointer to $T$ in it.

Let $B_{1,i}$ (resp., $B_{2,i}$) be the $i$-th bucket of $H_1$ (resp., $H_2$).
Let max denote the largest rank in the two heaps $H_1$ and $H_2$. The {\tt Meld} is analogous to
the summation of two $\sqrt{m}$-radix numbers of $\max$ digits. At position $i$,
buckets $B_{1,i}$ and $B_{2,i}$ are the ``digits''; there could also be
a ``carry-in'' bucket $B'_{i}$. The ``summing'' at position $i$ produces a new
$B_{1,i}$ and a ``carry-out'' $B'_{i+1}$. 
$B_0'$ will function as the ``carry in'' for position $0$.

The {\tt Meld} proceeds as in Figure~\ref{fig:meld}.

\begin{figure}
\vspace{0.2in}
\hrule
\begin{tabbing}
aaaa \= aaaa \= aaaa \= aaaa \= aaaa \= aaaa \kill \\
for $i=0$ to max$+1$ \\
begin \\
\> if only one of $B_{1,i}$, $B_{2,i}$ and $B'_{i}$ exists then \\
\>\>that bucket becomes $B_{1,i}$; {\tt Fill-Up} that bucket, if necessary; \\
\>\>there is no carry-out; \\
\> else \\
\>\>if $i\leq r$ \\
\>\>\>if $B_{1,i}$ (resp., $B_{2,i}$) contains elements then \\
\>\>\>\> send the elements of $B_{1,i}$ (resp., $B_{2,i}$) back to the \\
\>\>\>\>roots from which they were taken; \\ 
\>\>\>for each root $x$ pointed by $B_{1,i}$ or $B_{2,i}$ \\ 
\>\>\>\>if $x$ does not satisfy PNI, invoke {\tt Sift}$(x)$; \\ 
\>\>else \\ 
\>\>\>if $B_{1,i}$ (resp., $B_{2,i}$) and the last listnode $l_1$ (resp., $l_2$) \\
\>\>\>of the root $x_1$ (resp., $x_2$) with the smallest ckey in $B_{1,i}$ (resp., $B_{2,i}$) \\
\>\>\>have sizes $<B\sqrt{m}/2$ each, then \\ 
\>\>\>\>merge the elements in $B_{1,i}$ (resp., $B_{2,i}$) into $l_1$ (resp., $l_2$); \\ 
\>\>\>else \\ 
\>\>\>\>store the elements in $B_{1,i}$ (resp., $B_{2,i}$) in a new listnode $l$\\ 
\>\>\>\>and insert $l$ into the list of $x_1$ (resp., $x_2$) so that \\
\>\>\>\>all but the last listnode will have $\geq B\sqrt{m}/2$ elements; \\ 
\>\>\>\>if $x_1$ (resp., $x_2$) does not satisfy CNI, then {\tt Sift} it; \\ 
\>if the total number of root-pointers in $B_{1,i}$, $B_{2,i}$ and $B'_{i}$ is $<\sqrt{m}$, then  \\
\>\>move all root-pointers to $B_{1,i}$; {\tt Fill-Up} $B_{1,i}$; \\
\>\>delete $B_{2,i}$ and $B'_{i}$; There is no carry-out; \\
\end{tabbing}
\hfill {CONTINUED}
\label{fig:meld1}
\end{figure}

\begin{figure}
\begin{tabbing}
aaaa \= aaaa \= aaaa \= aaaa \= aaaa \= aaaa \kill \\
\>else \\
\>\>create a tree-node $x$ of rank $i+1$; \\
\>\>pool the root-pointers in $B_{1,i}$, $B_{2,i}$ and $B'_{i}$; \\
\>\>take $\sqrt{m}$ of those roots and make them children of $x$; {\tt Sift}$(x)$; \\
\>\>create a carry-out bucket $B'_{i+1}$; \\ 
\>\>place in it a pointer to $x$; this is to be the only root-pointer of $B'_{i+1}$; \\ 
\>\>move the remaining root-pointers into $B_{1,i}$; {\tt Fill-Up} $B_{1,i}$; \\
\>\>delete $B_{2,i}$ and $B'_{i}$; \\
end; \\
update the suffixmin pointers; \\
\end{tabbing}
\vspace{0.1in}
\hrule
\vspace{0.2in}
\caption{Meld}
\label{fig:meld}
\end{figure}

\subsubsection{{\tt Sift}}

The {\tt Sift} operation is invoked only on non-leaf nodes that fail to satisfy 
	PNI or CNI, whichever is relevant.
When the invocation returns, the node will satisfy the invariant.
We shall use in the below a procedure called
	extract that is to be invoked only on cnodes of rank $r+1$, 
	and pnodes, and is defined in Figure~\ref{fig:extract}.

\begin{figure}[ht]
\vspace{0.2in}
\hrule
\begin{tabbing}
aaaa \= aaaa \= aaaa \= aaaa \= aaaa \= aaaa \kill \\
extract($x$) \\
begin \\
\> let $N_x$ be the total number of elements in all the children of $x$ \\
\> put together; extract the smallest $\min\{B\sqrt{m}/2,N_x\}$ of those \\
\> elements and store them in $x$; \\
end
\end{tabbing}
\vspace{0.1in}
\hrule
\vspace{0.2in}
\caption{Extract}
\label{fig:extract}
\end{figure}

Suppose {\tt Sift} is invoked on a node $x$. 
This invocation could be recursive, or from {\tt Meld} or {\tt Fill-Up}.
{\tt Meld} and {\tt Fill-Up} invoke {\tt Sift} only on roots.
Recursive invocations of {\tt Sift} proceed top-down; thus, any recursive
	invocation of {\tt Sift} on $x$ must be from the parent of $x$. 
Also, as can be seen from the below, 
	as soon as a non-root fails its relevant invariant (PNI or CNI),
	{\tt Sift} is invoked on it.
Therefore, at the beginning of a {\tt Sift} on $x$, each child of $x$ must
	satisfy PNI or CNI, as is relevant.

{\bf If $x$ is a pnode} (and thus, PNI is the invariant violated), then
	$x$ contains less than $B\sqrt{m}/2$ elements.
Each child of $x$ satisfies PNI, and therefore has, unless it is a leaf,  
	at least $B\sqrt{m}/2$ elements.  
Invoke extract$(x)$.
This can be done in $O(\sqrt{m})$ I/Os by performing a $\sqrt{m}$-way merge of $x$'s children's arrays. 
For each non-leaf child $y$ of $x$ that now violates PNI, recursively invoke {\tt Sift}$(y)$.
Now the size of $x$ is in the range $[B\sqrt{m}/2, B\sqrt{m}]$, unless 
all of $x$'s children are empty leaves. 

{\bf If $x$ is a cnode of rank $r+1$}, then CNI is the invariant violated.
The children of $x$ are of rank $r$, and are thus pnodes.
There are two possibilities: (A)~This {\tt Sift} was invoked from a {\tt Fill-Up} or {\tt Meld}, and thus
	$x$ has one listnode $l$ left in it. (B)~This {\tt Sift} was invoked recursively, and thus
	$x$ has no listnode left in it.
In either case, to begin with, invoke extract$(x)$, and invoke {\tt Sift}$(y)$ 
	for each non-leaf child $y$ of $x$ that now violates PNI.
The number of elements gathered in $x$ is $B\sqrt{m}/2$, 
	unless all of $x$'s children are empty leaves. 
 
Suppose case (A) holds. 
Create a new listnode $l'$, and store in $l'$ the elements just extracted into $x$.
If $l'$ has a size of $B\sqrt{m}/2$, insert $l'$ at the front of $x$'s list;
	else if $l$ and $l'$ together have at most $B\sqrt{m}/2$ elements, then
	merge $l'$ into $l$;
	else, append $l'$ at the end of the list, and transfer enough elements
	from $l'$ to $l$ so that $l$ has a size of $B\sqrt{m}/2$.

If case (B) holds, then if $x$ has nonempty children, once again, extract$(x)$, and invoke {\tt Sift}$(y)$ 
	for each non-leaf child $y$ of $x$ that now violates PNI.
The total number of elements gathered in $x$ now is $B\sqrt{m}$, 
	unless all of $x$'s children are empty leaves. 
If the number of elements gathered is at most $B\sqrt{m}/2$,
	then create a listnode, store the elements in it, and make it the sole member
	of $x$'s list;
otherwise, create two listnodes, insert them in the list of $x$, store $B\sqrt{m}/2$ elements in the first, 
	and the rest in the second.
 
In both the cases, update the ckey of $x$ so that it will be the largest of all keys now present in $x$'s list.
	
{\bf If $x$ is a cnode of rank greater than $r+1$}, then
	while the size of $x$ is less than $s_{k}$, and not all children of $x$ hold
	empty lists, do the following repeatedly: 
        (i)~pick the child $y$ of $x$ with the smallest ckey, (ii)~remove the last listnode 
	of $x$ and merge it with the last listnode $y$, if they together have at most
	$B\sqrt{m}$ elements, (iii)~merge the resultant list of $y$ to the resultant list of $x$ such that all but the last listnode will have al least $B \sqrt{m}/2$ elements,
	(iv)~set the ckey of $x$ to the ckey of $y$, and (v)~invoke {\tt Sift}$(y)$ recursively.
If merging is not required, then the concatenation merely updates $O(1)$ pointers.
Merging, when it is needed, incurs $O(\sqrt{m})$ I/Os.

The {\tt Sift} operation removes all leaves it renders empty.
An internal node becomes a leaf, when all its children are removed.

\subsubsection{{\tt Fill-Up}}
The {\tt Fill-Up} operation is invoked by {\tt Deletemin} and {\tt Meld} 
	on a bucket $B_{i}$ when those operations find $B_i$ empty.
$B_i$ is filled up using the procedure given in Figure~\ref{fig:fillup}.

\begin{figure}
\vspace{0.2in}
\hrule
\begin{tabbing}
aaaa \= aaaa \= aaaa \= aaaa \= aaaa \= aaaa \kill \\
if $i\leq r$ then \\
	\> for each root $x$ in $B_i$ that does not satisfy PNI \\
	\> 	\> {\tt Sift}$(x)$; \\
	\> Let $N_i$ be the total number of elements in all the roots of $B_i$ put together; \\
	\> extract the smallest $\min\{B\sqrt{m}/2,N_i\}$ of those and store them in $B_i$; \\
else \\
	\> for each root $x$ in $B_i$ that does not satisfy CNI \\
	\>	\> {\tt Sift}$(x)$; \\
	\> pick the root $y$ with the smallest ckey in $B_i$; \\
        \> copy the contents of one $l$ of $y$'s listnodes (not the last one) into $B_i$; \\
        \> remove $l$ from the list of $y$.
\end{tabbing}
\vspace{0.1in}
\hrule
\vspace{0.2in}
\caption{Fill-Up}
\label{fig:fillup}
\end{figure}

A bucket remembers, for each element $e$ in it, the root from which $e$ was extracted. 
This is useful when the {\tt Meld} operation sends the elements in the bucket back to their
respective nodes. 

Even if a {\tt Fill-Up} moves all elements of a root $x$ without children into the bucket,
	$x$ is retained	until all its elements are deleted from the bucket. 
(A minor point: For $i\leq r$, if the roots in the bucket all have sent up all their
	elements into the bucket, are without children, and have at most
	$B\sqrt{m}/2$ elements together, then all of them except one can be deleted at the time
	of the {\tt Fill-Up}.)

\subsubsection{The Memory Requirement}
\label{memory:requirement}
The following lemma establishes the largest rank that can be present in a 
	heap.
\begin{lemma}
\label{lem:nodesofarank}
There are at most $\frac{N}{B\sqrt{m}^{k+1}}$ tree-nodes of rank $k$ when $N$ elements have been 
inserted into it.
\end{lemma}
{\bf Proof:}
We prove this by induction on $k$.
The basis is provided by the rank-$0$ nodes.
A node of rank $0$ is created when $B\sqrt{m}$ new elements have been accumulated in 
	the main memory buffer. 
Since the total number of elements inserted in the heap is $N$, 
	the total number of nodes of rank $0$ is at most $N/B\sqrt{m}$. 
Inductively hypothesise that the lemma is true for tree-nodes of
	rank at most $(k-1)$. 
Since a node of rank $k$ is generated when
	$\sqrt{m}$ root nodes of rank $k-1$ are combined, the number of 
	nodes of rank $k$ is at most
	$\frac{N}{B\sqrt{m}^{k}\sqrt{m}}=\frac{N}{B\sqrt{m}^{k+1}}$.
\hfill ${\Box}$

Therefore, if there is at least one node of rank $k$ in the heap, then 
$\frac{N}{B\sqrt{m}^{k+1}} \geq 1$, and so $k\leq \log_{\sqrt{m}}\frac{N}{B}$.
Thus, the rank of the EMSH is at most $\log_{\sqrt{m}} \frac{N}{B}$.
Note that there can be at most $\sqrt{m}-1$ trees of the same rank.

The main memory space required for a bucket is $O(B+\sqrt{m})$.
So, the total space required for all the buckets is
        $O((B+\sqrt{m}) \log_{\sqrt{m}} \frac{N}{B})$.
We can store all buckets in main memory, if we assume that
        $(B+\sqrt{m}) \log_{\sqrt{m}} \frac{N}{B} = O(M)$.
This assumption is valid for all values of $N=O(Bm^{M/2(B+\sqrt{m})})$.
Assume the modest values for $M$ and $B$ given in \cite{MR99}: say, a block is of size $1$ KB, and the main
	memory is of size $1$ MB, and can contain $B=50$ and $M=50000$ records respectively.
Then, if $N<10^{900}$, which is practically always, the buckets will all fit in the main memory.

\section{A Proof of Correctness}
\label{correct:emsh}
If the heap order property is satisfied at every node in the EMSH before an 
	invocation of {\tt Sift}$(x)$, then it will be satisfied after the
	invocation returns too.
This can be shown as follows.

If $x$ is a pnode, then the invocation causes a series of {\tt Extract}s, each
	of which moves up into a node a set of smallest elements in its children;
	none of them can cause a violation of the heap order property.

If $x$ is a cnode of rank $r+1$, then a set of smallest elements at $x$'s children
	move up into $x$ and become corrupt. 
All these elements have key values greater than $k'$, the ckey of $x$ prior to the {\tt Sift}.
The new ckey of $x$ is set to the largest key $k$ among the elements moving in. 
Thus, $k$ is smaller than each key in each of $x$'s
	children after the invocation; and $k>k'$. 

If $x$ is a cnode of rank greater than $r+1$, then a set of corrupt elements move
	into $x$ from $y$, the child of $x$ with the smallest ckey, and the ckey
	of $x$ is set to the ckey of $y$.
Inductively assume that the ckey of $y$ is increased by the recursive {\tt Sift} on $y$.  
Therefore, at the end of the {\tt Sift} on $x$, the ckey of $x$ is smaller than the
	ckey of each of $x$'s children.

In every other operation of the EMSH, all data movements between nodes 
	are achieved through {\tt Sift}s.
Thus, they too cannot violate the heap order property.

When we note that a {\tt Fill-Up} on a bucket $B_i$ moves into it
	a set elements with smallest keys or the smallest ckey from
	its roots, and that 
	the suffixmin pointer of $B_{0}$ points to the bucket with the smallest minkey
	among $\{B_{x}\ |\ x\geq 0\}$, we have the following Lemma. 

\begin{lemma}
If there is no cnode in the EMSH, then the element returned by {\tt Deletemin}
will be the one with the smallest key in the EMSH.
If there are cnodes, and if the returned element is corrupt (respectively, not corrupt), 
then its ckey (respectively, key) will be the smallest
of all keys in the pnodes and ckeys of the EMSH.
\end{lemma}

For all $k>r$, and for every nonleaf $x$ of rank $k$ that satisfies CNI, 
the size of the list in $x$ is at least $\lfloor s_k/2 \rfloor + 1$.
For a root $x$ of rank $k>r$, when the size of its list falls below $\lfloor s_k /2
	\rfloor + 1$, {\tt Sift}$(x)$ is not invoked until at least the next invocation
	of {\tt Fill-Up}, {\tt Meld} or {\tt Deletemin}.

The following lemma gives an upperbound on the size of the list.

\begin{lemma}
\label{lem:3sk}
For all $k>r$, and for every node $x$ of rank $k$, 
the size of the list in $x$ is at most $3 s_k$.
\end{lemma}
{\bf proof:}
We prove this by an induction on $k$.
Note that between one {\tt Sift} and another on a node $x$, the list of $x$ can
lose elements, but never gain.  

A {\tt Sift} on a node of rank $r+1$ causes it to have a list of size at most two;
$2\leq 3s_{r+1}=6$; this forms the basis. 

Let $x$ be a node of rank $ > r + 1$. Hypothesise that the upperbound
holds for all nodes of smaller ranks. When {\tt Sift}$(x)$ is called,
repeatedly, a child of $x$ gives $x$ a list $L'$ that is then added to the list $L$ of $x$,
until the size of $L$ becomes at least $s_{k}$ or $x$ becomes a leaf.
The size of each $L'$ is, by the hypothesis, at most $3 s_{k-1} \leq
2\lceil\frac{3}{2}s_{k-1}\rceil =2s_{k}$.


The size of $L$ is at most $s_k - 1$ before the last iteration.
Therefore, its size afterwards can be at most $3 s_k -1 < 3 s_k$. 
\hfill $\Box$

\begin{lemma}
\label{lem:skbounds}
For all values of $k > r$,
\begin{equation*}
\left(\frac{3}{2}\right)^{k - r -1} \leq s_k \leq 2 \left(\frac{3}{2}\right)^{k - r} - 1
\end{equation*}
\end{lemma}

{\bf Proof:}
A simple induction proves the lowerbound. 
Basis: $s_{r+1}=2\geq\left(\frac{3}{2}\right)^{0}=1$.
Step: For all $k>r+1$, $s_{k}=\lceil \frac{3}{2} s_{k-1} \rceil\geq\frac{3}{2} s_{k-1}\geq\left(\frac{3}{2}\right)^{k - r -1}$.

Similarly, a simple induction shows that, 
for all values of $k \geq r+4$, $s_k \leq 2 \left(\frac{3}{2}\right)^{k - r} - 2$.
Basis: $s_{r+4}=8\leq 2\left(\frac{3}{2}\right)^{4} - 2=8.125$.
Step: $s_{k}=\lceil \frac{3}{2} s_{k-1} \rceil\leq\frac{3}{2} s_{k-1}+1\leq
	\frac{3}{2}\left[2\left(\frac{3}{2}\right)^{k - r -1}-2\right]+1=2 \left(\frac{3}{2}\right)^{k - r} - 2$.
Note that
$s_{r+1}=2=2\left(\frac{3}{2}\right)-1$,
$s_{r+2}=3<2\left(\frac{3}{2}\right)^{2}-1$, and
$s_{r+3}=5<2\left(\frac{3}{2}\right)^{3}-1$,
Therefore, for all values of $k > r$, $s_k \leq 2 \left(\frac{3}{2}\right)^{k - r} - 1$.
 \hfill $\Box$

\begin{lemma}
If $m>110$, at any time there are at most $\epsilon N$ corrupt elements in the EMSH,
where $N$ is the total number of insertions performed.
\end{lemma}
{\bf Proof:}
All corrupt elements are stored in nodes of rank greater than $r$.
The size of the list of a node of rank $k>r$ is
    at most $3s_k$ by Lemma~\ref{lem:3sk}.
Each listnode contains at most $B\sqrt{m}$ corrupt elements.
Thus, the total number of corrupt elements at a node of rank $k>r$
	is at most $3s_kB\sqrt{m}$.
Suppose $m>110$. Then $\sqrt{m}>10.5$.

As $r = log_{\sqrt{m}}\frac{1}{\epsilon}$, 
by Lemma~\ref{lem:skbounds} and Lemma~\ref{lem:nodesofarank},
the total number of
corrupt elements are at most
\begin{equation*}
\begin{split}
\sum_{k>r}(3s_kB\sqrt{m})\frac{N}{B(\sqrt{m})^{k+1}} & =
\frac{N}{(\sqrt{m})^r}  \sum_{k>r}\frac{3  s_k}{(\sqrt{m})^{k-r}} \\
 & \leq \frac{N}{(\sqrt{m})^r} \sum_{k>r}\frac{3  (2  (3/2)^{k-r} - 1)}{(\sqrt{m})^{k-r}} \\
 & \leq \frac{N}{(\sqrt{m})^r} \sum_{k>r}\frac{6 
 (3/2)^{k-r}}{(\sqrt{m})^{k-r}} \\
 & \leq \frac{N}{(\sqrt{m})^r}  \frac{9}{(\sqrt{m}- 1.5)}\\
 & < \frac{N}{(\sqrt{m})^r} = \epsilon N
\end{split}
\end{equation*}
\hfill $\Box$

\section{An Amortised I/O Analysis}
\label{analy:emsh}

Suppose charges of colours green, red and yellow remain distributed over the
data structure as follows: (i)~each root carries $4$  green charges,
(iii)~each bucket carries $2\sqrt{m}$ green charges,
(ii)~each element carries $1/B$ red charges,
(iv)~each nonleaf node carries one yellow charge, and
(v)~each leaf carries $\sqrt{m}+1$ yellow charges.

In addition to these, each element also carries a number of blue charges.
The amount of blue charges that an element carries can vary with its position
in the data structure.

The amortised cost of each operation is its actual cost plus
the total increase in all types charges caused by it.
Now we analyse each operation for its amortised cost.

{\tt Insert}:
Inserts actually cost $\Theta(\sqrt{m})$ I/Os when the insert buffer in the main memory runs full,
which happens at intervals of $\Omega(B\sqrt{m})$ inserts.
Note that some {\tt Deletemin}s return elements in the insert buffer.
If an {\tt Insert} causes the buffer to become full, 
then it creates a new node $x$, a tree $T$ with $x$ as its only node, and
a heap $H'$ with $T$ as its only tree. The $B\sqrt{m}$ elements in the buffer
are copied into $x$. Moreover, a bucket is created for $H'$.
New charges are created and placed on all new entities.
Thus, this {\tt Insert} creates $4+2\sqrt{m}$ green charges, $\sqrt{m}$ red charges,
and $\sqrt{m}+1$ yellow charges. Suppose it also places $\Theta(r/B)$ blue charges
on each element of $x$. That is a total of 
$\Theta(r\sqrt{m})$ blue charges on $x$. 
That is, the total increase in the charges of the system is
$\Theta(r\sqrt{m})$. 
It follows that the amortised cost of a single {\tt Insert} is
$O(r/B)$. 

{\tt Meld}:
The buckets are processed for positions $0$ to max$+1$ in that order, where
max is the largest rank in the two heaps melded. The process is
analogous to the addition of two $\sqrt{m}$-radix numbers.
At the $i$-th position, at most three buckets are to be handled:
$B_{1,i}$, $B_{2,i}$ and the ``carry-in'' $B'_{i}$. 
Assume inductively that each bucket holds
$2\sqrt{m}$ green charges. 

If only one of the three buckets is present at
position $i$, then there is no I/O to perform, no charge is released, and, therefore,
the amortised cost at position $i$ is zero. 

If at least two of the three are present,
then the actual cost of the operations at position $i$ is $O(\sqrt{m})$.
As only one bucket will be left at position $i$, at least one bucket is
deleted, and so at least $2\sqrt{m}$ green charges are freed. Suppose, $\sqrt{m}$ of this 
pays for the work done. If there is no ``carry-out'' $B'_{i+1}$ to be formed,
then the amortised cost at position $i$ is negative.

If $B'_{i+1}$ is to be formed, then
place $\sqrt{m}$ of the remaining green charges
on it. When $B'_{i+1}$ is formed, $\sqrt{m}$ roots hook
up to a new node; these roots cease to be roots, and so together give 
up $4\sqrt{m}$ green charges; four of that will be placed on the new root; 
$4\sqrt{m}-4$ remain; $4\sqrt{m}-4\geq \sqrt{m}$, as $m\geq 2$. 
So we have an extra of $\sqrt{m}$ green charges to put on the carry-out
which, with that addition, holds $2\sqrt{m}$ green charges. No charge of other colours is freed.

The amortised cost is non-positive at each position $i$. 
So the total amortised cost of {\tt Meld} is also non-positive.

{\tt Deletemin}:
A typical {\tt Deletemin} is serviced from the main memory, and does
not cause an I/O. Occasionally, however, an invocation
to {\tt Fill-Up} becomes necessary. The actual I/O cost
of such an invocation is $O(\sqrt{m})$. A {\tt Fill-Up} is triggered
in a bucket when $\Theta(B\sqrt{m})$ elements are deleted from it.
At most a block of the bucket's elements are kept in the main memory.
Thus, a block will have to be fetched into the main memory once in every
$B$ {\tt Deletemin}s. 
The red charges of deleted items can pay for the cost of the all these I/Os.  
The amortised cost of a {\tt Deletemin} is, therefore, at most zero.

{\tt Findmin}: As no I/O is performed, and no charge is released,
the amortised cost is zero.

{\tt Fill-Up}:
This operation is invoked only from {\tt Meld} or {\tt Deletemin}.
The costs have been accounted for in those.

{\tt Sift}:
Consider a {\tt Sift} on a node $x$ of rank $\leq r+1$.
This performs one or two Extracts. The actual cost of the Extracts is $O(\sqrt{m})$.
If the number of extracted elements is $\Theta(B\sqrt{m})$,
	then each extracted element can contribute $1/B$ blue charges 
	to pay off the actual cost.
If the number of extracted elements is $o(B\sqrt{m})$, then
$x$ has become a leaf after the {\tt Sift}.
Therefore, all of $x$'s children were leaves before the {\tt Sift}. 
One of them can pay for the {\tt Sift} with its $\sqrt{m}$ yellow charges;
	at least one remained at the time of the {\tt Sift}.
Node $x$ that has just become a leaf, has lost the $\sqrt{m}$ children 
	it once had. 
If one yellow charge from each child has been preserved in $x$, then
	$x$ now holds $\sqrt{m}+1$ yellow charges, enough for a leaf.

Consider a {\tt Sift} on a node $x$ of rank $i > r+1$.
A number of iterations are performed, each of which costs $O(\sqrt{m})$ I/Os.
In each iteration, a number of elements move from a node of rank 
$i-1$ (namely, the child $y$ of $x$ with the smallest ckey) to a node of rank $i$
(namely, $x$). If the number of elements moved is $\Omega(B\sqrt{m})$, then
the cost of the iteration can be charged to the elements moved. 
Suppose each element moved contributes $\frac{1}{s_{i-1}B}$ blue charges.
Since the list of $y$ has at least $\lfloor s_{i-1}/2 + 1\rfloor$ listnodes,
in which all but the last have at least $B\sqrt{m}/2$ elements, the total number of
blue charges contributed is at least 
$\frac{1}{s_{i-1}B}\lfloor \frac{s_{i-1}}{2}\rfloor \frac{B\sqrt{m}}{2} =\Theta{(\sqrt{m})}$. 
Thus, the cost of the
	iteration is paid off.

If the number of elements moved is $o(B\sqrt{m})$, then $y$ was a leaf prior
to the {\tt Sift}, and so can pay for the {\tt Sift} with its $\sqrt{m}$ yellow charges.
If $x$ becomes a leaf at the end of the {\tt Sift}, it will have $\sqrt{m}+1$ yellow charges 
on it, as one yellow charge from each deleted child is preserved in $x$.

An element sheds $1/B$ blue charges for each level it climbs up, for the first
$r+1$ levels. After that when it moves up from level $i-1$ to $i$, it
sheds $\frac{1}{s_{i-1}B}$ blue charges. Therefore, with 
\[\frac{r+1}{B}+\sum_{i>r+1}\frac{1}{s_{i-1}B}=\frac{r+1}{B}+\sum_{i>r}\frac{1}{B}
 \left(\frac{2}{3}\right)^{i - r -1}=\Theta{\left(\frac{r}{B}\right)}\]
blue charges initially placed on the element, it can pay for its travel upwards. 

Thus, we have the following lemma.
\begin{lemma}
In EMSH, the amortised complexity of an {\tt Insert} is $O(\frac{1}{B}\log_{m} \frac{1}{\epsilon})$. {\tt Findmin}, {\tt Deletemin} and {\tt Meld} all have non-positive amortised complexity.
\end{lemma}

\section{Hard heap: A Meldable Priority Queue}
\label{hard:emsh}
When an EMSH has error-rate $\epsilon=1/(N+1)$, no element in it can be corrupt.
In this case, EMSH becomes an exact priority queue, which we call hard heap.
In it every node is a pnode, and every tree is a ptree. 
{\tt Deletemin}s always report the exact minimum in the hard heap. 
The height of each tree is $O(\log_{m}\frac{N}{B})$, as before.
But, since all nodes are pnodes, the amortised cost of an insertion is
$O(\frac{1}{B}\log_{m}\frac{N}{B})$ I/Os. The amortised costs of all 
other operations remain unchanged.

The absence of corrupt nodes will also permit us to implement a {\tt Delete}
operation: To delete the element with key value $k$, insert a ``Delete'' record
with key value $k$. Eventually, when $k$ is the smallest key value in the hard heap,
a {\tt Deletemin} will cause
the element with key $k$ and the ``Delete'' record to come up together. Then the
two can annihilate each other. The amortised cost of a {\tt Delete} is
$O(\frac{1}{B}\log_{m}\frac{N}{B})$ I/Os, the same as that of an {\tt Insert}.
 
None of the known external memory priority queues (EMPQs) \cite{Arge03,FJKT99,KS96},
support a {\tt meld} operation. However, in all of them, two queues could be {\tt meld}ed by 
inserting elements of the smaller queue into the larger queue one by one.
This is expensive if the two queues have approximately the same size.
The cost of this is not factored into the amortised complexities of those
EMPQs. 

The actual cost of a {\tt meld} of two hard heap's with $N$ elements each is
$O(\sqrt{m}\log_{m}\frac{N}{B})$ I/Os; the amortised cost of the {\tt meld} is subzero.
But this is the case only if the buckets of
both the heaps are in the main memory. Going by our earlier analysis in Section~\ref{memory:requirement},
if $N=O(Bm^{M/2k(B+\sqrt{m})})$ then $k$ heaps of size $N$ each can keep their buckets in the main memory.

The buckets of the heaps to be melded could be kept in the secondary memory, 
and brought into the main memory, and written back either side of the meld. 
The cost of this can be accounted
by an appropriate scaling of the amortised complexities.
However, operations other than {\tt meld} can be performed only if the buckets are 
in the main memory.
 
In hard heap, unlike in the other EMPQs, 
elements move only in the upward direction. 
That makes hard heap easier to implement.
Hard heap and the external memory heap of \cite{FJKT99} do not require any extra space
other than is necessary for the elements.
The other EMPQs \cite{Arge03,KS96} use extra space (See Table~\ref{comparison:emsh}).
  
The buffer tree \cite{Arge03} is a B+ tree, and therefore uses a balancing procedure.
However, because of delayed deletions, its height may depend on the number of pending
deletions, as well as the number of elements left in it.
The external memory heap of \cite{FJKT99} is a balanced heap, and therefore,
incurs a balancing cost. But, in it the number of
I/Os performed by a sequence of $S$ operations is
$O(\frac{S}{B}+\frac{1}{B}\sum_{i = 1}^{S} \log_{m}\frac{N_i}{B})$, where $N_i$ is the number
of elements remaining in the heap before the $i$-th operation; this is helpful
when the {\tt insert}s and {\tt deletemin}s are intermixed so that the number
of elements remaining in the data structure at any time is small.  

In comparison, hard heap is not a balanced tree data structure. It does not use a costly
balancing procedure like the heaps of \cite{FJKT99,KS96}.
However, for a sequence of $N$ operations, 
	the amortised cost of each operation is $O(\frac{1}{B} \log_m \frac{N}{B})$ I/Os.

We can make the amortised cost depend on $N_i$, the number
of elements remaining in the hard heap before the $i$-th operation, at the cost
of adding a balancing procedure.
In a sequence of operations, whenever $N_I$, the number of {\tt insert}s performed, and $N_D$, 
	the number of {\tt deletemin}s performed, satisfy $N_I - N_D < N_I / \sqrt{m}$,
	and the height of the hard heap is $\log_{\sqrt{m}} \frac{N_I}{B}$, 
	delete the remaining $N_I - N_D$ elements from the hard heap,  insert them back in,
	and set the counts $N_I$ and $N_D$ to $N_I-N_D$ and zero respectively;
	we can think of this as the end of an {\em epoch} and the beginning of the next in
	the life of the hard heap. 
The sequence of operations, thus, is a concatenation of several epochs.
Perform an amortised analysis of each epoch independently. 
The cost of reinsertions can be charged to the elements actually deleted in the previous
	epoch, thereby multiplying the amortised cost
	by a factor of $O(1+\frac{1}{\sqrt{m}})$.
It is easy to see that now the number of
I/Os performed by a sequence of $S$ operations is
$O(\frac{S}{B}+\frac{1}{B}\sum_{i = 1}^{S} \log_{m}\frac{N_i}{B})$.

\subsection{Heap Sort}
We now discuss an implementation of Heap Sort using hard heap, and
count the number of comparisons performed.
To sort, insert the $N$ input elements into an initially empty hard heap, and 
then perform $N$ {\tt deletemin}s.

When a node of rank $0$ is created, $O(B\sqrt{m} \log_2 (B\sqrt{m}))$ comparisons are performed;
that is $O(\log_2 (B\sqrt{m}))$ comparisons per element involved.
When an elements moves from a node to its parent, it participates in a $\sqrt{m}$-way merge;
	a $\sqrt{m}$-way merge that outputs $k$ elements requires to perform only
	$O(k\log_2 \sqrt{m})$ comparisons; that is $O(\log_2 \sqrt{m})$ comparisons per element involved.
Since the number of levels in the hard heap is at most $\log_{\sqrt{m}} N/B\sqrt{m}$, the total
	number of comparisons performed by one element is $\log_2 N$.
Each {\tt deletemin} operation can cause at most $\log_{\sqrt{m}} (N/\sqrt{mB})$ comparisons
	among the suffixmin pointers.
Thus, the total number of comparisons is $O(N \log_2 N)$. 
\section{Applications of EMSH}
\label{appl:emsh} 
The external memory soft heap data structure is useful for finding exact and 
	approximate medians, and for approximate sorting \cite{Ch00a}.
Each of these computations take $O(N/B)$ I/Os:
\begin{enumerate}
\item To compute the median in a set of $N$ numbers, insert the numbers
        in an EMSH with error rate $\epsilon$.
Next, perform $\epsilon N$ {\tt Deletemin}s. The largest number $e$ deleted
        has a rank between $\epsilon N$ and $2\epsilon N$.
Partition the set using $e$ as the pivot in $O(N/B)$ I/Os.
We can now recurse with a partition of size at most $\max\{\epsilon, (1-2\epsilon)\}N$.
The median can be found in $O(N/B)$ I/Os.
This is an alternative to the algorithm in \cite{Si02} which also requires $O(N/B)$ I/Os.
\item To approximately sort $N$ items, insert them into an EMSH with error rate $\epsilon$, and
        perform $N$ {\tt Deletemin}s consecutively. Each element can form an inversion with
        at most $\epsilon N$ of the items remaining in the EMSH at the time of its deletion.
        The output sequence, therefore, has at most $\epsilon N^2$ inversions.
We can also use EMSH to near sort $N$ numbers in $O(N/B)$ I/Os such that
        the rank of each number in the output sequence differs from
        its true rank by at most $\epsilon N$; the in-core algorithm given in \cite{Ch00a} suffices.
\end{enumerate}

  \chapter{The Minimum Cut Problem}
\label{cut:chapt}
\section{Introduction}

The minimum cut problem on an undirected unweighted graph is to partition the vertices into two 
	sets while minimizing the number of edges from one side of the partition to the other.
This is an important combinatorial optimisation problem. 
Efficient in-core and parallel algorithms for the problem are known. 
For a recent survey see \cite{Br07, Ka00, KaMo97, NaIb08}.
This problem has not been explored much from the perspective of massive data sets.
However, it is shown in \cite{AgDaRa+04, R03} that the minimum cut can be computed in 
	a polylogarithmic number of passes using only a polylogarithmic sized main
	memory on the streaming and sort model.

In this chapter we design an external memory algorithm for the 
	problem on an undirected unweighted graph.
We further use this algorithm for computing 
	a data structure which represents all cuts of size at most $\alpha$ times the
	size of the minimum cut, where $\alpha < 3/2$. 
The data structure answers queries of the following form: A cut $X$ (defined by a vertex
	partition) is given; find whether $X$ is of size at most
	$\alpha$ times the size of the minimum cut.
We also propose a randomised algorithm that is based on our deterministic algorithm, and 
	improves the I/O complexity.
An approximate minimum cut algorithm which performs fewer I/Os than 
	our exact algorithm is also presented.

\subsection{Definitions}
For an undirected unweighted graph $G=(V,E)$, a 
	cut $X=(S,V - S)$ is defined as a partition of the vertices
	of the graph into two nonempty sets $S$ and $V - S$.
An edge with one endpoint in $S$ and the other endpoint in $(V-S)$ is called a crossing edge of $X$.
The value $c$ of the cut $X$ is the total number of crossing edges of $X$.
	
The minimum cut (mincut) problem is to find a cut of minimum value. 
On unweighted graphs, the minimum cut problem is sometimes referred to as
	the edge-connectivity problem.
We assume that the input graph is connected, since otherwise the problem is
trivial and can be computed by any connected components algorithm.
A cut in $G$ is $\alpha$-minimum, for $\alpha > 0$,
     if its value is at most $\alpha$ times the minimum cut value of $G$.

A tree packing is a set of spanning trees, each with a weight assigned
	to it, such that the total weight of the trees
	containing a given edge is at most one.
The value of a tree packing is the total weight of the trees in it.
A maximum tree packing is a tree packing of largest value. 
(When there is no ambiguity,  will use ``maximum tree packing'' to refer also 
to the value of a maximum tree packing, and a ``mincut'' to the value of a
mincut.)
A graph $G$ is called a $\delta$-fat graph for $\delta > 0$, if the maximum tree  
packing of $G$ is at least $\frac{(1+\delta)c}{2}$ \cite{Ka00}.

After \cite{Ka00}, we say that a cut $X$ $k$-respects a tree $T$ (equivalently,
$T$ $k$-constrains $X$), if $E[X] \cap E[T] \leq k$, where $E(X)$ is the set of crossing edges of $X$, and $E(T)$ is the set of edges of $T$. 

\subsection{Previous Results}
Several approaches have been tried in designing in-core
algorithms for the mincut problem. See the results in \cite{FF56,
Ga95, GH61, HaOr94, KaSt96, Ka00,  NI92a, NII99, NOI94, StWa97}.
Significant progress  has been made in designing parallel algorithms
        as well \cite{GoShSt82, Ka93, Ka00, KaMo97}.
The mincut problem on weighted directed graphs is shown to be
        P-complete for LOGSPACE reductions \cite{GoShSt82, Ka93}.
For weighted undirected graphs, the problem is shown to be
         in NC \cite{KaMo97}.
We do not know any previous result for this problem on the external memory model.

However, the current best deterministic in-core algorithm on an unweighted graph computes the minimum cut in
        $O(E + c^2 V \log(V/c))$ time, and was given by Gabow \cite{Ga95}.
Gabow's algorithm uses the matroid characterisation of the minimum cut problem.
According to this characterisation, the minimum cut in a graph $G$ is equal to
        the maximum number of disjoint spanning trees in $G$.
The algorithm computes the minimum cut on the given undirected graph in two phases.
In first phase, it partitions the edges $E$ into spanning forests $F_i$, $i = 1, \ldots , N$ in $O(E)$ time using the algorithm given in \cite{NI92a}.
This algorithm requires $O(1)$ I/O for each edge in the external memory model.
In the second phase, the algorithm computes maximum number of disjoint
        spanning trees in $O(c^2 V \log (N/c))$ time.
To compute maximum number of disjoint spanning trees, augmenting paths like in
        flow based algorithms are computed.
The procedures used in computing augmenting paths access the edges randomly and require $O(1)$ I/Os for each edge in the external memory model.
Thus, this algorithm, when executed on the external
        memory model, performs $O(E + c^2 V \log(V/c))$ I/Os.

\subsection{Our Results}
We present a minimum cut algorithm that runs in $O(c (\msf{V,E}\log E + \frac{V}{B} \sort{V}))$ I/Os, and
	performs better on dense graphs than the algorithm of \cite{Ga95}, which requires $O(E + c^2 V \log(V/c))$ 
	I/Os, where $\msf{V,E}$ is the number of I/Os required in computing a minimum 
	spanning tree. 
For a $\delta$-fat graph, our algorithm 
	computes a minimum cut in $O(c  (\msf{V,E} \log E + \sort{E}))$ I/Os.
Furthermore, we use our algorithm to construct a data structure
	that represents all $\alpha$-minimum cuts, for $\alpha < 3/2$.
The construction of the data structure requires an additional $O(\sort{k})$ I/Os, where $k$ is
	the total number of $\alpha$-minimum cuts. 
Our data structure answers an $\alpha$-minimum cut query in
	$O(V/B)$ I/Os. 
The query is to verify whether a given cut (defined by a vertex partition), is 
	$\alpha$-minimum or not.

Next, we show that the minimum cut problem can be computed with high
probability in 
$O(c \cdot \msf{V,E} \log E  + \sort{E} \log^2 V +  \frac{V}{B} \sort{V} \log V)$ I/Os.
We also present a $(2 + \epsilon)$-minimum cut
algorithm that requires $O((E/V) \msf{V,E})$ I/Os and performs better on 
sparse graphs than our exact minimum cut algorithm.

All our results are summarised in Table \ref{result4}.
\begin{table}[t]
\begin{tabular}{|l|l|}
\cline{1-2} 
\hline
{\bf Problems on EM model}       & {\bf Lower/Upper Bounds} \\ \hline  \hline
mincut of an undirected          & $\Omega(\frac{E}{V}\sort{V})$   \\
unweighted graph                 & $O(c  (\msf{V,E}\log E + \frac{V}{B} \sort{V}))$   \\ \hline
mincut of a $\delta$-fat graph   & $O(c  (\msf{V,E} \log E + \sort{E}))$   \\ \hline
Monte Carlo mincut algorithm     & $O(c \cdot \msf{V,E} \log E  + $                        \\
with probability  $1-1/V$        & $ \sort{E} \log^2 V +  \frac{V}{B} \sort{V} \log V)$ \\ \hline
$(2 + \epsilon)$-approx. mincut  & $O((E/V) \msf{V,E})$                                 \\ \hline
Data Structure for               & $O(c  (\msf{V,E}\log E +  \frac{V}{B} \sort{V})+ \sort{k})$ \\
all $\alpha$-mincuts             & answers a query in $O(V/B)$ I/Os \\ \hline
\end{tabular}
\caption{Our Results}
\label{result4}

\end{table}

\subsection{Organisation of This Chapter}
The rest of the chapter is organised as follows. In Section~\ref{notation}, 
we define some notations used in this chapter.
In Section~\ref{lb:mincut}, we give a lower bound result for the minimum
	cut problem.
In Section~\ref{min:cut}, we present an external memory algorithm for the minimum cut problem.
Section~\ref{sec:data:structure:cut} describes the construction of a data structure
	that stores all $\alpha$-minimum cuts, 
	for $\alpha < 3/2$. 
In Section~\ref{randomize}, we improve the I/O complexity of our minimum cut algorithm by using randomisation.
In Section~\ref{fat:graph}, we discuss a special class of graphs for which a minimum cut
    can be computed very efficiently.
In Section~\ref{approximate}, we present a $(2 + \epsilon)$-minimum cut algorithm.

\section{Some Notations}
\label{notation}
For a cut $X$ of graph $G$, $E(X)$ is the set of crossing edges of $X$.
For a spanning tree $T$ of $G$, $E(T)$ is the set of edges of $T$.

Let $v\downarrow$ denote the set of vertices that are descendants of $v$ 
	in the rooted tree, and
	$v\uparrow$ denote the set of vertices that are ancestors of $v$ in
	the rooted tree.
Note that $v \in v\downarrow$ and $v \in v\uparrow$.
Let $C(A,B)$ be the total number of edges with one endpoint in vertex
set $A$ and
    the other in vertex set $B$. An edge with both endpoints in both sets is
    counted twice.
Thus, $C(u,v)$ is $1$, if $(u,v)$ is an edge,
     $0$ otherwise.
For a vertex set $S$, let $C(S)$ denote $C(S,V-S)$.  

\section{A Lower Bound for the Minimum Cut Problem}
\label{lb:mincut}
We use P-way Indexed I/O-tree to prove the lower bound. The P-way Indexed I/O-tree is defined in \cite{MR99} and is shown that it can be transformed to a binary decision for the same
	problem.
We can use the bound on the number of comparisions in the decision tree to establish a bound on the number of I/Os in the P-way Indexed I/O-tree.
The following lemma, given in \cite{MR99} can be used to prove the lower bound on I/Os.
\begin{lemma}\cite{MR99}
\label{thm:mr99}
Let $X$ be the problem solved by an I/O tree $I/O_T$, with $N$ the number of records in the input.
There exists a decision tree $T_c$ solving $X$, such that:
\[ \mbox{Path}_{T_c} \leq N \log B + D. I/O_T \cdot O(B \log \frac{M-B}{B} + \log P) 
\]
\end{lemma}

We define two decision problems {\bf P$_1$} and {\bf P$'_1$} as follows: 
\begin{definition}
{\bf P$_1$:} Given as input a set $S$ of $N$ elements,
	each with an integer key drawn from the range $[1,P]$, say ``yes'' when $S$ contains 
	either every odd element or at least one even element in the range $[1,P]$, and say
	``no'' otherwise (that is, when $S$ does not contain at least one
	odd element and any even element in the range $[1,P]$.)
\end{definition}
\begin{definition}
{\bf P$'_1$:} The problem is a restriction of {\bf P$_1$}. Suppose the input $S$ is divided into $P$ subsets each of size $N/P$
        and containing distinct elements from the range $[P+1,2P]$,  where
        $P < N < P(\lceil P/2 \rceil - 1)$.
{\bf P}$'_1$ is to decide whether $S$ contains
        either every odd element or at least one even element in the range $[P+1,2P]$.  

\end{definition}
First, we prove the following claim for {\bf P$_1$},
\begin{claim}
The depth of any decision tree for {\bf P$_1$} is
        $\Omega(N\log P)$.
\end{claim}
\begin{proof}
$(\lceil P/2\rceil -1)^N$ and $\lceil\frac{P}{2}\rceil(\lceil P/2\rceil -1)^N$ 
	are, respectively, lower and upper bounds on the number of different ``no'' instances.
If we prove that in any linear decision tree for {\bf P$_1$}, there is a one-one correspondence between ``no'' instances and leaves that decide ``no'', then the depth of any decision tree for {\bf P$_1$} is
	$\Omega(N\log P)$.

Our proof is similar to the ones in \cite{MR99}. 
We consider tertiary decision trees in which each decision node has three outcomes: $<$, $=$, $>$. 
Each node, and in particular each leaf, corresponds to a partial order on $S\cup \{1,\ldots,P\}$.
Consider a ``no'' leaf $l$ and the partial order $PO(l)$ corresponding to $l$.
All inputs visiting $l$ must satisfy $PO(l)$. 
Let $C_1, \ldots C_k$ be the equivalence classes of $PO(l)$. 
Let $u_i$ and $d_i$ be the maximum and minimum values respectively of the elements of 
	equivalence class $C_i$ over all ``no'' inputs that visit $l$.
Exactly one input instance visits $l$ if and only if $u_i=d_i$ for all $i$.
If $u_i\not =d_i$ for some $C_i$ then, pick a ``no'' input $I$ that visits $l$
	and fabricate a ``yes'' input $I'$ as follows: assign an even integer $e$, $d_i<e<u_i$,
	to every element in $C_i$, and consistent with this choice and $PO(l)$ change 
	the other elements of $I$ if necessary.
Note that this fabrication is always possible.
Since $I'$ is consistent with $PO(l)$, it visits $l$; a contradiction.
Hence our claim.
\end{proof}

Now we consider the restriction {\bf P$'_1$} of {\bf P$_1$}.
\begin{claim}
\label{claim:P}
The depth of any decision tree for {\bf P$'_1$} is
        $\Omega(N\log P)$.
\end{claim}
\begin{proof}
A lower bound on the number of ``no'' instances of {\bf P$'_1$}
	is $(P' \cdot (P'-1) \cdot (P'-2) \cdot \ldots \cdot (P' - N/P))^P=\Omega(P)^N$, where
        $P' = \lceil P/2 \rceil - 1$.
An argument similar to the above shows that in any decision tree for 
	{\bf P$'_1$}, the ``no'' instances and leaves that decide ``no''
        correspond one to one.
Therefore, the depth of any decision tree for {\bf P$'_1$} is
        $\Omega(N\log P)$.
\end{proof}
\begin{theorem}
The depth of any I/O-tree for computing the min-cut of a $V$ vertex, $E$ edge graph is $\Omega(\frac{E}{V} \sort{V})$ assuming $E \geq V$ and $\log V < B \log \frac{M}{B}$.
\end{theorem}
\begin{proof}
We construct an undirected graph $G=(V,E)$ from an input instance $I$ of {\bf P$'_1$} as follows.
\begin{enumerate}
\item Let the integers in $[1,2P]$ constitute the vertices.
\item Make a pass through $I$ to decide if it contains an even element.
If it does, then for each $i \in [P+1,2P-1]$,
        add an edge $\{i,i+1\}$ to $G$. Otherwise, remove all even integers (vertices) $> P$ from $G$.
\item Make a second pass through $I$. If the $j$th subset of $I$ contains $P+i$, then add an edge $\{j, P+i\}$ to $G$.
\item For $i \in [1,P-1]$, add an edge $\{i , i+1\}$ to $G$.
\end{enumerate}
Here $|V| = \Theta(P)$ and $|E| = \Theta(N)$. 
One example is shown in figure~\ref{example}.
The construction of the graph requires $O(N/B)$ I/Os.
It needs looking at the least significant bits (LSBs)
	of the keys of the elements;
	if the LSBs are assumed to be given separately from the rest of the keys, this
	will not violate the decision tree requirements.
The value of the minimum cut of $G$ is at least $1$ iff {\bf P$'_1$} answers ``yes'' on $I$.

The bound then follows from the lemma~\ref{thm:mr99} and claim~\ref{claim:P}.
\end{proof}

\begin{figure}
---------------------------------------------------------------------------------------------

$P = 5$, \\
Yes Instance $I_{\mbox{yes}} = \{I_1 = 7, I_2 = 7, I_3 = 8, I_4 = 7, I_5 = 6 \}$ \\
No Instance $I_{\mbox{no}} = \{I_1 = 7, I_2 = 7, I_3 = 7, I_4 = 7, I_5 = 7 \}$ \\

\begin{center}
\includegraphics[height = 3cm, width = 11.0cm]{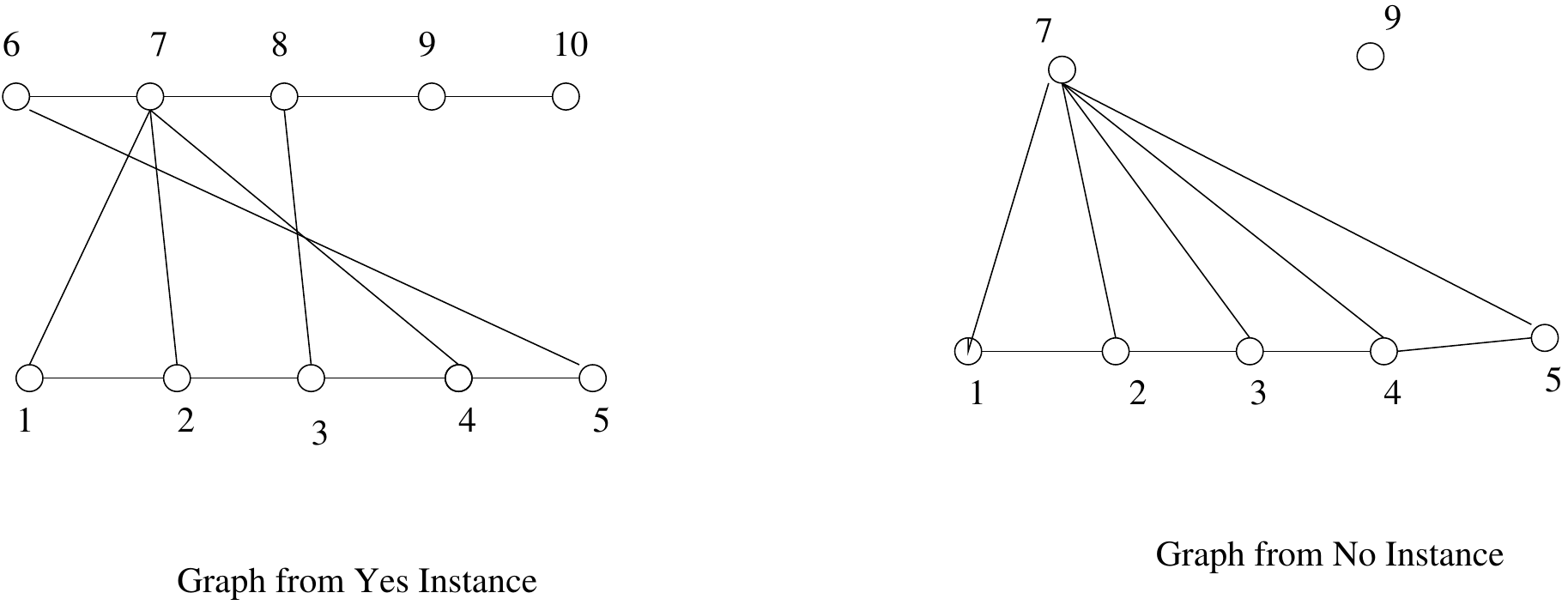}
\end{center}

---------------------------------------------------------------------------------------------
\caption{Example of 'Yes' and 'No' Instance}
\label{example}
\end{figure}

\section{The Minimum Cut Algorithm}
\label{min:cut}
We present an I/O efficient deterministic algorithm for finding mincuts on undirected 
	unweighted graphs.
Our I/O efficient algorithm is based on the semi-duality between minimum
	cut and tree packing.
The duality was used by Gabow \cite{Ga95} in designing
	a deterministic minimum cut in-core algorithm for both directed and undirected 
	unweighted graphs.
It was also used by Karger
	\cite{Ka00} in designing a faster but randomized in-core algorithm for
	undirected weighted graphs.
Our algorithm uses Karger's ideas \cite{Ka00}. 

Nash-Williams theorem \cite{NW61} states that any undirected graph
	with minimum cut $c$ has at least $c/2$ edge disjoint spanning trees.
It follows that in such a packing, for any minimum cut, 
	there is at least one spanning tree that 2-constrains the minimum cut.
Once we compute such a packing, the problem reduces to finding
	a minimum cut that is 2-constrained by some tree in the packing.
The assumption on edge disjointness is
	further relaxed by Karger \cite{Ka00} in the following lemma.
\begin{lemma}{\cite{Ka00}}
\label{lem:tree:packing}
For any  graph $G$, for any tree packing $P$ of $G$ of
	value $\beta c$, and any cut $X$ of $G$ of value $\alpha c$ 
	($\alpha \geq \beta$), at least 
	$(1/2)(3 - \alpha/\beta)$ fraction (by weight) of trees of $P$ $2$-constrains $X$.
\end{lemma}
If an approximate algorithm guarantees a $\beta c$ packing $P$, for $\beta>1/3$, at least
	$\frac{1}{2}(3 - \frac{1}{\beta})$ fraction (by weight) of the trees in $P$ 
	$2$-constrains any given minimum cut $X$.
In particular, there is at least one tree in $P$ that $2$-constrains
	any given minimum cut $X$.

\subsection{The Algorithm}

From the above discussion, we can conclude that the minimum cut problem
can be divided into two subproblems, (i) compute an approximate
	maximal tree packing $P$ of value $\beta c$, for $\beta>1/3$, and 
	(ii) compute a minimum cut of the graph $G$ that is $2$-constrained 
	by some tree in $P$.

\subsubsection{Subproblem 1}
\label{phase1}

We use the greedy tree packing algorithm given in \cite{PST95}\cite{ThKa00}\cite{Y95}. 
It is described in Figure~\ref{tree:packing}.
A tree packing in $G$ is an assignment of weights to the spanning trees of $G$
    so that each edge gets a load of
\begin{center}
$l(u,v) = \sum_{T: (u,v) \in T} w(T) \leq 1$
\end{center}
The  value of tree packing is $W=\sum_{T} w(T)$.
The algorithm is given in Figure~\ref{tree:packing}.

\begin{figure}[t]
\vspace{0.2in}
\hrule
\begin{tabbing}
aaaa \= aaaa \= aaaa \= aaaa \= aaaa \= aaaa \kill \\

Initially no spanning tree has any weight, and all edges have load $0$. \\

Set $W = 0$ \\

While no edge has load $1$ perform all of the following steps \\

\> Pick a load minimal spanning tree $T$. \\

\> $w(T) = w(T) + \epsilon^2 / 3 \log E$. \\

\> $W = W + \epsilon^2 / 3 \log E$. \\

\> For all edges $(u,v)$ selected by $T$, \\

\>\> $l(u,v) = l(u,v) + w(T)$ \\

Return $W$. \\

\end{tabbing}
\vspace{0.2in}
\hrule
\caption{Greedy Tree Packing Algorithm for Computing a $(1-\epsilon)$-approximate tree packing}
\label{tree:packing}
\end{figure}

As mentioned in \cite{ThKa00}, the algorithm obtains the following result.
\begin{theorem}\cite{PST95, Y95}
The greedy tree packing algorithm,
	when run on a graph $G$,
	computes a $(1 - \epsilon)$-approximate tree packing
	of value $W$; that is, $(1 - \epsilon)\tau \leq W \leq \tau$, where $\tau$ is
	the maximum value of any tree packing of $G$, and $0 < \epsilon < 1$.
\end{theorem}

Since each iteration increases the packing value by $\epsilon^2 / (3 \log E)$,
	and the packing value can be at most $c$, the algorithm 
	terminates in $O(c \log E / \epsilon^2)$ iterations.
The I/O complexity of each iteration is dominated by the minimal spanning tree
	computation. 
Thus, number of I/Os required is $O(c \cdot \msf{V,E} \log E)$.
Since the value of the maximum tree packing is at least $c/2$, 
	the size of the computed tree packing is at least $(1 - \epsilon)c/2$. 
From Lemma~\ref{lem:tree:packing}, it follows that, for $\epsilon < 1/3$, 
	and any minimum cut $X$, the computed
	tree packing contains at least one tree that $2$-constrains $X$.
	
\subsubsection{Subproblem 2}
\label{phase2} 

Let $T=(V,E')$ be a spanning tree of graph $G=(V,E)$.
For every $K\subseteq E'$ there is unique cut $X$ so that $K=E(T)\cap E(X)$.
$X$ can be constructed as follows: Let $A=\emptyset$.
For some $s\in V$, for each vertex $v$ in $V$,
add $v$ to set $A$, iff the path in $T$ from $s$ to $v$ has an even number of edges from $K$;
clearly $X=(A,V-A)$ is a cut of $G$.  

A spanning tree in the packing produced by Subproblem~1 $2$-constrains
every mincut of $G$. We compute the following:
(1)~for each tree $T$ of the packing, and for each tree edge $(u,v)$ in 
	$T$, a cut $X$ such that $(u,v)$ is the only edge of $T$ crossing $X$, 
(2)~for each tree $T$ of the packing, and for each pair of tree edges $(u_1,v_1)$ 
and $(u_2,v_2)$, a cut $X$ such that $(u_1,v_1)$ and $(u_2,v_2)$ 
are the only edges of $T$ crossing $X$.
A smallest of all the cuts found is a minimum cut of $G$.

First we describe the computation in (1).
Root tree $T$ at some vertex $r$ in $O(\sort{V})$ I/Os \cite{ChGoGr+95}.
(See Section~\ref{notation} for notations.)
$C(v \downarrow)$ is the set
	of edges whose
	one endpoint is a descendent of $v$,
	and the other endpoint is a nondescendent of $v$.
If $(v, p(v))$ is the only tree edge crossing a cut $X$, then
	$C(v \downarrow)$ is the value of cut $X$,
	where $p(v)$ is the parent of $v$ in $T$.
As given in \cite{Ka00}, The value of $C(v \downarrow)$ is
\begin{center}
$C(v \downarrow) = d^{\downarrow}(v) - 2 
\rho^{\downarrow}(v)$
\end{center}
where $d^{\downarrow}(v)$ is the total number of nontree edges
	incident on vertices in $v \downarrow$, and $\rho{\downarrow}(v)$
	is the total number of nontree edges whose both endpoints are
	in $v \downarrow$.
$C(v \downarrow)$ is to be computed for all vertices $v$ in 
tree $T$, except for the root $r$. 

$d^{\downarrow}(v)$ can be computed by using 
	expression tree evaluation, if
	the degree of each vertex $v$ is stored with $v$. 
$\rho^{\downarrow}(v)$ can be computed using least common
ancestor queries and expression tree evaluation. Once
$d^{\downarrow}(v)$ and $\rho^{\downarrow}(v)$ are known for every vertex
	$v \in T$, $C(v \downarrow)$ can be computed for every vertex $v$  
	using expression tree evaluation.
If we use the I/O efficient least common ancestor and expression tree evaluation 
	algorithms of \cite{ChGoGr+95, zeh}, the total number of I/Os needed
	for the computation in~(1) is $O(\sort{V})$.
 
For the computation in (2), consider two tree edges
	$(u,p(u))$ and $(v,p(v))$, the edges from two vertices $u$ and
	$v$ to their respective parents.
Let $X$ be the cut characterised by these two edges (being the only tree edges
	crossing $X$).

We say vertices $u$ and $v$ are incomparable, if $u \not\in v\downarrow$
    and $v \not\in u\downarrow$; that is, if they are not on the same root-leaf path.
If $u \in v \downarrow$ or $v \in u \downarrow$, then $u$ and $v$ are 
called comparable and both are in the same root-leaf path. 

In the following, when we say the cut of $u$ and $v$, we mean the cut defined by
edges $(p(u),u)$ and $(p(v),v)$.


As given in \cite{Ka00} and shown in Figure~\ref{figure:incomparable} and 
	Figure~\ref{figure:comparable}, 
	if $u$ and $v$ are incomparable then the value of cut $X$ is
	\[C(u \downarrow \cup \; v\downarrow)=C(u \downarrow) + C(v \downarrow) - 2 C(u \downarrow,
	v\downarrow)\] If vertices $u$ and $v$ are comparable then the value of $X$ is 
	\[C(u \downarrow - \; v \downarrow)=
	C(u \downarrow) - C(v \downarrow) + 2(C(u \downarrow, v \downarrow) - 2 \rho^{\downarrow}(v))\]

For each tree in the packing, and for each pair of vertices in the tree we need to
compute the cuts using the above formulae. We preprocess each tree $T$ as follows.
Partition the vertices of $T$ into clusters $V_1, V_2, \ldots, V_{N}$ (where $N=\Theta(V/B)$), each of size 
	$\Theta(B)$, except for the last one, which can of a smaller size. 
Our intention is to process the clusters one at a time by reading each $V_i$ into the main memory
to compute the cut values for every pair with at least one of the vertices in $V_i$.
We assume that $T$ is rooted at some vertex $r$. 

\paragraph{Partitioning of vertices:}

For each vertex $v \in V$, a variable $\var{v}$ is initialised to $1$.
The following steps are executed for grouping the vertices into
clusters.

Compute the depth of each vertex from the root $r$. 
Sort the vertices $u \in V$ in the decreasing order of the composite key   
	$\langle {\tt depth}(u), p(u)  \rangle$. Depth of $r$ is $0$. 
Access the vertices in the order computed above. Let $v$ be the current vertex. 
\begin{itemize}
\item Compute $Y=\var{v} + \var{v_1} + \ldots + \var{v_k}$, where $v_1, v_2, \ldots, v_k$ are the
	children of $v$. If $Y<B$ then set $\var{v}=Y$.
\item Send the value $\var{v}$ to the parent of $v$, if $v$ is not the root $r$.
\item If $Y>B$, divide the children of $v$ into clusters ${\mathcal{Q}} = Q_1, Q_2, \ldots, Q_l$ such that
        for each cluster $Q_i$, $\sum_{u \in Q_i} \var{u} = \Theta(B)$.
	If $v=r$, it joins one of the clusters $Q_i$. 
\end{itemize}

After executing the above steps for all vertices, consider the vertices $u$, one by one, in the
reverse order, that is, in increasing order of the composite key
$\langle {\tt depth}(u), p(u)  \rangle$.
Let $v$ be the current vertex.
If $v$ is the root, then it labels itself with the label of the cluster to which it belongs.
Otherwise, $v$ labels itself with the label received from its parent.
If $v$ has not created any clusters, then it sends its label to all its
        children.
Otherwise, let the clusters created by $v$ be $Q_1, Q_2, \ldots, Q_l$; $v$ 
	labels each cluster uniquely and sends to each child $v_{i}$  
	the label of the cluster that contains $v_i$. 
At the end, every vertex has got the label of the cluster that contains it. 


In one sort, all vertices belonging to the same cluster $V_i$ can be brought  
	together. Since $T$ is a rooted tree, each vertex $u$
	knows its parent $p(u)$. We store $p(u)$ with $u$ in $V_i$.
Thus, $V_i$ along with the parent pointers, forms a subforest $T[V_i]$ of $T$, and we obtain 
	the following lemma.
\begin{lemma}
The vertices of a tree $T$ can be partitioned into clusters $V_1, \ldots V_{N}$ (where $N=\Theta(V/B)$), 
	of size $\Theta(B)$ each, in $O(\sort{V})$ I/Os, with
	the clusters satisfying the following property:
	for any two roots $u$ and $v$ in $T[V_i]$, $p(u)=p(v)$.
\end{lemma}
\begin{proof}
The partitioning procedure uses the time forward processing method for sending
	values from one vertex to another and can be computed 
	in $O(\sort{V})$ I/Os \cite{ChGoGr+95, zeh}.
The depth of the nodes can be found by computing an Euler
Tour of $T$ and applying list ranking on it~\cite{zeh, ChGoGr+95} in  
	$O(\sort{V})$ I/Os.
Thus, a total of $O(\sort{V})$ I/Os are required for the partitioning 
	procedure.

The property of the clusters mentioned in the lemma follows from the way the
clusters are formed. Each cluster $V_i$ is authored by one vertex $x$, and therefore
each root in $T[V_i]$ is a child of $x$. 
\end{proof}

Connect every pair $V_i,V_j$ of clusters
    by an edge, if there exists an edge $e \in E'$ such that one of its endpoint
    is in $V_i$ and the other endpoint is in $V_j$.
The resulting graph $G'$ must be a tree, denoted as cluster tree $T'$.
Note that $T'$ can be computed in $O(\sort{V})$ I/Os.
Do a level order traversal of the cluster tree:
sort the clusters by depth, and then by key
(parent of a cluster) such that (i) deeper clusters come first,
and (ii) the children of each cluster are contiguous. 
We label the clusters in this sorted order: $V_1, V_2, \ldots, V_{N}$.
Within the clusters the vertices are also numbered the same way.

Form an array $S_1$ that lists $V_1,\ldots,V_N$ in that order; after $V_i$ and
	before $V_{i+1}$ are listed nonempty $E_{ij}$'s, in the increasing order
	of $j$; $E_{ij} \subseteq E - E'$, is the set of non-$T$ edges of $G$ with one endpoint in $V_i$ 
	and the other in $V_j$.
With $V_i$ are stored the tree edges of $T[V_i]$.
Another array $S2$ stores the clusters $V_{i}$ in the increasing order of $i$. 
The depth of each cluster in the cluster tree can be computed in $O(\sort{V})$ I/Os \cite{ChGoGr+95, zeh},
    and arrays $S_1$ and $S_2$ can be obtained in $O(\sort{V + E})$ I/Os.

\paragraph{Computing cut values for all pair of vertices:}

Now, we describe how to compute cut values for all pair of vertices.
Recall that the value of the cut for two incomparable vertices $u$ and $v$ is
\begin{center}
$C(u \downarrow \cup \; v \downarrow) = C(u \downarrow) + C(v
\downarrow) -
    2 C(u \downarrow, v \downarrow)$
\end{center}
and for two comparable vertices $u$ and $v$ is 
\begin{center}
$C(u \downarrow - \; v \downarrow) = C(u \downarrow) - C(v
\downarrow) +
    2(C(u \downarrow, v \downarrow) - 2 \rho^{\downarrow}(v))$
\end{center}

Except for $C(u \downarrow, v \downarrow)$, all the other values of both expressions 
	have already been computed.
$C(u \downarrow, v \downarrow)$ can be computed using
the following expression.
\begin{center}
$C(u \downarrow, v \downarrow) = \sum_{\hat{u}} C(u_k \downarrow, v \downarrow) +
\sum_{\hat{v}} C(u, v_{l} \downarrow) +
C(u , v)$
\end{center}
where, $\hat{u}$ and $\hat{v}$ vary over the children of  $u$ and $v$ respectively. 
In Figure~\ref{subprob2}, we give the procedure for computing $C(u \downarrow, v \downarrow)$ 
	and cut values for all pairs $u$ and $v$.
In the procedure $p(u)$ is the parent of $u$ in $T$, $P(V_{i})$ is the
parent of $V_i$ in the cluster tree.

\begin{figure}
\vspace{0.2in}
\hrule
\begin{tabbing}
aaaa \= aaaa \= aaaa \= aaaa \= aaaa \= aaaa \kill \\
Let binary $L_{i}$ be $1$ iff ``$V_{i}$ is a leaf of the cluster tree''; $\bar{L_i}$ is its negation \\
Let binary $l_{u}$ be $1$ iff ``$u$ is a leaf in its cluster''; $\bar{l_u}$ is its negation \\
For $1\leq i\leq N$, For $1\leq j\leq N$ \\
\> For each $(u,v)\in V_i\times V_j$ considered in lexicographic ordering \\
\>\> if $l_v$ and $\neg L_j$ \\
\>\>\> Deletemin$(Q_1)$ to get $\langle j, u, v, Y \rangle$; add $Y$ to $Y_v$; \\
\>\> if $l_u$ and $\neg L_i$ \\
\>\>\> Deletemin$(Q_2)$ to get $\langle i, j, u, v, X \rangle$; add $X$ to $X_u$; \\
\> For each $(u,v)\in V_i\times V_j$ considered in lexicographic ordering \\
\>\> $A=\sum_{\hat{u}} C(u_k \downarrow, v \downarrow)$ \\
\>\> $B=\sum_{\hat{v}} C(u, v_{l} \downarrow)$ \\
\>\> $C(u\downarrow, v\downarrow)= A\bar{l_u} + X_u\bar{L_i}l_u + B\bar{l_v} + Y_v\bar{L_j}l_v + C(u , v)$ \\
\>\> if $u$ and $v$ are incomparable vertices \\
\>\>\> $C(u \downarrow \cup \; v \downarrow) = C(u \downarrow) + C(v \downarrow) - 2 C(u \downarrow, v \downarrow)$ \\
\>\> if $u$ and $v$ are comparable vertices \\
\>\>\> $C(u \downarrow - \; v \downarrow) = C(u \downarrow) - C(v
\downarrow) +
    2(C(u \downarrow, v \downarrow) - 2 \rho^{\downarrow}(v))$ \\

\> Let $r_{i1}, \ldots  r_{ik}$ be the roots in $T[V_i]$ \\
\> Let $r_{j1}, \ldots  r_{jl}$ be the roots in $T[V_j]$ \\
\> For each vertex $u \in V_i$ \\
\>\> $Y^u=C(u , r_{j1} \downarrow) + \ldots + C(u, r_{jl} \downarrow)$ \\
\>\> Store $\langle P(V_j), u, p(r_{j1}), Y^u\rangle$ in $Q_1$ \\
\> For each vertex $v \in V_j$ \\
\>\> $X^v=C(r_{i1}\downarrow, v \downarrow) + \ldots + C(r_{ik} \downarrow, v\downarrow)$ \\
\>\> Store $\langle P(V_i), j, p(r_{i1}), v, X^v\rangle$ in $Q_2$ \\
\end{tabbing}
\vspace{0.2in}
\hrule
\caption{Procedure to compute cut values for all pair of vertices}
\label{subprob2}
\end{figure}

For each $i,j \in N$, $V_i$, $V_j$ and $E_{ij}$ are brought in main memory.
Note that the size of $E_{ij}$ can be at most $O(B^2)$. 
We assume that size of main memory is $\Omega(B^2)$.
$C(u \downarrow)$ and $\rho^{\downarrow}(u)$ are stored with vertex $u$.

\begin{figure}
\begin{center}
 
\includegraphics[height=2.5in]{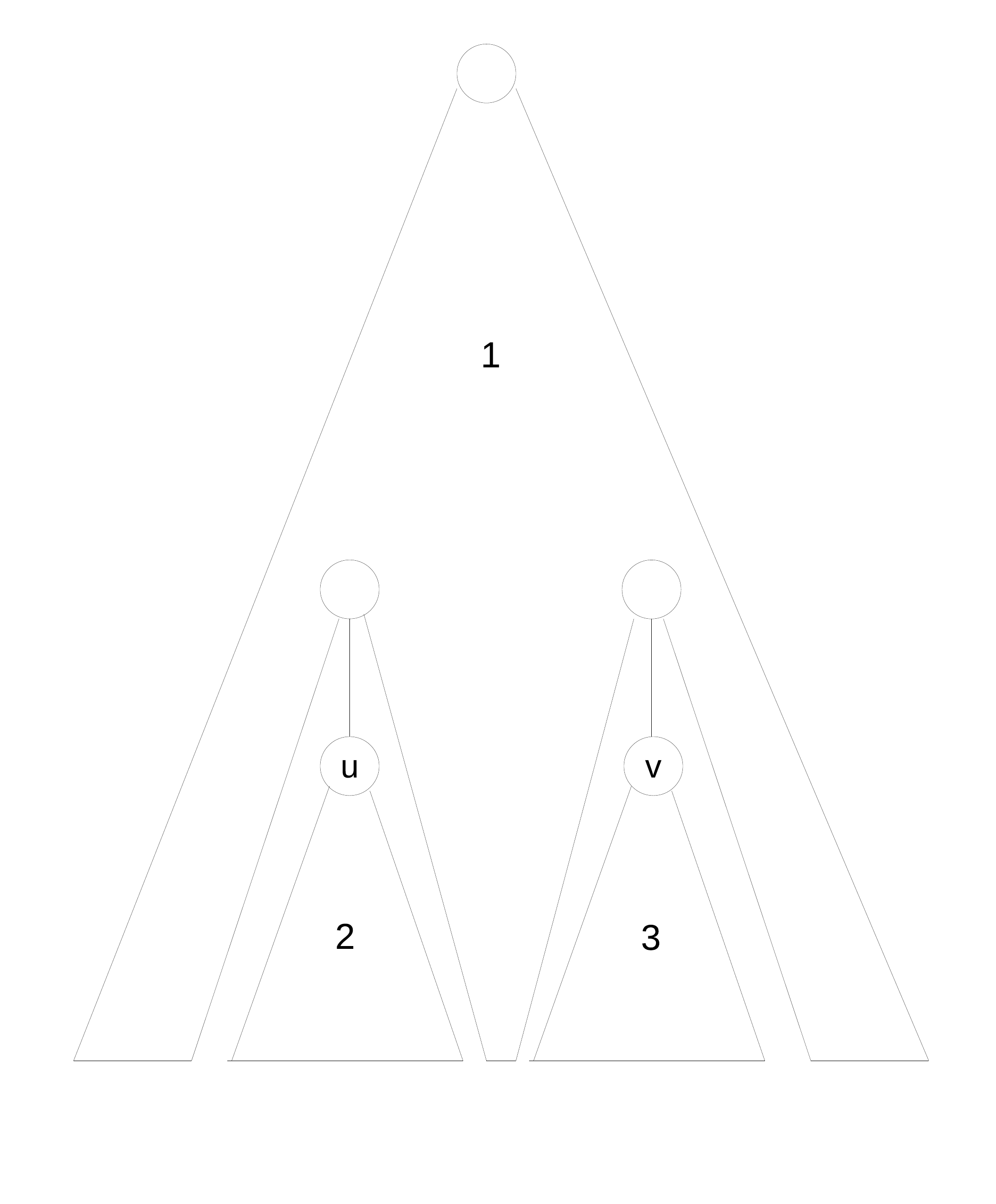}
 
\caption{$C(u \downarrow \cup \; v \downarrow)$: set of edges from region 2 to 1 and 3 to 1}
\label{figure:incomparable}
\end{center}
\end{figure}

\begin{figure}
\begin{center}
\includegraphics[height=2.5in]{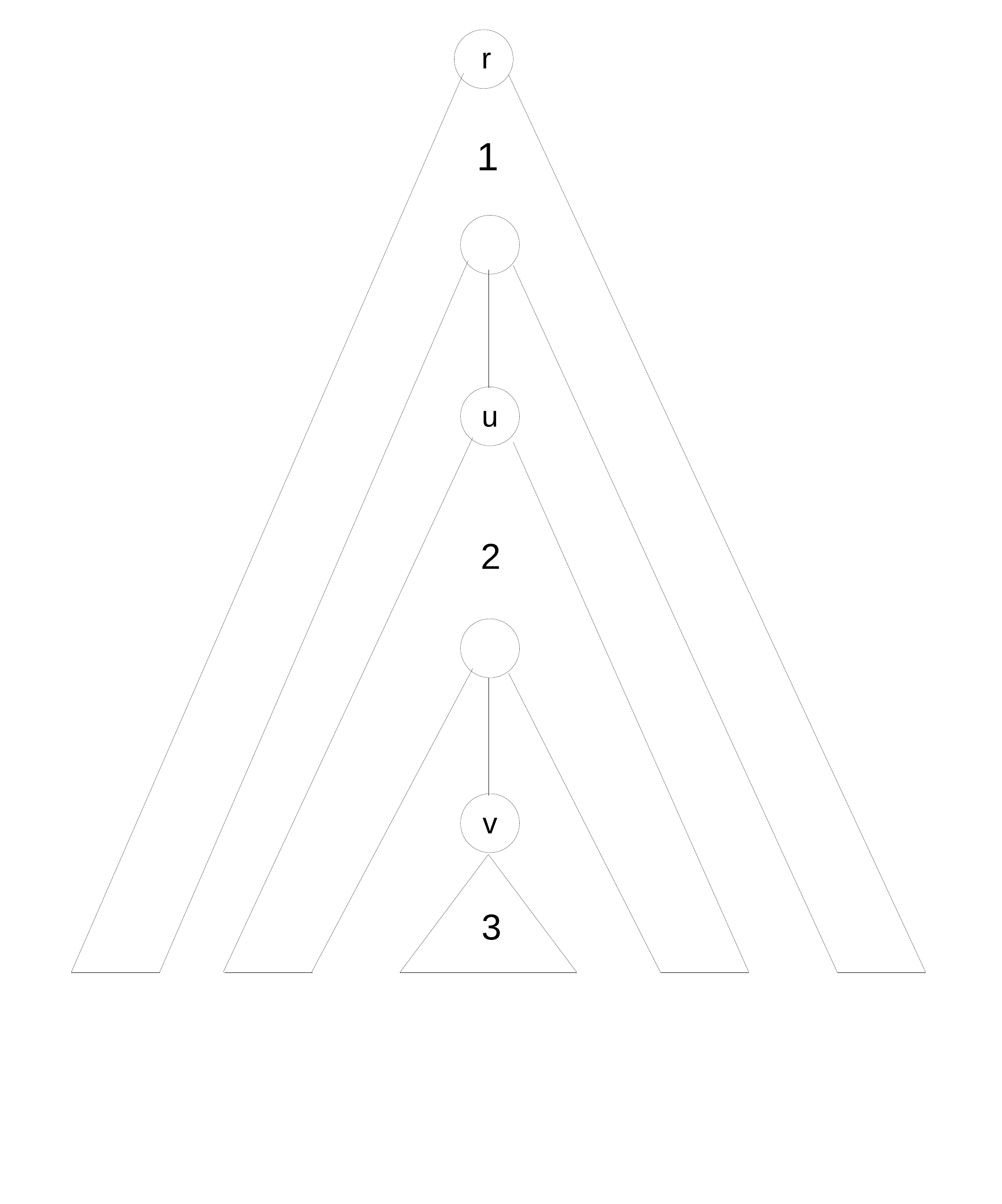}
\caption{$C(u \downarrow - \; v \downarrow)$: set of edges from region 2 to 1 and 2 to 3}
\label{figure:comparable}
\end{center}
\end{figure}

We mark all ancestors of each vertex $u \in V_i$. Note that
    the ancestor vertices of each vertex $u \in V_i$ which reside in other clusters $V_j$ for
	$j \neq i$, are the same, and $V_j$ is an ancestor of $V_i$ in cluster tree $T'$.
We can mark all the ancestor vertices  in  additional $O(V/B)$ I/Os. 

Two priority queues $Q_1$ and $Q_2$ are maintained during the execution of the
	algorithm.
$Q_1$ holds value $Y_{ij}^{uv}= C(u , v_1 \downarrow) + \ldots + C(u, v_l \downarrow)$ 
	with key value $\langle j, u, v \rangle$ 
	for each vertex $u \in V_i$ and $v \in V_j$, while cluster $V_j$ is yet to be  
	accessed for $V_i$, and after $V_{k}$ (with $k<j$, and containing exactly 
$v_1, \ldots v_l$ 
	among the children of $v$), has been processed for $V_{i}$, and
	$C(u , v_1 \downarrow), \ldots C(u, v_l \downarrow)$ have been computed.
Note that it is not necessary that all children of $v$ are in one cluster $V_j$. 
Similarly $Q_2$ holds value $X_{ij}^{uv}=C(u_1 \downarrow , v \downarrow) + \ldots + 
C(u_l \downarrow, v \downarrow)$
	with key value $\langle i , j , u, v \rangle$ 
	for each vertex $u\in V_i$ and $v \in V_j$, while cluster $V_j$ is yet to be  
	accessed for $V_i$, and after $V_j$ has been processed for $V_{k}$ 
	(with $k<j$, and containing exactly $u_1, \ldots u_l$ 
	among the children of $u$), and
	$C(u_1 \downarrow , v \downarrow), \ldots, C(u_l \downarrow, v \downarrow)$
	have been computed.
	Note that it is not necessary that all children of $u$ are in one node $V_k$.   

The correctness of the algorithm is easy to prove. 
Since, for each cluster we perform $O(V/B)$ I/Os and $O(V)$ insertions in each priority
	queue $Q_1$ and $Q_2$ and the vertices are partitioned into $\Theta(V/B)$ clusters, 
	the total I/O cost is $O(\frac{V}{B} \sort{V})$.
We obtain the following lemma.
\begin{lemma}
For a tree $T$, a minimum cut can be computed in $O(\sort{E} + (V/B)\sort{V})$ I/Os,
    if at most two edges of it are in $T$.
\end{lemma}

We execute the above operations for all trees in packing and hence
obtain the
    following theorem.
\begin{theorem}
We can compute a minimum cut in $O(c(\msf{V,E} \log E + \sort{E}
+ (V/B)\sort{V} ))$ I/Os
    for the given undirected unweighted graph $G$, where $c$ is the minimum cut value.
\end{theorem}

\section{The Data Structure}
\label{sec:data:structure:cut} 

From Lemma~\ref{lem:tree:packing}, we know that
for every $\alpha$-minimum cut, $\alpha < 3/2$, there is at least
one tree in the packing that crosses the cut at most twice. It has been shown that
when $\alpha<3/2$, there can
be at most $O(V^2)$ $\alpha$-minimum cuts \cite{Ka00}. 

Our mincut algorithm computes the cut value corresponding to each edge and 
	each pair of edges in every tree of a packing. 
Once the mincut has been found, run the algorithm again.
In this run, we identify all $\alpha$-minimum cuts.

We use a perfect hash function to map each $\alpha$-minimum cut into a unique 
	key. Recall that a cut is characterised by a vertex set partition. 
Use the hash key to store all $\alpha$-minimum cut values with the respective tree edges and
trees, in a hash table.
 
If the total number of $\alpha$-minimum cuts
is $k$, then the size of the hash table is $O(k)$, and it can be constructed in
$O(c(\msf{V,E} \log E + \sort{E} + (V/B)\sort{V} ) + \sort{k})$ I/Os.

We can perform the following query in our data structure: given a cut
(defined by a vertex partition), find whether the cut
is $\alpha$-minimum or not. We can answer the query by computing the hash key
of the given cut. Looking up the table will require $O(1)$ I/Os on an average.
Computing of the hash key 
requires $O(V/B)$ I/Os. 

\begin{lemma}
A data structure, containing all $\alpha$-minimum cuts, for $\alpha <
3/2$ can be constructed in $O(c (\msf{V,E} \log E + \sort{E} +
(V/B)\sort{V}) + \sort{k})$ I/Os using $O(k)$ space. The following
query can be answered in $O(V/B)$ I/Os: Given a cut (a partition of the vertex set),
find whether it is $\alpha$-minimum or not.
\end{lemma}

\section{The Randomised Algorithm}
\label{randomize}

The I/O complexity of computing the minimum cut can
be improved, if spanning trees from the packing are chosen randomly.
We assume that the minimum cut is large and $c > \log^2 V$.
We use ideas from \cite{Ka00}.

The maximum tree packing $\tau$ is at least $c/2$. Consider a minimum cut $X$ and
a packing $P$ of size $\tau'=\beta c$.
Suppose $X$ cuts exactly one tree edge of $\eta \tau'$ trees in $P$,
    and cuts exactly $2$ tree edges of $\nu \tau'$ trees in $P$.
Since $X$ cuts at least three edges of the remaining trees in $P$, 
\[
 \eta \tau' + 2\nu \tau' + 3(1 - \eta - \nu) \tau'  \leq  c \]
\[ 3 - 2\eta - \nu  \leq 1/ \beta  \]
\[ \nu  \geq  3-1/\beta - 2 \eta \]

First assume that $\eta > \frac{1}{2\log V}$. Uniformly randomly we pick a tree 
	$T$ from our approximate maximal tree packing.
The probability is $1/(2 \log V)$ that we pick a tree so that exactly one edge of it crosses
	particular minimum cut $X$.
If we choose $2 \log^2 V$ trees, then the probability of not selecting a tree
    that crosses the minimum cut exactly once is
\[ \left( 1 - \frac{1}{ 2\log V} \right)^{2 \log^2 V} < 2^{- \frac{ \log^2 V}{ \log V}}
 < \frac{1}{V} \]
Thus, with probability $(1 - 1/V)$, we compute a minimum cut.

Now suppose that $\eta \leq \frac{1}{2\log V}$. Then, we have
\[ \nu \geq 3- \frac{1}{\beta} - \frac{1}{\log V} \]
Randomly pick a tree. The probability is $\nu$ that we will pick a tree 
	whose exactly two tree edges crosses the minimum cut.
If we select $\log V$ trees from the packing then the probability of not 
	selecting the right tree is
\[ \left(1 - \nu \right)^{\log V} \leq \left(\left(\frac{1}{\beta}-2\right)+\frac{1}{\log V} \right )^{\log V}
\leq \frac{1}{V} \]
If $\beta$ is small enough; for example $\beta=1/2.2$.
Therefore, we compute a minimum cut with probability $1 - 1/V$. This
reduces the I/O complexity to $O(c \cdot \msf{V,E} \log E +
\sort{E}  \log^2 V + \frac{V}{B} \sort{V} \log V)$. 

\section{On a $\delta$-fat Graph}
\label{fat:graph}
A graph $G$ is called a $\delta$-fat graph for $\delta > 0$, if the maximum tree packing of $G$ is at least
    $\frac{(1+\delta)c}{2}$ \cite{Ka00}.
We can compute an approximate maximal tree packing of size at least
$(1 + \delta/2)(c/2)$ from our
    tree packing algorithm by choosing $\epsilon = \frac{\delta}{2 (1 + \delta)}$.
Since $c$ is the minimum cut, a tree shares on an average
    $2 / (1 + \delta/2)$ edges with a minimum cut which is less than $2$, 
	and thus is $1$.
Hence, for a $\delta$-fat graph, for each tree $T$ we need only to investigate cuts
that contain exactly one edge of $T$.
Hence, the minimum cut algorithm takes only
    $O(c  (\msf{V,E} \log E + \sort{E})) $ I/Os; this is dominated by
    the complexity of the tree packing algorithm.

\section{The $(2 + \epsilon)$-minimum cut algorithm}
\label{approximate}

In this section, we show that a near minimum cut can be computed more
    efficiently than an exact minimum cut.
The algorithm (Figure~\ref{approx:cut}) is based on the algorithm of 
\cite{KaMo97, M93}, and computes a cut of value between $c$ and $(2 + \epsilon)c$, 
if the minimum cut of the graph is $c$.

\begin{figure}
\vspace{0.2in}
\hrule
\begin{enumerate}
\item Let $\lambda_{\mbox{min}}$ be the minimum degree of graph $G$.
\item $k = \frac{\lambda_{\mbox{min}}}{2 + \epsilon}$ for $\epsilon > 0$.
\item find sparse $k$-edge connected certificate $H$ (see below for definition).
\item construct graph $G'$ from $G$ by contracting edges not in $H$ and
    recursively find the approximate minimum cut in the contracted graph $G'$.
\item return the minimum of $\lambda_{\mbox{min}}$ and cut returned from step $4$.
\end{enumerate}
\vspace{0.2in}
\hrule
\caption{Approximate minimum cut algorithm}
\label{approx:cut}
\end{figure}

We omit the correctness proof here. It can be found in \cite{KaMo97}. The depth
of recursion for the algorithm is $O(\log E)$ \cite{KaMo97} and
in each iteration the number of edges are reduced by a constant factor.
Except for step $3$, each step can be executed in $\sort{E}$ I/Os.
Next, we show that step $3$ can be executed in $O(k \cdot \msf{V,E})$ I/Os.

A $k$-edge-certificate of $G$ is a spanning subgraph $H$ of $G$ such that
    for any two vertices $u$ and $v$, and for any positive integer $k' \leq k$,
there are $k'$ edge disjoint paths between $u$ and $v$ in $H$ if and only if
    there are $k'$ edge disjoint paths between $u$ and $v$ in $G$.
It is called sparse, if $E(H) = O(k V)$.
There is one simple algorithm, given in \cite{NI92a}, which computes a sparse
    $k$-edge connectivity certificate of graph $G$ as follows.
Compute a spanning forest $F_1$ in $G$; then compute a
    spanning forest $F_2$ in $G - F_1$; and so on;
continue like this to compute a spanning forest $F_i$ in $G - \cup_{1 \leq j < i} F_j$,
    until $F_k$ is computed.
It is easy to see that connectivity of graph $H = \cup_{1 \leq i \leq k} F_i$ is
    at most $k$ and the number of edges in $H$ is $O(kV)$.
Thus, we can compute a sparse $k$-edge connectivity certificate of graph $G$ in
    $O(k (\msf{V,E} + \sort{E})$ I/Os. 

Since $\lambda_{min}$ is $O(E/V)$ and the number
    of edges is reduced by a constant factor in each iteration, a total
    of $O(\frac{E}{V} \cdot \msf{V,E})$ I/Os are required to compute a cut of 
    value between $c$ and $(2 + \epsilon)c$.

\section{Conclusions from this Chapter}
In this chapter, a minimum cut algorithm was designed exploiting the
semi-duality between minimum-cut
        and tree-packing.
On sparse graphs, the I/O complexity of the second phase dominates.
Computing the second phase of our algorithm in $O(\sort{E})$ I/Os
instead of
        $O(\frac{V}{B} \sort{V} + \sort{E})$ would be an interesting nontrivial result.

An approximate algorithm given in this chapter executes faster than the
above.
Can we improve the error in minimum cut further without compromising much
on the I/O complexity?

   \chapter{Some Lower and Upper Bound Results on Interval Graphs}
\label{ig:chapt}
\section{Introduction}

External memory algorithms for restricted classes of graphs such as planar graphs, grid graphs,
	and bounded treewidth graphs have been reported. 
Special properties of these classes of graphs make it easier on them, in comparison to general graphs, to find 
	algorithms for fundamental graph problems such as single source
	shortest paths, breadth first search and depth
	first search.
A survey of results and references can be found in \cite{Vi08}.

In a similar vein, we study interval graphs in this chapter.   
We present efficient external
	memory algorithms for the single source shortest paths,
	optimal vertex colouring,
	breadth first search and depth first search problems on interval graphs.
Note that optimal vertex colouring is  
	NP-hard for general graphs.
We give I/O lower bounds for the minimal vertex colouring of interval graphs, 
	$3$-colouring of doubly linked lists, finding of the connected
	components in a set of monotonic doubly linked lists, and $2$-colouring of a set of  
	monotonic doubly linked lists. 

\subsection{Definitions}
\label{subsec:prem:ig} 
A graph $G=(V,E)$ is called an interval graph,
    if for some set $\Im$ of intervals of a linearly ordered set,
    there is a bijection $f:V\rightarrow\Im$
    so that two vertices $u$ and $v$ are adjacent in $G$ iff $f(u)$ and
    $f(v)$ overlap.
Every interval graph has an interval representation in which endpoints
    are all distinct \cite{G80}.
Hence, $G$ can be represented by a set of endpoints $\cal E$
    of size $2\mid\Im\mid$, where, for each
    $I\in\Im$, there are unique elements $l(I)$ and $r(I)$ in
    $\cal E$ corresponding
    respectively to the left and right endpoints of $I$.
We define  an ``inverse'' function ${\cal I}:{\cal E}\rightarrow\Im$,
    which gives the corresponding interval for each member of ${\cal E}$.
That is, if $e$ is either $l(I)$ or $r(I)$ and $I\in\Im$ then, ${\cal I}(e)=I$.
We say that an interval $I_{1}$ leaves an interval $I_{2}$ to the left (resp. right) if 
$l(I_1) < l(I_2) < r(I_1) < r(I_2)$ (resp., $l(I_2) < l(I_1) < r(I_2) < r(I_1)$).

Interval graphs are perfect graphs \cite{G80};
    that is, for an interval graph $G=(V,E)$, and for
    every induced subgraph $G'$ of $G$, the chromatic number
    of $G'$ is equal to the clique number of $G'$.
Interval graphs also form a subclass of chordal graphs \cite{G80};
	that is, every cycle of length greater than $3$ has a chord,
	which is an edge joining two vertices that are not adjacent in 
	the cycle.
Vertex colouring of a graph means assigning colours to its
vertices
    so that no two adjacent vertices get the same colour; this is minimal
    when the smallest possible number of colours have been used.
    A maximum clique is a largest subset of vertices in which each pair
    is adjacent.
    A clique cover of size $k$ is a partition of the vertices $V = A_1 + A_2$
        $\ldots + A_k$ such that each $A_i$ is a clique.
    A smallest possible clique cover is called a minimum clique
        cover.

The single source shortest paths (SSSP) problem on interval graphs 
is to computing the shortest paths from a given \emph{source} vertex (interval) 
to all other vertices, where each vertex is assigned a positive weight; 
the length of a path is the sum of the weights of the vertices on the path, including
its endpoints.
 
A BFS tree of a graph $G$ is a subtree $T$ of $G$ rooted at some vertex $s$ such 
	that for each vertex $u$, the path from $s$ to $u$ in $T$ is
	a path of minimum number of edges from $s$ to $u$ in $G$.
A DFS tree of an undirected graph $G$ is a rooted subtree of $G$ such that 
	for each edge $(u,v)$ in $G$, the least common ancestor of $u$ and
	$v$ is either $u$ or $v$.
The breadth first search and depth first search problems are to compute a BFS tree and a DFS tree 
	respectively.

A doubly linked list is a directed graph in which both the out-degree and in-degree of
a vertex can be at most two.
The vertices are given in an array (say $A$) and each vertex has a pointer to 
	the next vertex and previous vertex.
Thus $A[i] = (k,j)$ will mean that $j$ is a successor of  $i$ and $k$ is a
	predecessor of $i$. 
	If $j>i$ and $j$ is the successor of $i$ then the $i$-to-$j$ pointer is said to be 
	``forward'', otherwise
	(if $j<i$ and $j$ is successor of $i$) it is said to be ``backward'' pointer.
A stretch of forward (backward) pointers starting from say $i$ is a maximal collection
	of nodes which can be reached from $i$ by traversing only forward (backward)
	pointers.
Any doubly linked list will have alternate stretches of forward and backward pointers.    
$3$-colouring a doubly linked list(denoted as ``$3$LC'') is the problem of
    vertex colouring a list with $3$
    colours.

Assume that the
vertices of a doubly linked list are
    numbered in some order. The doubly linked list is called  monotonic if the vertex
    numbered $v$ can be a successor of the vertex numbered $u$ if and
    only if $u < v$.
    That is, when the vertices of the list are arranged left to right in
    increasing order of their numbers, all the links are from left to right for
	successor pointers and right to
	left for predecessor pointers.
    MLCC is the problem of labelling each node  in a collection of
    disjoint monotonic doubly linked lists, by the first element of the
    corresponding list and $2$MLC is the problem of vertex colouring each 
	monotonic doubly linked list with 2 colours.

\subsection{Previous work}
External memory algorithms have been designed for many fundamental graph problems. 
All known external memory algorithms for the single source shortest paths (SSSP), 
	breadth first search (BFS), and depth first search (DFS) problems on general graphs perform well 
	only for dense graphs. See the results and references
in \cite{Vi08}; also see Table~\ref{result:previous5}. 
For many graph problems a lower bound of 
       	$\Omega(\min \{V, \frac{E}{V}\sort{V} \})$ on I/Os applies \cite{MR99}.

Some restricted classes of sparse graphs, for example 
planar graphs, outerplanar graphs, grid graphs and bounded tree width graphs, have
been considered in designing I/O efficient 
algorithms for SSSP, DFS and BFS. 
Exploitation of the structural properties of specific classes of sparse graphs has led to algorithms 
for them that
	perform faster than the algorithms for a general graphs.
Most of these algorithms require $O(\sort{V+E})$ I/Os.  
See the results and references in \cite{Vi08}.

For any constant $k$, the set of graphs of tree width $k$
includes all interval graphs with maximum clique size $k$. For
bounded tree width graphs, we can compute SSSP, BFS and DFS, and also some other problems
that are known to be NP-hard for general graphs 
in $O(\sort{E+V})$ I/Os
\cite{MaZe07}.

Interval graphs form a well-known subclass of perfect graphs and have applications 
	in archeology,
biology, psychology, management, engineering, VLSI design, circuit
routing, file
        organisation, scheduling and transportation \cite{G80}.
A lot of work have been done in designing sequential and parallel algorithms
	for various problems on an interval graph 
\cite{ACL95, D97, G80, H92, SaSa99}.

List ranking is a well-known problem. A known lower bound on I/Os for this problem is 
	$\Omega(\perm{N})$ \cite{ChGoGr+95}, where $N$ is the size of 
	the linked list and $\perm{N}$ is the number of I/Os required to 
	permute $N$ elements.
The best known upper bound on I/Os for list ranking and $3$ colouring of lists is
	$O(\sort{N})$ \cite{ChGoGr+95}.
\begin{table}
\begin{tabular}{|l|l|l|}
               \cline{1-2}
               \hline 
               {\bf Problem} & {\bf Result} & \bf{ References} \\
               \hline \hline
               Single source shortest  & $O\left(\sqrt{\frac{VE}{B}} \log V + 
\msf{V,E} \right)$ & \cite{MeZe06}  \\
                problem &  &      \\
               \hline
               Breadth first search & $O\left(\sqrt{(VE)/B} + \sort{E} + \ssf{V,E}\right)$ 
 & \cite{MeMe02}\\
                \hline               
Depth first search & $O\left(\min\{V + \sort{E} + (VE)/M , (V + E/B)\log V\}\right)$
 & \cite{ChGoGr+95} \cite{KS96} \\
               \hline 
            \end{tabular}
\caption{Previous Results: The term \ssf{V,E} and \msf{V,E} represent the I/O bounds
        for computing spanning forest and minimum spanning forest respectively.}
\label{result:previous5}
\end{table}
\subsection{Our Results}
We present some the lower  and upper bound results on interval graphs.
The results are described below and summarised in Table~\ref{result:our5}.

\subsubsection{Lower Bound Results}

We show that
finding the connected components in a collection of disjoint
monotonic doubly linked lists (MLCC) of size $V$ is equivalent to the minimal interval
graph colouring (IGC) problem on an interval graph whose interval representation is
	given.
The number of I/Os needed for both are shown to be $\Omega(\frac{V}{B} \log_{M/B} \frac{\chi}{B})$, 
where $\chi$ is the chromatic
    number of an interval graph, or the total number of disjoint monotonic doubly linked lists, as is relevant.
We also show that 3-colouring of a doubly linked list ($3LC$) of size $V$ is
reducible to  2-colouring of a set of disjoint monotonic doubly linked lists (2MLC)  in 
$O(\scan{V} + \sort{\chi})$ I/Os.
It is also shown that 2MLC and 3LC of sizes $V$ each have lower bounds of 
$\Omega(\frac{V}{B} \log_{M/B} \frac{\chi}{B})$ on
I/Os, where $\chi$ is the number of disjoint 
	monotonic  doubly linked lists, and the total number of forward and backward stretches
	in the doubly linked list respectively.
\subsubsection{Upper Bound Results}
\begin{itemize}
\item {\bf SSSP and BFS/DFS tree computations:} We
present an SSSP algorithm that requires $O(\sort{V})$
    I/Os, and an BFS tree computation algorithm that requires $O(\scan{V})$ I/Os, and
     an DFS tree computation algorithm that requires $O(\frac{V}{\chi}\sort{\chi})$ I/Os.
The input graph is assumed to be represented as a set of intervals in sorted
	order.
\item {\bf minimally vertex colouring interval graphs (IGC):}
We show that IGC can be computed in
an optimal $O(\frac{V}{B}\log_{M/B}\frac{\chi}{B})$ I/Os, if
the input graph is represented as a set of intervals in sorted
	order. 
\item {\bf Algorithms for 3LC, 2MLC, and MLCC Problems}
Optimal algorithms are given for 3LC, 2MLC, MLCC problems.
\end{itemize}

\begin{table}
\begin{center}
\begin{tabular}{|l|l|c|}
\cline{1-3}
\hline
{\bf Problem }       &       {\bf  Notes }           &  {\bf I/O Bound} \\
\hline \hline
SSSP           &   input is a set of intervals   & $O(\sort{V})$           \\
\hline
BFS            &   input is a set of intervals in sorted order  & $O(\scan{V})$            \\
\hline
DFS            &   input is a set of intervals in sorted order & $O(\frac{V}{\chi} \sort{\chi})$            \\
\hline
IGC            &   input is a set of intervals in sorted order    &   $O\left(\frac{V}{B} \log_{M/B} 
\frac{\chi}{B}\right)$ \\
\hline
3LC, 2MLC, MLCC      &                            & $O(\frac{V}{\chi}\sort{\chi})$ \\
\hline

IGC, MLCC, 3LC, 2MLC      &                           &     $\Omega\left(\frac{V}{B} \log_{M/B}
\frac{\chi}{B}\right)$ \\
\hline
\end{tabular} 
\caption{Our Results}      
\label{result:our5}
\end{center}
\end{table}

\subsection{Organisation of This Chapter}
In Section~\ref{igc}, we present the lower bound results on MLCC, IGC
	and 3LC problems.
In Section~\ref{upper:lower:igc}, the algorithms for finding chromatic number and
	minimal vertex colouring problems are given.
The algorithms for SSSP, BFS and DFS
are given respectively in Section~\ref{shortest:path},
and Section~\ref{bfs:dfs}.

\section{The Lower Bound Results}
\label{igc} 

In this section, we discuss the lower bound results for MLCC, 2MLC, 3LC, and IGC problems.

\subsection{Equivalence of MLCC and IGC}

\subsubsection{A Reduction from MLCC to IGC}

\label{mlcc:igc} Consider an instance of MLCC of size $V$ with $K$
	components, given in an array $A[1\ldots V]$ such that for each $A[i]$, 
	its successor and predecessor both are stored in $A[i]$.
From this instance we construct an equivalent instance of IGC as follows.
 
Allocate an array $D[1\ldots 2K+V]$.
Copy $A[1\ldots V]$ into $D[K+1, \ldots K+V]$. While copying, offset every pointer
by $K$ so that it continues to point to the same element as before.
Scanning $D$ in order from location $K+1$ to $K+V$, for $1\leq j\leq K$, make $D[j]$ the
predecessor of the $j$-th headnode in $D[K+1, \ldots K+V]$, and make $D[K+V+j]$ the successor of the
$j$-th lastnode in $D[K+1, \ldots K+V]$. Accordingly, define successors and predecessors 
for the nodes in $D[1\ldots K]$ and $D[K+V+1\ldots 2K+V]$.
(A headnode is a node without a predecessor, and a lastnode
is a node without a successor.) Now $D$ holds in it a set of $K$ monotonic linked lists
so that every headnode is in $D[1\ldots K]$ and every lastnode is in $D[K+V+1\ldots 2K+V]$.
$O(\scan{V})$ I/Os are enough to prepare $D$ as stated above.

Construct a set $\Im$ of intervals such that,
    for $1\leq i,j\leq V+2K$, if the successor of $D[i]$ is $D[j]$ then
    add $I_i=[i,j-\frac{1}{2}]$ to $\Im$. $l(I_i) = i$, and $r(I_i) = j - \frac{1}{2}$.
For every integer $i$ in $D[1\ldots K+V]$, $i$ is a left endpoint, and for every
integer $j$ in $D[K+1\ldots 2K+V]$, $j-\frac{1}{2}$ is a right endpoint.
The successor (predecessor) pointers give the corresponding the right (left) endpoint
for each left (right) endpoint. An instance of IGC can be constructed from these
intervals in $O(\scan{V + 2K})$ I/Os.
 
Consider the interval graph $G$ defined by $\Im$. Identifying each
component of $D$ with
    a unique colour, we get a valid colouring of $G$;
    that is, $\chi(G)\leq K$.
Now, suppose, there is an optimal colouring of $G$
    that for two consecutive edges $(u,v)$ and $(v,w)$ of $D$,
    gives different colours to their corresponding intervals
    $I_u=[u,v-\frac{1}{2}]$ and
    $I_v=[v,w-\frac{1}{2}]$ in $G$.
But, all endpoints are distinct. There is no endpoint between
    $v-\frac{1}{2}$ and  $v$. 
Both $I_u$ and $I_v$ share $K-1$ mutually adjacent neighbours.
Therefore, 
    and we get $\chi(G)= K+1 > K$, a contradiction.
That is, an optimal colouring of $G$ will use one colour per component of $D$.
In other words, $\chi(G)=K$ and any optimal colouring of $G$
    will identify the connected components of $D$, and hence of $A$.
Thus, we have the following lemma.

\begin{lemma}
\label{lem:mlcc2igc:nc1} An instance of MLCC of size $V$ with $K$
components can be reduced
    to an instance of IGC of size $O(V)$ and chromatic number $K$,
    in $O(\scan{V})$ I/Os.
\end{lemma}

\subsubsection{A Reduction from IGC to MLCC}
\label{igc:mlcc}

A in-core algorithm for IGC is easy to visualise.
Let $Q$ be a queue of size $\chi(G)$, which is
    initialised with all the available $\chi(G)$ colours.
Consider the endpoints of the intervals one by one in non-decreasing
    order.
For each left endpoint encountered, remove a colour from the front
    of $Q$ and colour the corresponding interval with it.
For
    each right end point, release the colour of the corresponding interval
    onto the back of $Q$.
When the last left endpoint is considered, the graph would be coloured.

We attempt a reduction of IGC to MLCC using this algorithm. Let
$L=\{l_{1},\ldots,l_{V}\}$ and $R=\{r_{1},\ldots,r_{V}\}$
    respectively be the sets of the left and right endpoints of the
    intervals given in non-decreasing order.
For each interval $I$, the rank of $r(I)$ (resp. $l(I)$) is stored with $l(I)$ 
	(resp. $r(I)$) in $L$ (resp. $R$).
For $1\leq i\leq V$, let $t_{i}$ be the rank of $r_{i}$ in $L$.
That is, $l_{1}<\ldots<l_{t_{i}}<r_{i}<l_{t_{i}+1}$. So, when
$r_{i}$ releases the colour of the interval ${\cal I}(r_{i})$
    onto the back of $Q$, $t_i$ left endpoints and $i$ right endpoints
    would have been encountered and the length of $Q$ would be
    $\chi(G)-t_{i}+i$.
Hence, the colour released by  $r_{i}$ will be taken up by the
    $(\chi(G)-t_{i}+i)$-th left endpoint from now on;
    that is, by the left endpoint $l_{t_{i}+\chi(G)-t_{i}+i}=l_{\chi(G)+i}$.
In other words, both ${\cal I}(r_{i})$ and
    ${\cal I}(l_{\chi(G)+i})$  are to get the same colour, and no interval
    with a left endpoint between their respective left endpoints will get
    that colour.

Define the successor and predecessor of each $L[i]$ as follows: let
        $\mbox{succ}(l_i)$ be $l_{\chi(G)+j}$, for $1\leq i\leq V$, where 
	${\cal I}(l_i) = {\cal I}(r_j)$ and $\mbox{pred}(l_i)$ be 
	$l({\cal I}(r_{i - \chi(G)}))$, for $\chi(G) \leq i\leq V$.
For $1 \leq i \leq \chi(G)$, each $l_i$ is the headnode of a monotonic linked
	list.
It is clear that this defines a collection of monotonic linked lists over $L$.
Once we find the connected components in this collection we have got an optimal 
	colouring of the graph.
We can obtain monotonic linked lists in $O(\scan{V})$ I/Os as follows: for each 
	interval $I$, rank of $r(I)$ (resp. $l(I)$) is stored with 
	$l(I)$ (resp. $r(I)$) in $L$ (resp. $R$).
The successor of each $l_i$ can be computed in
    one scan of $L$.
To compute the predecessors, scan $L$ and $R$ together, with the scan of
$R$ staying $\chi(G)$ nodes behind, and set for each $l_i$, $l({\cal I}(r_{i - \chi(G)}))$ 
as the predecessor of $l_i$. 
Thus, we have the following lemma:

\begin{lemma}
\label{lem:igc2mlcc:nc1}
IGC can be reduced to MLCC in $O(\scan{V})$ I/Os.
\end{lemma}

Therefore,

\begin{theorem}
\label{th:igcEQmlcc} IGC is equivalent to MLCC on the external memory model.
\end{theorem}

\subsection{Lower Bounds for 2MLC and MLCC}
The lower bounds for 2MLC and MLCC are obtained by showing each problem
	equivalent to a modified split proximate neighbours problem.
The split proximate neighbours (SPN) problem is defined as follows: ``Given are two arrays $A$ and $B$, each
	a permutation of size $K$ of elements from the range $[1,K]$ so that for each $i$ in $[1,K]$,
	the occurrences of $i$ in $A$ and $B$ know the addresses of each other.
	Permute $A\cup B$ so that for each $i$, the two occurrences of $i$
	are stored in the same block.''
Without loss of generality, we can assume that $A$ is sorted. 
The following lower bound is known for this problem:
\begin{lemma}\cite{ChGoGr+95, zeh}
\label{prox:neigh:lb}
The split proximate neighbours problem requires $\Omega(\perm{K})$ I/Os for an input of 
	size $2K$, where $\perm{K}$ is the number of I/Os required to permute
	$K$ elements.
\end{lemma}

The assumption that for each $i$, the occurrences of $i$ in $A$ and $B$ know the addresses of each other
does not make the problem any easier. This assumption is akin to assuming that in the
permutation problem each element in the input knows its position in the output. 
See \cite{AgVi88, zeh}. 

Consider a set $S$ of $\frac{N}{2K}$ independent instances of SPN of size $K>M$ each.
Let $\mathcal{P}$ be the problem of solving all instances of $S$.
A lower bound for $\mathcal{P}$ is $\Omega(\frac{N}{K}\perm{K})$.
This follows from Lemma~\ref{prox:neigh:lb}, and the fact that $N$ cannot be greater than $K!$
	because $N < B(M/B)^B$ when $\sort{N} < N$. 

Create a set of $K$ monotonic linked lists from $\mathcal{P}$ as follows.
For all $i$, $1\leq i\leq K$, and for all $j$, $1 < j < N/2K$, 
let the successor (resp. predecessor) of the $i$ in $A_j$
	be the $i$ in $B_j$ (resp. $B_{j-1}$), and let the successor (resp. predecessor) 
	of the $i$ in $B_j$ be the $i$ in $A_{j+1}$ (resp. $A_j$).
For all $i$, $1\leq i\leq K$, let the successor (resp. predecessor) of the $i$ in $A_{N/2K}$ (resp. $A_{1}$)
	be NULL, and let the predecessor (resp. successor) of the $i$ in $A_{N/2K}$ (resp. $A_{1}$)
	be the $i$ in $B_{N/2K-1}$ (resp. $B_{1}$).    

Consider solving the above MLCC instance $C$. 
We claim that during the execution of any MLCC algorithm, every link of $C$ has to come into the main memory. 
Suppose this is wrong.
Then an adversary can cut an edge $e$ that does not
come into the main memory; the algorithm would still give the same output, 
even though the number of connected components has increased. 

Whenever the link from the $i$ in $A_j$ to the $i$ in $B_j$ comes into the
main memory, make a copy of the two occurrences of $i$ into a block.
When block becomes full, write it into the external memory.
Thus, with an extra $N/B$ I/Os, problem $\mathcal{P}$ is also computed during the execution of
	an algorithm for MLCC.
Thus the lower bound of $\mathcal{P}$ also holds for MLCC.

The same argument holds for 2MLC, and 3LC too.
Therefore,
\begin{theorem}
Each of IGC, 3LC, MLCC and 2MLC require  
$\Omega(\frac{N}{K} \sort{K})$ I/Os on the external memory model.
\end{theorem}

\section{The Algorithms for $3$LC, $2$MLC, and MLCC}
\label{3lc:2mlc}

Note that each node of the given doubly linked list knows both its predecessor and successor.
A general doubly linked list $L$, given in an array $A$, can be visualised as being constituted of
    alternating stretches of forward and backward pointers in an array (described in
        Subsection~\ref{subsec:prem:ig}), and
    hence can be decomposed into two instances of $2$MLC, one consisting only
    of forward stretches and the other consisting only of backward stretches.
Let $L'$ and $L''$ be the $2$MLC instances formed by the forward and backward
    stretches of $L$, respectively.
Invoke a $2$MLC algorithm on each.
The output of the $2$MLC invocations can be used to $4$-colour $L$ as follows:
for each node $x$ in $L$, colour $x$ with ``$ab$'' if 
    $a$ and $b$ are the colours ($0$ or $1$) of $x$ in $L'$ and $L''$ respectively. 

Each node $v$, if $v$ has both a successor and a predecessor in $L'$,
	can infer their colours in $L'$ from its own colour;
	if $v$ is coloured $0$, they are coloured $1$ and vice versa.
There are at most $2\chi$ nodes in $L'$ with a forward predecessor or a backward successor,
where $\chi$ is the total
        number of forward and backward stretches in $L$.
In one pass, isolate these nodes, sort them and inform each of them the colours
	of its successor and predecessor.
Now every node knows the colours
	of its successor and predecessor in $L'$. Process $L''$ in a similar fashion.
All of this takes only $O(\scan{V} + \sort{\chi})$ I/Os.

A $4$-colouring can be reduced to a $3$-colouring in one scan of $A$ as follows:
    for each vertex of colour $4$, assign it the smallest colour
not assigned to its predecessor or successor.
This can be done in a single scan of $A$.

Thus, $3$LC can be reduced to $2$MLC in $O(\scan{V} + \sort{\chi})$ I/Os.

$2$MLC can be solved as follows:
Maintain a priority queue $PQ$. Scan the array $A$.
For each node $A[i]$, if $A[i]$ is a headnode, then give it a colour $c$ of $0$;
otherwise, perform deletemin on $PQ$ to know the colour $c$ of $A[i]$;
if $A[i]$ is not a lastnode, then
give colour $1-c$ to its successor $A[j]$,
and insert $1-c$ with key $j$ into $PQ$.
At any time, $PQ$ holds at most $\chi$ colours. The total number of I/Os needed is, 
therefore, $O(\frac{V}{B} \log_{\frac{M}{B}} \frac{\chi}{B})$.
The same algorithm can be used in computing MLCC also with a change in colouring each headnode
	uniquely, and sending the same colour to successor.

That is, a linked list can be 3-coloured in 
$O(\frac{V}{B} \log_{\frac{M}{B}} \frac{\chi}{B})$ I/Os, where $\chi$ is the total
        number of forward and backward stretches in $L$.

\section{An Interval Graph Colouring Algorithm }
\label{upper:lower:igc} 

First, we consider a problem closely
related to IGC, namely,
    that of finding the chromatic number of an interval graph
    with a known interval representation.
\subsection{Finding the Chromatic Number}
\label{chromatic}
Let an interval graph $G=(V,E)$ be represented by
    $L=\{l_{1},\ldots,l_{V}\}$
    and
    $R=\{r_{1},\ldots,r_{V}\}$
    the left and right endpoints of the intervals  of $G$
    given in non-decreasing order.
The sequences $L$ and $R$
    can be cross ranked in one scan.
For $1\leq i\leq V$, let $s_{i}$ be the rank of $l_{i}$ in $R$.
That is, $r_{1}<\ldots<r_{s_{i}}<l_{i}<r_{s_{i}+1}$. Then,
$i-s_{i}$ ($=t_i$ say), is the size of the clique formed precisely
by
    those intervals that contain $l_{i}$.
In other words, $t_{i}$ is the size
    of the clique that ${\cal I}(l_{i})$ forms along with its in-neighbours,
	if each edge thought to be directed away from the interval with the 
	smaller left endpoint.
Observe that any maximal clique $\cal C$ of an interval graph should
    contain a vertex $v$, such that the set of all in-neighbours of $v$
    is precisely the set ${\cal C}-\{v\}$.
Thus, the chromatic number of $G$ can be obtained by
    taking the maximum of $t_{i}$ over all $i$ and we have following lemma.
\begin{lemma}
\label{lem:chrnocnstnt}
The clique number of an interval graph can be found in
    $O(\scan{V})$ I/Os on the External Memory model, provided,
    the left and right endpoints of
    the intervals are given in sorted order, separately.
\end{lemma}

\subsection{The IGC Algorithm}
Our algorithm is based on the sequential algorithm given in Subsection~\ref{igc:mlcc}.
We assume that left and right endpoints of all intervals are stored together in
    an array $A$ in sorted order, and for each interval $I$, the left end point $l(I)$ knows the right
    endpoint $r(I)$.
The algorithm executes the following steps:
\vspace{0.3in}
\hrule
\vspace{0.1in}
\begin{enumerate}
\item Find the chromatic number $\chi(G)$ using the algorithm given in Section~\ref{chromatic}.

\item Initialize an external memory queue $Q$ \cite{P03} with all the available $\chi(G)$ colours.

\item Read the endpoints from $A$ in order.

\item For each left endpoint $l(I)$ encountered in $A$, remove a colour from the front of $Q$,
    colour interval $I$ with it, and insert this colour with
    the right endpoint $r(I)$ into a priority queue $PQ$
    with $r(I)$ as the key.

\item For each right endpoint $r(I)$ encountered in $A$, obtain the colour of $I$ from $PQ$ using a deletemin
    operation, insert this colour onto the back of $Q$.
\end{enumerate}
\vspace{0.1in}
\hrule
\vspace{0.3in}
When the last of the endpoints is considered, graph would be coloured.
The sizes of $PQ$ and $Q$ can be at most $\chi(G)$ at any time.
The I/O complexity of $PQ$ operations dominates the I/O complexity of $Q$ operations.
There are $V$ insertions and $V$ deletions.
Therefore Algorithm colours the graph in
    $O(\frac{V}{B} \log_{\frac{M}{B}} \frac{\chi(G)}{B})$ I/Os \cite{Arge03}.
The correctness of the algorithm comes from the in-core algorithm given in section~\ref{igc:mlcc}.
\begin{lemma}
An interval graph can be coloured in $O(\frac{V}{B} \log_{\frac{M}{B}} \frac{\chi(G)}{B})$
    I/Os, provided that the left and right endpoints of the intervals are given in sorted order,
    and each left endpoint knows the rank of the corresponding right endpoint, 
	where $V$ is the total number of intervals.
\end{lemma}

\section{Single Source Shortest Paths}
\label{shortest:path} For the shortest paths problem, we consider a weighted set of 
	intervals $\Im$ such that each interval $I_i$ is assigned a positive weight
        $w(I_i)$.
A path from interval $I_i$ to $I_j$ is a sequence $\rho = \{I_{v_1},\ldots I_{v_l}\}$
        of intervals in $\Im$, where $I_{v_1} = I_1$, $I_{v_l}=I_j$, and any two consecutive 
	intervals, $I_{v_r}$ and $I_{v_{r+1}}$ overlap for every $r \in \{ 1, \ldots l-1\}$.
The length of $\rho$ is the sum of the weights of its intervals.
$\rho$ is a \emph{shortest} path from $I_i$ to $I_j$,
        if it has the smallest length among all possible paths between $I_i$ to $I_j$ in $\Im$.
The single source shortest paths problem is that of computing a shortest path
        from a given \emph{source} interval to all other intervals.

We assume, without loss of generality, that we are computing the shortest paths
	from the source interval to only those intervals that end after the source begins
	(in other words, have their right endpoints to the right of the left endpoint of the source). 
The same algorithm can be run with the direction reversed to compute the shortest paths
	from the source interval to those intervals that begin before the source ends;
the neighbours of the source feature in both sets, and will get assigned the same values in
both computations; the two computations together will give the shortest paths to all
intervals from the source.

Computing, thus, only in one direction, we can also assume that every interval that intersects
the source interval $I_s$ begins after $I_s$. This would enable us to assume, 
without loss of generality, that $I_s$ has the smallest left endpoint of all intervals. 

\subsection{The Algorithm}
As discussed above, we assume that $I_s$ has the smallest left endpoint.
Suppose all intervals are stored in an array $A$ in sorted order of left endpoints.
(Otherwise, sort them in $O(\sort{V})$ I/Os).
Let the intervals in sorted order be $I_s=I_1, \ldots , I_V$, where $I_i = (l_i , r_i)$
	and $l_i$ (resp. $r_i$) denotes the left (resp. right) endpoint of $I_i$.
Let $d_i$ be the shortest distance of an interval $I_i$ from  the source
    interval $I_1$.

Initially, set $d_1 = w_1$ for $I_1$. For each $I_i$ such that 
	$l_1 < l_i < r_1$, set $d_i = d_1 + w(I_i)$, and the parent of $I_i$ to $I_1$.
If $r_i > r_1$, then insert $\langle d_i, I_i\rangle$ in a priority queue $PQ$ with $d_i$ as the key.
Then execute the following steps.
\begin{tabbing}
aaaa \= aaaa \= aaaa \= aaaa \= aaaa \= aaaa \kill \\
set $r=r_1$; \\
while $PQ$ is not empty \\
\>perform a deletemin operation on $PQ$; let $\langle d_i,I_i\rangle$ be the record returned; \\
\>for every $I_j \in A$ such that $r < l_j < r_i$ \\
\>\> set $d_j = w(I_i) + w(I_j) $, and the parent of $I_j$ to $I_i$ \\
\>\> If $r_j > r_i$ then \\
\>\>\> insert $d_j$ with $I_j$ in $PQ$ with key value $d_j$; \\
\>set $r=\max\{r,r_i\}$; \\
\end{tabbing}
It is easy to see that random access is not required in $A$, because after each deletemin
the scan of $A$ resumes at $r$. 
Since, each interval is inserted in $PQ$ exactly once, we perform $V$ deletemin and 
	insert operations in $PQ$. 
The amortised I/O complexity of each operation is 
$O(\frac{1}{B} \log_{\frac{M}{B}} \frac{V}{B})$ (See Chapter~\ref{emsh:chapt}).
Thus, shortest paths can be computed in $O(\sort{V})$ I/Os.

The correctness of the algorithm follows from this lemma:
\begin{lemma}\cite{ACL95}
\label{path}
If there exists a shortest path from $I_1 \leadsto I_i$ then all intervals
    in this path covers a contiguous portion of the line from $l_1$
    to $r_i$.
\end{lemma}

\begin{lemma}
For each interval $I_i$, $d_i$ is the shortest distance from the source
    interval $I_1$ to
    $I_i$.
\end{lemma}
\begin{proof}
For two intervals $I_j$ and $I_{j'}$, our algorithm chooses $I_j$ as 
the parent of $I_{j'}$ because it is the nearest to $I_1$ (has the smallest $d$-value) of all
	intervals that contain the left endpoint of $I_{j'}$.
Therefore, if $d_{j}$ is the shortest distance to $I_{j}$, then
	$d_{j'}$ is the shortest distance to $I_{j'}$. 
The $d$-values of all neighbours of $I_1$ indeed match their shortest distances from $I_1$.
\end{proof} 

While computing shortest distances, our algorithm computes
the shortest paths tree also.

Given a set of $V$ weighted intervals and a source interval, a shortest path of each interval from the source
    interval can be computed using our algorithm in $O(\sort{V})$ I/Os.

\section{Breadth First Search and Depth First Search}

\label{bfs:dfs} Here we show how to compute a breadth first search tree
and depth first
    search tree for an interval graph in $O(\scan{V})$ I/Os,
    if intervals are given in sorted order of left
    endpoints i.e. $\Im = \{I_1, I_2, \ldots I_V\}$.
We assume, without loss of generality, that source vertex is $I_1$

\subsection{Breadth First Search}
To begin with, interval $I_1$ is assigned as the parent of the 
intervals whose left
    endpoints are  contained in $I_1$.
Among those selected intervals, we find an interval $I_j$ whose right
    endpoint is maximum.
This interval $I_j$ is assigned as parents of the remaining
intervals whose
    left endpoints are contained in $I_j$, and for whom a parent has not been assigned yet.
We repeat this process until no interval is left. The resulting tree
is a BFS tree. Since we have to scan $\Im$ once, the BFS tree
is computed in  $O(\scan{V})$ I/Os.

\subsection{Depth First Search}
To find a depth first search tree, an interval $I_j$ is assigned as the parent
    of an interval $I_i$, if $I_j$ has the largest left endpoints among
	$I_k$ such that $l_k < l_i < r_k$; that is, the parent of
	an interval is its youngest older neighbour if the intervals are assumed
	to be on the timeline.
Note that this scheme indeed defines a DFS tree, because if the resultant tree has
	a cross edge from  a ``older'' interval $I_u$ to an ``younger'' interval $I_v$, then it can be shown that
	a proper ancestor of $I_v$ that is lower than the lowest common ancestor of $I_u$ and $I_v$
	failed to choose $I_u$ as its parent as it ought to have.
To compute the tree, scan the endpoints from the smallest to the largest,
while maintaining the set of open intervals in a priority queue $PQ$ with left endpoint
as key. Whenever, a left endpoint $l$ is encountered use a Findmax operation on $PQ$ to 
find the youngest interval open at that point in time. After, thus defining the DFS
parent of ${\cal I}(l)$, insert ${\cal I}(l)$ into $PQ$ with key $l$.
Whenever a right endpoint $r$ is encountered, delete ${\cal I}(r)$ from $PQ$. 
Therefore, a DFS tree can be computed in 
$O(\frac{V}{B} \log_{\frac{M}{B}} \frac{\chi}{B})$ I/Os.

\section{Conclusions from this Chapter}
In this chapter, we present lower and upper bound results for many problems on an interval 
	graphs.
Once an interval representation is given for an interval graph, various problems like
	SSSP, BSF and DFS become easier.
External memory algorithms for interval graph recognition, and computing of an interval representation 
	of the given graph are yet to be found.

   \part{Algorithms on the W-Stream Model and its Variants}
  \chapter{Some Algorithms on the W-Stream Model}
\label{wstream:chapt}
\section{Introduction}

In this chapter, we discuss the designing of W-Stream algorithms for sorting, 
	list ranking, and some fundamental tree and graph problems. 
A lower bound of $\Omega(N/(M \log N))$ on passes is known for some
	of these problems \cite{DEM+07, DFR06, R03}, where $N$ is either the number of 
	elements in the input stream or the number of vertices, the latter if the input is
	a graph. 
The number of bits available in the working memory is $M \log N$. 
We shall often refer to the working memory as just the memory.
For sorting and the graph problems considered, we give improved upper bound results. 
For list ranking and the tree problems, we present algorithms that are easier
	than the known ones to implement, and perform as well. 
Lower bounds for some problems are also established. 

The tree problems considered are the finding of Euler tours of trees,
	rooting of trees, labelling of rooted trees,
	and expression tree evaluation.
The graph problems considered are the finding of a maximal
	independent set, a $(\Delta + 1)$ colouring, a maximal matching, 
	a $2$-approximate vertex cover, and $\epsilon$-approximate single 
	source shortest paths.

\subsection{Definitions of the Problems}

\noindent
{\bf Sorting :} The problem is to arrange in order a set of elements drawn from a linearly ordered set.

\noindent
{\bf List ranking :} A linked list is a directed graph in which both the out-degree
	and in-degree of a vertex can be at most one.
A vertex of zero in-degree is called the head of the list.
The vertices are given in an array and each vertex $x$ has a pointer
	to the next vertex $y$ \cite{CoVi89, SaSa99}.
The rank of a vertex $x$ in the list is the number of edges on the path from the head of the list  to $x$. 
The list ranking problem is to compute the rank of each vertex in the list.  

\noindent {\bf Euler tour of a tree:} A tree is an acyclic graph. 
An Euler Tour of a tree is a traversal of its edges that starts and ends
at one distinguished vertex $s$ and traverses each edge exactly twice.

\noindent {\bf Rooting of a tree :}
A tree is rooted, if its edges are directed so that
for each vertex $x$ other than a distinguished vertex $s$ called the root, 
the out-degree of $x$ is one. The out-neighbour of $x$ is the parent of $x$.
The rooting of a tree is the process of so directing the edges.

\noindent {\bf Labelling of a tree :} Some of the labellings we consider are 
preorder numbering, postorder numbering, finding of the depths and the number of 
descendants of each vertex.

\noindent {\bf Expression tree evaluation :} 
In an expression tree, each internal vertex is labelled by a
function which can be computed in $O(1)$ time and each leaf is labelled 
by a scalar. A leaf evaluates to its scalar. The value at an internal node 
is obtained by applying its function to the values of its children.
The problem is to compute the values at all vertices in the tree. 

\noindent
{\bf Maximal independent set:}
Given a graph $G = (V,E)$, a subset $I$ of $V$ is called an
independent set, if no two vertices in $I$ are adjacent in $G$. 
The maximal independent set problem is  to compute an
independent set which is closed under inclusion.

\noindent
{\bf ($\Delta + 1$)-colouring of a graph:}
The problem is to assign one of $\Delta + 1$ colours to each vertex so that
        no two adjacent vertices get the same colour, where $\Delta$ is the maximum
vertex degree of the graph.

\noindent {\bf $\epsilon$-Approximate single source shortest
paths:} Given a weighted graph $G = (V,E)$ with a non-negative
integer weight for each edge
    and a source vertex
    $s$, the single source shortest paths problem (SSSP) is to compute, for each vertex $v \in V$, 
	a shortest path from $s$ to $v$.
A path $P$ from $s$ to $v$ is called an $\epsilon$-approximate
    shortest path, if the length of $P$ is at most $(1 + \epsilon)$ times
    the length of a shortest path from $s$ to $v$.

\noindent
{\bf Maximal matching:}
Given a graph $G = (V,E)$, a subset of $E$ is called a matching, if no two edges of it
    have a common endpoint. The maximal matching problem is to compute
    a matching which is closed under inclusion.

\noindent
{\bf $2$-Approximate vertex cover:}
Given a graph $G = (V,E)$ with positive weight at each vertex, a subset $S$ of $V$ is
    called a vertex cover, if
    each edge has at least one endpoint in $S$.
The $2$-approximate vertex cover problem is to compute a vertex cover of weight
    at most two times the weight of a minimum vertex cover.

\subsection{Previous Results}

On the W-Stream model, sorting has a lower bound of $\Omega(N/(M \log N))$ on passes \cite{R03}. 
The best known algorithms take $O(N/M)$ passes \cite{DEM+07, MP80}. 
Both algorithms \cite{DEM+07, MP80} perform $O(N^2)$ comparisons.
The total Number of comparisons can be reduced to the optimal 
$O(N\log N)$ at the cost of increasing the number of passes to $O((N \log N)/M)$
using a simulation of an optimal PRAM algorithm \cite{DEM+07}. But in this
case, the total number of operations is $\Theta((N^2 \log N)/M)$, where
each reading or writing of an element into the  memory counts as an operation.

List ranking, Euler tour of trees, rooting of trees, labelling of trees and
expression tree evaluation can all be solved in $O(N / M)$
passes using simulations of optimal PRAM algorithms \cite{Ja92}, as shown in \cite{DEM+07}. 

The lower bound for the maximal independent set problem is
$\Omega(V/(M \log V))$ on passes \cite{DEM+07} when the input is 
an unordered edge list. The best known algorithm uses $O(V \log V/M)$ passes to find a maximal
independent set with high probability, when the input graph is presented as an unordered 
edge list.

The lower bound for the SSSP problem is
$\Omega(V/(M \log V))$ on passes \cite{DFR06}, when the input graph is 
presented as an unordered edge list.
The best known SSSP algorithm \cite{DFR06} executes in $O((C V \log V)/ \sqrt{M})$ 
	passes and finds shortest paths with high probability, where $C$ is the maximum weight 
	of an edge, and $\log C = O(\log V)$. 
This algorithm is not very efficient for large weights. 
This algorithm can also be used for performing a breadth first search (BFS) 
	of the graph in $O((V \log V)/ \sqrt{M})$ passes with high probability.
These algorithms assume that the input graph is given as an unordered edge list.

\subsection{Our results}
In this chapter, we present the following results.
\begin{itemize}

\item Lower bounds of $\Omega(N/(M \log N))$ on passes for list ranking and maximal matching.
A lower bound for list ranking also applies to expression tree evaluation, finding the depth of
every node of a tree, and finding the number of descendants of every node in a tree.

\item An algorithm that sorts $N$ elements in $O(N /M)$ passes
while performing $O(N\log M+N^2/M)$ comparisons and $O(N^2/M)$ elemental reads.  
Our algorithm does not use a simulation, and is easier to implement than the 
earlier algorithms.

\item Algorithms for list ranking, and tree problems
such as Euler Tour, rooting of trees, labelling of rooted trees and
expression tree evaluation that use $O(N /M)$ passes each. 
Unlike the previous algorithms, our algorithms are easy to implement as they do not use
simulations. 

\item Algorithms for finding a maximal independent set and a $\Delta + 1$ colouring of graphs.
We show that when the input graph is presented in an adjacency list representation, 
each can be found deterministically in $O(V / M)$ passes. 
We also show that when the input is presented as an unordered edge list, 
each can be found deterministically in 
$O(V /x)$ passes, where 
$x = O(\min\{M, \sqrt{M\log V} \})$ for MIS, and
$x = O(\min\{M, \sqrt{M\log V}, \frac{M\log V}{\Delta\log\Delta} \})$
for $\Delta + 1$ colouring. 

\item  Algorithms for maximal matching and $2$-approximate weighted vertex cover that are
deterministic and require $O(V / M )$ passes.  The vertex cover algorithm assumes that
the weight of each vertex is $V^{O(1)}$.
The input here is assumed to an unordered edge list.
The lower bound of maximal matching problem is shown to be $\Omega(V/(M \log V))$ on passes.

\item An algorithm that, for all vertices $v \in V$, computes with high probability an 
	$\epsilon$-shortest 
path from a given source vertex $s$ to $v$ in $O(\frac{V \log V \log W}{\sqrt{M}})$ passes, 
where $W$ is the sum of the weights of all edges.
We assume that $\log W = O(\log V)$. 
If $C$ is the maximum weight of an edge, then $W \leq VC$, and
our algorithm improves on the previous bound by a factor of $C/\log({VC})$ at the cost a 
small error in accuracy. 
Here again, we assume the input to be  given as an unordered edge list.
\end{itemize}

\subsection{Organisation of This Chapter}

In Section~\ref{lower:bound:wstream}, we prove lower bounds for the list ranking
	and maximal matching problems. In Section~\ref{upper:bound:wstream}, 
we present algorithms for the various problems mentioned above.
In particular, in Subsection~\ref{sorting}, we present a sorting algorithm. In
Subsection~\ref{list:ranking}, we give a list ranking algorithm.
In Subsection~\ref{tree:algo}, we present several tree
algorithms.
In Subsection~\ref{mis:colouring}, we present algorithms for the maximal independent set
	and $(\Delta+1)$ colouring problems. 
In Subsection~\ref{sssp:wstream}, we give an approximate SSSP algorithm,
	and in Subsection~\ref{matching} and \ref{vertex:cover}, we present
	algorithms for maximal matching and $2$-approximate weighted vertex cover, respectively.

\section{The Lower Bound Results}
\label{lower:bound:wstream}

In this section, we prove lower bounds for the list ranking and maximal matching 
	problems.
The lower bound for each problem $P$ is proved by reducing the bit-vector
	disjointness problem ${\cal D}$ to $P$.
Results on the bit-vector disjointness problem in Communication Complexity 
	have been proved useful in establishing lower bound results for  
	problems on the streaming model \cite{HRR99}.
In this problem, two players $A$ and $B$ have bit vectors $a$ and $b$ respectively,
	each of length $N$.
Player $B$ wants to know if there exists an index $i$ such that $a_i=1$ and
        $b_i=1$.
It is known that this problem requires $\Omega(N)$ bits of communication between
$A$ and $B$ \cite{NK}. 
Let $a_1,\ldots, a_N,b_1,\ldots , b_N$  be an instance of ${\cal D}$.
On the W-Stream model, between one pass and the next over the sequence,
at most $O(M\log N)$ bits of information can be transferred between the two halves
of the input. Any algorithm that solves ${\cal D}$ in $o(N/(M\log N))$ passes
would, therefore, cause $o(N)$ bits to be transferred between the two halves.
That is, $\Omega(N/(M\log N))$ is a lower bound on the number of passes
for ${\cal D}$ on the W-Stream model.

First we show a reduction of ${\cal D}$ to
the list ranking problem.
Assume that the input $I_d$ to ${\cal D}$ is a bit sequence
$a_1,\ldots, a_N,b_1,\ldots , b_N$.
Our reduction constructs a set $L$ of lists with
        $\{a_1,\ldots, a_N,b_1,\ldots ,b_N\}$ as its vertex set.
We define the successor function $\sigma$ on the vertices of $L$ as follows:
For $1\leq i\leq N$, $\sigma(a_i)=b_{i}$, and $\sigma(b_i)=a_{i+1}$.
The list can be formed in $O(1)$ passes on the W-Stream model.
Invoke a list ranking algorithm on this list. If the list ranking algorithm
is one that exploits the associativity of addition for its correctness, then
at some point in the algorithm, for every $i$, the link from $a_i$ to $\sigma(a_i)$ 
must be loaded into main memory; at this point in time, the list ranking algorithm
can say if $a_i=b_i=1$; thus, any list ranking algorithm based on the associativity 
of addition can be amended to answer ${\cal D}$ as an aside.
Without this property we cannot guarantee that for every $i$, the link from $a_i$ to $\sigma(a_i)$
must be loaded into main memory. 

Thus the lower bound of ${\cal D}$ applies to list ranking too.

Next we show a reduction of ${\cal D}$ to the maximal matching problem.
We construct a graph whose vertex set is $(a_1, \ldots a_N, b_1, \ldots, b_N, c_1, \ldots, c_N, 
d_1, \ldots, d_N, x_{1}, \ldots, x_{N}, y_{1}, \ldots, y_{N})$ as follows:
Add an edge $(a_i, d_i)$ into the edge set if $a_i=1$, add
    two edges $(a_i, x_{i})$ and $(x_{i}, c_i)$, otherwise.
Add an edge $(b_i, d_i)$                   if $b_i=1$, add
    two edges $(b_i, y_{i})$ and $(y_{i}, c_i)$, otherwise.
A maximal matching of this  graph is of size exactly $2N$ if and
only if  both $A$ and $B$ do not have $a_i=b_i=1$ for any $i$.
Thus, the lower bound of ${\cal D}$ applies to the maximal matching problem too.

A lower bound for list ranking also applies to expression tree evaluation, finding the depth of
every node of a tree, and finding the number of descendants of every node in a tree. Therefore,
we have following lemma.
\begin{lemma}
The problems of list ranking, maximal matching,
expression tree evaluation, finding the depth of
every nodes of a tree, and finding the number of descendants of every node in a tree
all require $\Omega(N/(M \log N))$ passes on the W-Stream model.
\end{lemma}

\section{The Upper Bound Results}
\label{upper:bound:wstream}
\subsection{Sorting}
\label{sorting}
Unlike the previous algorithms, our algorithm does not use a PRAM simulation, and hence is easier
to implement.

First, in one pass, we create $N/M$ sorted sequences called runs, each of size $M$;
the last run may be of a smaller size, if $N$ is not a multiple of $M$.  
The total number of elemental reads and comparisons are $\Theta(N)$ and $\Theta(N \log M)$,
	respectively. 

Next, we perform $\log(N/M)$ merge-phases, each of which $2$-way merges consecutive
pairs of runs in its input.

At the beginning of the $k$-th merge-phase, we have $N/2^{k-1}M$ runs 
\[A_1, B_1, A_2, B_2, \ldots, A_{N/2^kM}, B_{N/2^kM}\]
where each $A_i$ and $B_i$ is of size $2^{k-1}M$, except that the last 
	run $B_{N/2^kM}$ may be of a smaller size.   
We merge pairs $(A_i, B_i)$ concurrently and store the partial result in 
	an array $C_i$ which is kept between $B_i$ and $A_{i + 1}$.
Initially $C_i$ is without elements. 
For each pair $A_i$ and $B_i$, the following steps are executed:
\begin{quote}
Read $M/2$ elements from $A_i$ into the  memory; let the rest of $A_i$ stream through into the output.
Read $M/2$ elements from $B_i$ into the  memory; let the rest of $B_i$ stream through into the output.
Merge these elements in-core. Keep $D_i$, the merged array, in the  memory. 
Merge $D_i$ with $C_i$, on the fly, as $C_i$ streams in, and output the merged sequence as the new $C_i$.
\end{quote}
In one pass, for all $i$, we move $M/2$ elements each of $A_i$ and $B_i$ into $C_i$, while keeping
all three in sorted order. The size of $C_i$ increases by $M$, and the sizes of $A_i$ and $B_i$ decrease by $M/2$ each.
After $2^k$ passes, $A_i$'s and $B_i$'s become empty, and we have $N/2^kM$ runs of size $2^kM$ each. 
Thus, the total number of passes required to reduce the number of runs from $N/M$ to one is
$\sum_{k=1}^{\log(N/M)} 2^k= 2(N/M-1) = O(N/M)$
 
In the $k$-th merge phase, in the $j$-th pass, we perform $(1+j)M$ comparisons for each pair.
The total number of comparisons performed is
\[ O(N\log M) + \sum_{k=1}^{log(N/M)} \sum_{j=1}^{2^k} (1+j)M \cdot \frac{N}{2^kM} = O\left(N\log M + \frac{N^2}{M}\right)\]
The total number of elemental reads is, clearly,  $O(N^2/M)$.

\subsection{List Ranking}
\label{list:ranking} 
In this section, we present an algorithm that ranks a list of $N$ nodes 
	in $O(N/M)$ passes without using a PRAM simulation. 
It is assumed that the input list $L$ is stored in an array $A$ and each node of 
the list knows the addresses of its successor and predecessor. Each node $u$ holds
two variables $w(u)$ and $r(u)$ initialised to one and zero, respectively.

Our algorithm repeatedly splices out sets of independent sublists from the remaining list,
and then splices them back in, in the reverse order.

Two sublists of a list $L$ are independent if there is no link in $L$ between a
node in one and a node in the other. A set $S$ of sublists of $L$ is independent 
if its members are pairwise independent. The splicing out of a sublist 
$L'=(a_1,\ldots,a_k)$ involves setting the predecessor $p(a_1)$ of $a_1$ and the
successor $s(a_k)$ of $a_k$, respectively, as the predecessor and successor of
each other; it also involves adding $W=\sum_{i=1}^{k} w(a_i)$ 
to $w(p(a_1))$. (In the above, $a_i$ is the predecessor of $a_{i+1}$.) 
A later splicing in of the sublist involves a reversal
of the above pointer arrangements; it also involves setting 
$r(p(a_1))$ to $r(p(a_1))-W$, and then $r(a_i)$ to $r(p(a_1))+\sum_{j=1}^{i} w(a_j)$. 
The splicing in/out of a set of independent sublists
involves the splicing in/out of its members individually, one after the other.

Divide the array $A$ into segments of size $M$ each; the last segment may be of a 
size less than $M$. Repeatedly, load a segment into the  memory. Splice out the sublist
induced by the nodes of the segment. This can be done when the remaining nodes of the list
stream by. When all the nodes have gone by, send the spliced out nodes too into the output.
Thus in one pass the length of the list reduces by $M$. After $N/M-1$ passes, the list
would fit in the  memory. Rank the list by computing a prefix sum of the
$w$-values along it and storing the result in the corresponding $r$-values.
Thereafter, splice in the segments in the reverse of the order in which they were removed.
The list $L$ would now be ranked.  

That is, a list of $N$ nodes can be ranked in $O(N/M)$ passes.

\subsection{The Tree Algorithms}
\label{tree:algo} 

In this section we show that a variety of fundamental 
	problems on trees can be solved in $O(N/M)$ passes without
	 using PRAM simulations.
Our algorithms are easy to implement and use sorting and list ranking procedures. 
In particular, we consider Euler Tour, rooting of a tree, labelling of a
rooted tree and expression tree evaluation.

\subsubsection{Euler Tour}

An Euler Tour $L$ of a tree $T$ is a traversal of $T$'s edges 
that starts and ends at the same vertex, and uses each edge exactly twice. 
Suppose $T$ is presented as an unordered edge-list.
Replace each edge $e = \{v,u\}$ by two directed edges $(v,u)$
and $(u,v)$; one is the twin of the other. Sort the resultant edge list on the first component of the
ordered pairs. Then all outgoing edges of each vertex $v$ come together. 
Number them consecutively: $e_1,\ldots,e_k$; let $e_{(i+1)\mbox{ mod }k}$ be the
successor of $e_i$. For each $(u,v)$, define {\tt next}$(u,v)$ as the successor 
of $(v,u)$. The {\tt next} pointers define an Euler tour of $T$ \cite{Ja92, zeh}. 
They can be computed for all edges in $O(N/M)$ passes: load $M$ edges
into the  memory; let the other edges stream through; when the twin of
an edge $(u,v)$ in the  memory passes by, copy its successor pointer as 
the {\tt next}  pointer of $(u,v)$. Thus, $M$ edges can be processed in one pass, and
so, a total of $O(N/M)$ passes are required to compute the Euler tour.

\subsubsection{Rooting a tree}

The rooting of a tree $T$ is the process of choosing a vertex $s$ as the root and
labelling the vertices or edges of $T$ so that the labels assigned
to two adjacent vertices $v$ and $w$, or to edge $(v,w)$, are
sufficient to decide whether $v$ is the parent of $w$ or vice
versa. Such a labelling can be computed using the following steps:
\begin{enumerate}
\item Compute an Euler Tour $L$ of tree $T$ 
\item Compute the rank of every edge $(v,w)$ in $L$ 
\item For every edge $(v,w) \in T$ do: if rank of $(v,w)$ $<$ rank of $(w,v)$ then $p(w) = v$, otherwise $p(v) = w$
\end{enumerate}
The first two steps take $O(N/M)$ passes as shown above. 
For step $3$, load the edges into the memory $M$ at a time,
for every edge $(v,w)$, when it is in the memory and edge $(w,v)$ streams by,
orient it.
Thus,
an undirected tree can be rooted in $O(N/M)$ passes.

\subsubsection{Labelling a Rooted Tree}

A labelling of a rooted tree provides useful information about the
structure of the tree. Some of these labellings are defined in
terms of an Euler tour of the tree that starts at the root $s$. These
labelling are preorder numbering, postorder numbering, depth of
each vertex from the root $s$ and the number of descendants of each
vertex.

To compute the preordering numbering, assign to each edge $e=(v,w)$ a weight of one
if $v = p(w)$, zero otherwise. The preorder number of 
each vertex $w \neq s$ is one more than the weighted rank of the edge $(p(w),w)$
in the Euler tour of $T$. The root has preorder number of one. A
postorder numbering can also be computed in a similar fashion.

In order to compute the depth of each vertex, assign to each
edge $e = (v,w)$ a weight of one if $v = p(w)$, $-1$ otherwise.
The depth of a vertex $w$ in $T$ is the weighted rank of edge
$(p(w),w)$ in the Euler Tour.

In order to compute the number $\vert T(v) \vert$ of descendants
of each vertex $v$, assign weights to each edge the same as for the
preorder numbering. In particular, for
every non-root vertex $v$, let $r_1(v)$ and $r_2(v)$ be the ranks
of the edges $(p(v),v)$ and $(v,p(v))$. Then $T(v)= r_2(v) - r_1(v)
+1$.

As each of Euler tour, list ranking and sorting requires $O(N/M)$
passes on the W-Stream model, the above labellings can all be computed in
$O(N/M)$ passes.

\subsubsection{Expression Tree Evaluation}
In an expression tree, each internal vertex is labelled by a
function which can be computed in $O(1)$ time and each leaf is labelled 
by a scalar. A leaf evaluates to its scalar. The value at an internal node 
is obtained by applying its function to the values of its children.
The problem is to compute the values at all vertices in the tree. 
To solve it, first sort the vertices by depth and parent
in that order so that deeper vertices come first, and 
the children of each vertex are contiguous. Then the vertices are processed in
sorted order over a number of passes. In one pass, we load $M$ vertices 
with known values into the memory and partially compute the functions at 
their respective parents as they stream by revealing the functions they hold.
The computed partial values are also output with the parents. 
Since in one pass $M$ vertices are processed, $O(N/M)$ passes are enough. 

Here we have assumed that the function at each internal node is an associative operation.

\subsection{Maximal Independent Set and $(\Delta + 1)$ Colouring}
\label{mis:colouring}

We consider two different input representations: (i) a set of adjacency lists, and
	(ii) an unordered edge list.

\subsubsection{The input is a set of adjacency lists}

In this representation, all edges incident on a vertex are stored
        contiguously in the input.
We assume that a list of all vertices is stored before the adjacency list; otherwise
	$O(V/M)$ passes are required to ensure that. 
For an edge $\{u,v\}$, its entry $(u,v)$ in the adjacency list of $u$ is treated as an
	outgoing edge of $u$; its twin $(v,u)$ in the adjacency list of $v$ is an
	incoming edge of $u$.

\paragraph{Maximal Independent Set:}

Divide the vertex set $V$ into segments $C_1, \ldots, C_{\lceil V/M \rceil}$,
of size $M$ each, except for the last segment which can be of size less than $M$.
The algorithm has $V/M$ iterations, each of which has two passes.
In $i$-th iteration, read the $i$-th segment $C_i$ into the  memory, and start a streaming of the
edges. For all $v\in C_i$, if an incoming edge $(u,v)$ of $v$ is found to be marked, then
mark $v$. This signifies that a neighbour of $v$ has already been elected into the MIS.
(Initially, all edges are unmarked.) Start another pass over the input.
In this pass, when the adjacency list of $v\in C_i$ steams in, if $v$ is unmarked, then
elect $v$ into the MIS. Use $v$'s adjacency list to mark all unmarked neighbour of $v$ in $C_i$. 
Also mark all entries of $v$'s adjacency list. When the pass is over, every vertex in $C_i$
is either in the MIS or marked. When all $V/M$ iterations are over, every vertex
is either in the MIS or marked. 

That is, the MIS of the graph is computed in $O(V/M)$ passes.

\paragraph{$(\Delta+1)$ vertex colouring:}

For each vertex $v$, append a sequence of colours $1,\ldots,\delta(v) + 1$, called
the palette of $v$, at the back of $v$'s adjacency list; $\delta(v)$ is the degree of $v$.
Divide the vertex set $V$ into segments $C_1, \ldots, C_{\lceil V/M \rceil}$
of size $M$ each, except for the last segment which can be of size less than $M$.
The algorithm has $V/M$ iterations, each of which has two passes.
In the $i$-th iteration, read the $i$-th segment $C_i$ into the  memory, and start a streaming of the
edges. When the adjacency list of $v\in C_i$ steams in, use it to mark all coloured neighbours 
of $v$ in $C_i$. When the palette of $v$ arrives, give $v$ the smallest colour 
that is in the palette but is not used by any of its coloured neighbours in $C_i$.
When the pass is over, every vertex in $C_i$ is coloured.
Start another pass meant for updating the palettes. 
When the adjacency list of an uncoloured vertex $u$ arrives,
use it to mark all the neighbours of $u$ in $C_i$. When the palette of $u$
arrives, delete from it the colours used by the marked vertices. 

That is, a graph can be $(\Delta+1)$ vertex coloured in $O(V/M)$ passes.

\subsubsection{Input is an edge list}

The above algorithms process the vertices loaded into the memory in the order
in which their adjacency lists stream in. They will not work if the input
is an unordered set of edges, for which case we now present alternative algorithms.

Divide the vertex set $V$ into segments 
	$C_1, \ldots, C_{\lceil V/x\rceil}$,
        of size $x$ each, except for the last segment which may be of a smaller size;
$x$ is a parameter to be chosen later.
The algorithm has $V/x$ iterations, each of which has two passes.
In the $i$-th iteration, we store $C_i$ into the  memory. We also maintain in the memory the
adjacency matrix $A$ of the subgraph induced by $C_i$ in $G$. 
Start a streaming of the edges, and use them to fill $A$, and also to mark the
loaded vertices if they are adjacent to vertices elected into the MIS in earlier iterations.
Remove the marked vertices from the subgraph $G[C_i]$, and compute in-core an MIS of 
the remaining graph. When the pass is over, every vertex in
$C_i$ is either in the MIS or marked.
Read the stream in again, and mark all edges whose one endpoint is
        in MIS.
Thus, a total $O(V/x)$ passes are required.
The total space required in the  memory is $x\log V + x^2$, and that must be $O(M \log V)$.
Therefore, $x$ must be $O(\min\{M, \sqrt{M\log V}\})$.
That is, the MIS of the graph is computed in $O(V/M)$ passes, when $V\geq 2^M$,
and in $O(V/\sqrt{M\log V})$ passes, otherwise.

In a similar manner, we can also colour the graph with $(\Delta + 1)$ colours in 
$O(V/x)$ passes for $x = O(\min\{M, \sqrt{M\log V}, \frac{M\log V}{\Delta\log\Delta} \})$.
We keep palettes of possible colours with the memory loaded
vertices. This would require $x(\Delta+1)\log\Delta$ additional bits in the memory. 
Then, $x\log V + x^2+ x(\Delta+1)\log\Delta$ must be $O(M \log V)$.

\subsection{Single Source Shortest Paths}
\label{sssp:wstream} 

Now we present a randomised algorithm for the $\epsilon$-approximate single source 
shortest paths problem. The input is a weighted graph $G = (V,E)$ with a
non-negative integer weight associated with each edge,
and a source vertex $s$ from which to find an $\epsilon$-approximate shortest
path for each vertex $v \in V$.
The sum of weights of all edges is $W$, and $\log W = O(\log V)$.

The randomized algorithm of Demetrescu et al. \cite{DFR06} solves SSSP 
in $O((C V \log V) / \sqrt{M})$ passes, where $C$ is the largest 
weight of an edge in $G$, and $\log C = O(\log V)$.

Our algorithm uses a subroutine from \cite{DFR06}, and some ideas from \cite{KlSu97}.

\subsubsection{Demetrescu et al.'s Algorithm}

We now briefly describe the algorithm of Demetrescu et al.

The algorithm first picks a subset $A$ of vertices such that 
$A$ includes the source vertex $s$, the other vertices of $A$ are picked uniformly randomly,
and $|A|=\sqrt{M}$. 
A streamed implementation of Dijkstra's algorithm (which we call ``StreamDijkstra'') computes exact shortest 
	paths of length at most $l =
\frac{\alpha C V \log V}{\sqrt{M}}$ (for an $ \alpha > 1$) from each vertex $u \in A$ in
	$O(\frac{V}{\sqrt{M}} + l)$ passes.

Next an auxiliary graph $G'$ is formed in the working memory on vertex set $A$, where the weight
	$w'(x,y)$ of an edge $\{x,y\}$ in $G'$ is set to the length of the shortest path from $x$ to $y$ 
	found above. 
SSSP is solved on $G'$ using $s$ as the source vertex. 
This computation takes place within the working memory.
For $x\in A$, let $P'(x)$ denote the path obtained by taking the shortest path in $G'$
	from $s$ to $x$, and replacing every edge $\{u,v\}$ in it by the shortest path
	in $G$ from $u$ to $v$.
For $v \in V$, do the following: For $x$ in $A$, concatenate $P'(x)$ with the shortest path
	from $x$ to $v$ found in the invocation of StreamDijkstra to form $P'(x,v)$.
Report as a shortest path from $s$ to $v$ the shortest $P'(x,v)$ over all $x\in A$.
This reporting can be done for all vertices in $O(V/M)$ passes, once the auxiliary graph is computed.
The total number of passes is, therefore, $O(\frac{C V \log V}{\sqrt{M}})$.
It can be shown that the reported path is a shortest path with high probability.
See \cite{DFR06}.

Now we describe StreamDijkstra in greater detail, as we will be using it as a subroutine.
StreamDijkstra takes a parameter $l$.

Load $A$ into the  memory. Recall, $|A|=\sqrt{M}$.
Visualise the input stream as partitioned as follows: 
$\gamma_1, \delta_1, \ldots, \gamma_q, \delta_q$, where each $\delta_i$ is an empty sequence, 
each $\gamma_i$ is a sequence of edges $(u,y_i,w_{uy_i})$, and $\forall i<q$, $y_i\not= y_{i+1}$. 
For $c_j\in A$, let $P_{c_{j}}=\{c_j\}$ and $d_j = 0$; $P_{c_{j}}$ will always have
	a size of at most $\sqrt{M}$ and will stay in the  memory.
Execute the following loop:
\begin{tabbing}
aaaa \= aaaa \= aaaa \kill
loop \\
     \> Perform an extraction pass; \\
     \> Perform a relaxation pass; \\
     \> if every $P_{c_j}$ is empty, then halt; \\
endloop
\end{tabbing}

In a extraction pass, stream through the $\gamma$-$\delta$ sequence;
in general, $\delta_i$ is a sequence $(d_{i1},f_{i1}),\ldots,(d_{i\sqrt{M}},f_{i\sqrt{M}})$ where $d_{ij}$
is an estimate on the distance from $c_j$ to $y_i$ through an edge in $\gamma_i$,
and $f_{ij}$ is a boolean that is $1$ iff $y_i$ is settled w.r.t. $c_j$.
For $j=1$ to $\sqrt{M}$, let $d_j$ be the smallest $d_{ij}$ over all $i$ such that $y_i$ is unsettled w.r.t. $c_j$.
For $j=1$ to $\sqrt{M}$, copy into $P_{c_j}$ at most $\sqrt{M}$ $y_i$'s 
so that $d_{ij}=d_j\leq l$.

In a relaxation pass, stream through the $\gamma$-$\delta$ sequence; 
As the pass proceeds, $i$ varies from $1$ to $q$.
For $j=1$ to $\sqrt{M}$, initialise $X_j=\infty$. 
For each $(u,y_i,w_{uy_i})\in \gamma_i$, and for each $c_j\in A$, if $u\in P_{c_j}$, and
$X_j>d_j+w_{uy_i}$ then set $X_j=d_j+w_{uy_i}$.
For each $(d_{ij},f_{ij})\in\delta_{i}$, if $(X_{j}\not=\infty)$ and $d_{ij}>X_j$ then 
set $d_{ij}$ to $X_j$.
For each $y_i \in P_{c_j}$ and $(v , y_i , w_{vy_i}) \in \gamma_i$, set flag $f_{ij} = 1$. 

\subsubsection{Our Approximate Shortest Paths Algorithm}

We use ideas from \cite{KlSu97} to compute in $O(\frac{V \log V \log W}{\sqrt{M}})$ passes
paths that are approximate shortest paths with high probability; here $W$ is the sum of the edge weights.
Our algorithm invokes Procedure StreamDijkstra $\lceil \log W \rceil$ times in as many phases.
 
The algorithm first picks a subset $A$ of vertices such that 
$A$ includes the source vertex $s$, the other vertices of $A$ are picked uniformly randomly,
and $|A|=\sqrt{M}$. Let $l'=\frac{\alpha V\log V}{\sqrt{M}}$, for an $\alpha>1$.
For $1$ to $\lceil \log W \rceil$, execute the $i$-th phase.
Phase $i$ is as follows: 
\begin{itemize}
\item Let $\beta_i=(\epsilon \cdot 2^{i-1})/l'$.

\item Round up each edge weight upto the nearest multiple of $\beta$.
Replace zero with $\beta$. Formally, the new weight function $w_i$ on edges is defined as follows:
$w_i(e) = \beta_i \lceil w(e) / \beta_i \rceil$, if $w(e) > 0$; $w_i(e)=\beta_i$, if $w(e) = 0$.

\item Let $l = \lceil \frac{2 (1 + \epsilon) l'}{\epsilon} \rceil$.
Invoke Procedure StreamDijkstra with $l\beta_i$ as the input parameter.  

\item For each vertex $x \in A$ and $v \in V$, if $p_i(x,v)$ and $\hat{P}(x,v)$ are the shortest 
paths from $x$ to $v$ computed in the above, and in the earlier phases respectively, then 
set $\hat{P}(x,v)$ to the shorter of $p_i(x,v)$ and $\hat{P}(x,v)$. 
\end{itemize}

If $P$ is a path from $x\in A$ to $v\in V$ such that its length is between $2^{i-1}$
to $2^i$ and the number of edges in it is at most $l'$, then
the length $w(p_i)$ of the path $p_i(x,v)$ computed in the $i$-th phase above is at most
$(1 + \epsilon)$ times the length $w(P)$ of $P$.
We can prove this as follows. (A similar proof is given in
\cite{KlSu97}.)

We have, $w_i(e) \leq w(e) + \beta_i$.
So, $w_i(P) \leq w(P) +  \beta_i l'=w(P)+\epsilon 2^{i-1}$. 
As $2^{i-1}\leq w(P)$, this means that $w_i(P) \leq (1 + \epsilon)w(P)$.
Furthermore, since $w(P) \leq 2^i$, $w_i(P) \leq (1 + \epsilon) 2^i$.
Thus, if Procedure StreamDijkstra iterates at least $\frac{(1+\epsilon)2^i}{\beta_i}
=\frac{2(1+\epsilon)l'}{\epsilon}\leq l$ times, then $P$ would be encountered by
it, and therefore $w(p_i)$ would be at most $(1 + \epsilon)w(P)$.
Thus, $\hat{P}(x,v)$ at the end of the $\lceil \log W \rceil$-th phase, will indeed
be an $\epsilon$-approximate shortest path of size at most $l'$, 
for every $v \in V$ and $x \in A$.

Since, each phase requires  $O(\frac{V}{\sqrt{M}} + l)$ passes, where
$l = \lceil \frac{2 (1 + \epsilon) \alpha V \log V }{\epsilon \sqrt{M}}\rceil$,
the total number of passes required for
$\lceil \log W \rceil$ phases is $O(\frac{(1 + \epsilon) V \log V \log W}{\epsilon \sqrt{M}})$.

The rest is as in the algorithm of Demetrescu et al.
An auxiliary graph $G'$ is formed in the working memory on vertex set $A$, where the weight
	$w'(x,y)$ of an edge $\{x,y\}$ in $G'$ is set to the length of the shortest path from $x$ to $y$ 
	found above. 
SSSP is solved on $G'$ using $s$ as the source vertex. 
For $x\in A$, let $P'(x)$ denote the path obtained by taking the shortest path in $G'$
	from $s$ to $x$, and replacing every edge $\{u,v\}$ in it by the reported path $\hat{P}(u,v)$
	in $G$ from $u$ to $v$.
For $v \in V$, do the following: For $x$ in $A$, concatenate $P'(x)$ with the reported path $\hat{P}(x,v)$.
Report as a shortest path from $s$ to $v$ the shortest $P'(x,v)$ over all $x\in A$.

\begin{lemma}
Any path computed by our algorithm has a length of at most $(1 +
\epsilon)$ times the length of a shortest path between the same endpoints with probability at
least $1 - 1/V^{\alpha - 1}$.
\end{lemma}

\begin{proof}
The proof is similar to the one in \cite{DFR06}. The lemma is obvious for
shortest paths of at most $l'$ edges. Now consider a path $P$ of
$\tau>l'$ edges. $P$ has $\lfloor \tau / l' \rfloor$
subpaths of size $l'$ each, and a subpath of size at most $l'$.
We show that each subpath contains at least one vertex
$x$ from set $A$. The probability of not containing any vertex
from $A$ in a subpath is at least $(1 - \vert A \vert /V)^{l'} <
2^{- \frac{\vert A \vert l'}{V}} = 1 /V^\alpha$.

Since, there are at most $V/l' \leq V$ disjoint subpaths,
the probability of containing a vertex from set $A$ in
each subpath is at least $1 - (1/V^{\alpha-1})$. Furthermore, each subpath is of size at most $(1 +
\epsilon)$ times the shortest path of $G$. Thus, the computed
path is at most $(1 + \epsilon)$ times the shortest path with probability $1 - 1 / V^{\alpha -1 }$. 
\end{proof}

\noindent
Putting everything together, we have the following lemma.
\begin{lemma}
The paths computed by our algorithm are $\epsilon$-approximate shortest paths with high probability;
the algorithm runs in $O(\frac{(1 + \epsilon) V \log V \log W}{\epsilon \sqrt{M}})$ passes.
\end{lemma}
If $C$ is the maximum weight of an edge, then $W \leq VC$, and
our algorithm improves the Demetrescu et al.'s  algorithm \cite{DFR06} by a factor of $C/\log{VC}$ at the cost a small error in accuracy of negligible probability.  

\subsection{Maximal Matching}
\label{matching}

Here we show that a maximal matching of a graph can be computed in $O(V/M)$ 
	passes. In each pass, our algorithm executes the following steps:

\begin{enumerate}
\item While the memory holds less than $M$ edges, if the incoming edge $e$ is
independent of all the edges held in the memory, add $e$ to the memory; else discard $e$.
Let $L$ denote the set of edges in the memory when the while-loop terminates.
Clearly, $L$ is a matching in $G$. Nothing has been output till this point in the algorithm.

\item Continue with the streaming of the edges. Discard the edges whose one end point 
is in $L$, write out the others.

\item store the edges of $L$ at the end of the input stream.
\end{enumerate}

Repeat passes until the edge stream is empty. It is easy to
prove that this algorithm correctly computes a maximal matching. Since in
each pass we remove $\Theta (M)$ vertices and its
incident edges from the graph, a total of $O(V/M)$ passes are sufficient.

\subsection{Vertex Cover}
\label{vertex:cover} 

A vertex cover in an undirected graph $G =(V,E)$ is a set of vertices 
$S$ such that each edge of $G$ has at
least one endpoint in $S$. Computing a vertex cover of minimum
size is an NP-complete problem \cite{GJ79}. It is a well-known
result that the endpoints of the edges in a maximal
    matching form a vertex cover whose weight is at most twice that of a minimum
    vertex cover \cite{Va01}.
Our above algorithm for maximal matching, thus, computes a $2$-approximate vertex cover in
    $O(V/M)$ passes.

We now show that a weighted vertex cover of approximation
    ratio $2$ can also be found in $O(V/M)$ passes.
In the weighted vertex cover problem a positive weight is associated with
    each vertex, and the size of a vertex cover is defined as the sum of the weights of the vertices 
    in the vertex cover.
Our algorithm uses the idea of Yehuda et al. \cite{YBF+04}, which involves executing the following
\begin{quote}
If there exists an edge $(u,v)$ such
that $\varepsilon=\min \{ \mbox{weight}(u), \mbox{weight}(v)\} > 0$ then set $\mbox{weight}(u) =  \mbox{weight}(u) - \varepsilon $ and
    $\mbox{weight}(v) = \mbox{weight}(v) - \varepsilon$
\end{quote}
until there does not exist an edge whose both endpoints have nonzero weight.
Let $C$ be the set of the vertices whose weights reduce to $0$.

$C$ is a vertex cover of weight at most twice that of a minimum vertex cover.
We can prove this as follows \cite{YBF+04}. Consider the $i$-th round of the algorithm. Let
$(u, v)$ be the edge selected in this round and $\varepsilon_i$ be
the value deducted at $u$ and $v$. Since every vertex
cover must contain at least one of $u$ and $v$, decreasing both
their values by $\varepsilon_i$ has the effect of lowering the
optimal cost, denoted as $C^*$, by at least $\varepsilon_i$. Thus
in the $i$-th round, we pay $2 \varepsilon_i$ and effect a
drop of at least $\varepsilon_i$ in $C^*$. Hence, local ratio
between our payment and the drop in $C^*$ is at most $2$ in each round.
It follows that the ratio between our total payment and total drop
in $C^*$, summed over all rounds, is at most 2.

Proceeding as in the maximal matching algorithm, in one pass, the weight of 
$O(M)$ vertices can be reduced to $0$. Therefore, $O(V/M)$ passes are required to execute
the above algorithm. Thus, a $2$-approximate weighted vertex cover can
be computed in $O(V/M)$ passes, if the weight of each vertex is $V^{O(1)}$.

\section{Conclusions from this Chapter}
\label{conclusion:wstream}

For list ranking and some tree problems, solved before now using PRAM
simulations, we present alternative algorithms that avoid PRAM simulations. 
While PRAM simulations are
helpful in establishing a bound, they are hard to implement. 
Thus, our algorithms are easier. 

Our results on the maximal independent set and $(\Delta+1)$-colouring problems
show that the hardness of a problem may lie to an extent in the input 
representation.
 
     \chapter{Two Variants of the W-Stream Model and Some Algorithms on Them}
\label{modstream:chapt}
\section{Introduction}

The classical streaming model, which accesses the input data in the form of a stream,
	has been found useful for data-sketching and statistics problems \cite{DFR06}, but 
       classical graph problems and geometric problems are found to be hard to solve on it. 
Therefore, a few variants of the stream model have been proposed.
One such is the W-Stream model \cite{R03, DFR06}.
In this chapter, we propose two further variants, and on them, 
	design deterministic algorithms for the
	maximal independent set and $(\Delta+1)$ colouring problems on general graphs, and 
	the shortest paths problem on planar graphs.
The proposed models are suitable for offline data.

\subsection{Definitions}

Let $G = (V,E)$ be an embedded planar graph with nonnegative integer weights.
Then $E \leq 3V - 6$. A separator for $G = (V,E)$ is a subset $C$
of $V$ whose removal partitions $V$ into two disjoint subsets $A$
and $B$ such that any path from vertex $u$ of $A$ to a vertex $v$ of $B$
in $G$ contains at least one vertex from $C$.

See Chapter~\ref{wstream:chapt}  
for definitions of the maximal independent 
	set, $\Delta+1$ colouring, single source shortest paths, and breadth first search 
	problems.
The all pairs shortest paths problem is to compute a shortest path between every pair of vertices. 

\subsection{Some Previous Results}

The streaming model which was introduced in
	\cite{AlMaSz99, HRR99, MP80}, contains only one read-only input stream and
	uses a polylogarithmic sized working memory. 
Only a constant number of read-only passes are allowed, where one read-only pass is to read 
	the input stream sequentially from the beginning to the end. 
Due to such restrictions, this model can compute only approximate solutions for many problems
	\cite{DFR06}.
Many graph and geometric problems have been considered hard to solve on this model.  
Therefore, a few variants have been proposed of this model. 
One of them is the W-Stream model \cite{DFR06, R03}; this allows the input stream to
	be modified during a pass.
Various graph problems have been solved in this model. See the references \cite{DEM+07, DFR06}.
Problems on special graphs like planar graph etc. have not been explored in the W-Stream model.

\subsection{Our Results}
In this chapter, we propose two models which are variants of the 
	W-Stream model. We give 
	the following algorithms that run on two of those models: an
	$O(V/M)$ passes maximal independent set algorithm and an
	$O(V/x)$ passes $(\Delta+1)$-colouring algorithm,where $x = O(\min \{M, \sqrt{M\log V}\})$, both for general graphs, 
	and an $O((\sqrt{V}+ \frac{V}{M})\log V+\frac{V}{\sqrt{M}})$ passes single source shortest paths algorithm 
	and an $O(\frac{V^2}{M})$ passes all pairs shortest paths algorithm, both
for planar graphs.

\subsection{Organisation of This Chapter}

In Section~\ref{model:modified:wstream}, we propose two variants of the W-Stream model.
In Section~\ref{algo:mod:wstream}, we present some algorithms that run on two of those variants.
In particular, in Subsection~\ref{mis:mod:wstream}, we give a maximal independent
	set algorithm.
In Subsection~\ref{colouring:mod:wstream}, we give a $(\Delta+1)$-colouring algorithm.
In Subsection~\ref{sssp:mod:wstream}, we present an SSSP algorithm for planar graphs.

\section{Two Variants of the W-Stream Model}
\label{model:modified:wstream}
The W-Stream model is described in detail in Chapter~\ref{intro:chapt}.
This model has the following parameters: $N$ is the size of the input, and
        $M \log N$ is the size of the working memory. $P$ the number of passes
        executed by an algorithm is the metric of its performance.
Let the input stream be denoted by $S_0$.
In $i$-th pass, for $i\geq 0$, stream $S_i$ is read and modified (only sequential read
        and write) into stream $S_{i+1}$.
$S_{i+1}$ will be read in the $(i+1)$-st pass.
The size of $S_{i+1}$ can be a constant factor larger than the size of $S_0$.
Streams $S_1, S_2, \ldots , S_{i-1}$ are not used in passes $j\geq i$.
Offline processing is indicated as intermediate streams are allowed.

The W-Stream model handles two streams at a time, one for input and the other for output.
The rate of access need not be the same on the two streams.
It is as if there are two read-write heads, one for each stream, that are handled independently;
        the heads read or write only in their forward movements; a rewind of the head
        in the backward direction, during which no data is read or written, signals the end of a pass;
        a restart of the forward movement marks the beginning of the next pass.
Given the above, the implicit assumption in the W-Stream model that the ends of passes on
the two streams must synchronise seems too restrictive. One might as well
        let one head make more than one pass during just one pass of the other.
In this spirit, we suggest the following two variants of the W-Stream model.

Note that a presence of multiple storage devices with independently handled heads
has been assumed in memory models earlier too \cite{BeJaRu07, GrSc05, ViSh94}.

\paragraph{Model ${\cal M}1$:}
Assume two storage devices (say, disks) $D_1$ and $D_2$ that
can read and write the data independently. In each pass, each disk
assumes only one mode---read or write, and the pass is, accordingly, called a read or write pass.
A disk can perform several read passes while the other is involved in just one
        read/write pass.
The output of any write pass has a length of $O(N)$,  where $N$ is the length
        of the input to the $0$-th pass.
If both disks are in read mode, then they must be reading the
last two output streams produced by the algorithm before then.
If $P_1$ and $P_2$ are the passes executed by an algorithm on the two disks,
        we say, the algorithm runs in $P_1 + P_2$ passes.

This model is stronger than the W-Stream model.
The bit vector disjointness problem can be solved on it in $O(1)$ passes:
given a pair $(A,B)$ of bit sequences of length $N$ each, in one pass over the input disk,
        copy $B$ into the other disk, and then execute read passes on both disks
        concurrently.
On the W-Stream model, as mentioned before, the problem requires $\Omega(N/(M\log N))$ passes.

\subsubsection{Model ${\cal M}2$:}
This is similar to the ${\cal M}1$ except in that at any given time, only one disk is
        allowed to be in a read pass and only one disk is allowed to be in a write pass.

Clearly, ${\cal M}2$ is at least as strong as the W-Stream model, and ${\cal M}1$ is at least
as strong as ${\cal M}2$.

\section{The Algorithms}
\label{algo:mod:wstream}
In this section, we present algorithms for the
	maximal independent set and $(\Delta+1)$ colouring problems on general graphs, and 
	the shortest paths problem on planar graphs.
The algorithms run on both ${\cal M}1$ and ${\cal M}2$ in the same number of passes.
For all algorithms the input is an unordered edge list.

\subsection{Maximal Independent Set}
\label{mis:mod:wstream} 

A maximal independent set can be computed by repeating the following greedy
strategy until no vertex is left:
Select a remaining vertex into the independent set, and remove it and all its adjacent
	vertices from the graph.
We use the same in our algorithm. 

Divide the vertex set $V$ into segments $C_1, \ldots, C_{2V/M}$,
of size $M/2$ each, except for the last segment which can be of size less than $M/2$.
For each segment $C_i$, let $E_1(C_i)$ denote the set of edges with both endpoints
	are in $C_i$ and let $E_2(C_i)$ denote the set of edges with exactly
	one endpoint is in $C_i$. 
From the unordered input stream of edges we construct the following stream:
\[ \sigma = \langle E_2(C_1), \; C_1, \; E_1(C_1), \ldots E_2(C_{2V/M}), \; C_{2V/M}, \;  E_1(C_{2V/M}) \rangle\]
Make multiple read passes over the input, while the output is being written.
For segment $C_i$, in the first read pass over the input, filter the edges in $E_2(C_i)$ 
into the output, then stream out $C_i$, and in another read pass over 
the input, filter out the edges in $E_1(C_i)$. 
A total of $O(V/M)$ read passes over the input are enough to prepare the sequence.

For all $i$, concurrently sort $E_1(C_i)$ so that, for each node, all its outgoing edges in $E_1(C_i)$ come together.
Use the  algorithm of the previous chapter for sorting. 
Note that an algorithm designed on the W-Stream model can be executed on  ${\cal M}1$ and ${\cal M}2$
in the same number of passes without any change.
Since the size of each $E_1(C_i)$ is $O(M^2)$, and the sorts proceed concurrently,  
	a total of $O(M)$ passes are sufficient.

The algorithm has $O(V/M)$ iterations, each of which has a single pass.
In the $i$-th iteration, read the $i$-th segment $C_i$ into the  memory.
Some vertices of $C_i$ could be marked signifying that a neighbour has 
already been elected into the MIS, unless $i=1$.
Start a streaming of $\sigma$. Use the edges in $E_1(C_i)$ to execute the following:
for all unmarked $u\in C_i$, elect $u$ into the MIS and mark all its neighbours in $C_i$.
For $j>i$, when $E_2(C_j)$ streams in, use it to find the neighbours that the newly
elected vertices have in $C_j$, and mark all those when $C_j$ streams in later.
At the end of the pass over $\sigma$, for every vertex $v$ in $C_i$, either $v$ is marked, or
$v$ is in the MIS and every neighbour of $v$ is marked.
When all iterations are over, every vertex
is either in the MIS or marked. 

Thus, we obtain the following lemma.
\begin{lemma}
The MIS of the graph is computed in $O(V/M)$ passes on both ${\cal M}1$ and ${\cal M}2$,
even when the input is an unordered edge-list.
\end{lemma}

\subsection{$(\Delta + 1)$-colouring Problem}
\label{colouring:mod:wstream} 

Divide the vertex set $V$ into segments 
	$C_1, \ldots, C_{V/x}$,
        of size $x$ each, except for the last segment which may be of a smaller size;
	$x$ is a parameter to be chosen later.
For each segment $C_i$, let $E(C_i)$ denote the set of edges with at least one endpoint
	in $C_i$, and let $\colour{C_1}$ denote a specification, for each vertex $v\in C_i$,
	of a palette of available colours for $v$. 
From the unordered input stream of edges we construct the following stream:
\[ \sigma=\langle C_1, \; E(C_1), \; \colour{C_1}, \; \ldots, \; C_{V/x}, \; E(C_{V/x}), \; \colour{C_{V/x}} \rangle\] 
The palette of $v$ is initialised here
	to $\{1,\ldots,\delta(v)+1\}$, where $\delta(v)$ is the degree of $v$.
A total of $O(V/x)$ read passes over the input are enough to prepare the sequence, if $x<M$.

The algorithm has $V/x$ iterations, each of which has a single pass.
In the $i$-th iteration, we store $C_i$ into the  memory. We also maintain in the memory the
adjacency matrix $A$ of the subgraph induced by $C_i$ in $G$. 
Start a streaming of $\sigma$, and use $E(C_i)$ to fill $A$.
For each $v\in C_i$, when the palette of $v$ arrives, give $v$ the smallest colour 
that is in the palette but is not used by any of its coloured neighbours in $C_i$.
When $\colour{C_i}$ has streamed by, every vertex in $C_i$ is coloured.
For $j>i$, when $E(C_j)$ streams in, use it to construct in the memory the adjacency matrix of the
subgraph induced by $C_i \cup C_j$, and then as $\colour{C_j}$ streams in, update the
palettes in it using this subgraph.  
At the end of the pass over $\sigma$, for every vertex $v$ in $C_i$, $v$ is coloured, and the palette
of every uncoloured neighbour of $v$ is updated. When all iterations are over, every vertex
is coloured. 

The total space required in the  memory is $x\log V + x^2$, and that must be $O(M \log V)$.
Therefore, $x$ must be $O(\min\{M, \sqrt{M\log V}\})$.

Thus we obtain the following lemma.
\begin{lemma} 
A $(\Delta+1)$-colouring of the graph is computed in $O(V/M)$ passes, when $V\geq 2^M$,
and in $O(V/\sqrt{M\log V})$ passes, otherwise.
\end{lemma}

\subsection{Shortest Paths on Planar Graphs}
\label{sssp:mod:wstream}

Let $\hat{G}=(V,E)$ be the given embedded planar graph.

First we transform $\hat{G}$ into a planar graph $G$ in which every vertex has a degree of at most $3$. 
For every vertex $u$ of degree $d>3$, if $v_0, v_1, \ldots, v_{d-1}$ is a cyclic ordering 
of the neighbours of $u$ in the planar embedding, then replace $u$ with new vertices $u_0$, $u_1$,
$\ldots$, $u_{d-1}$; add edges $\{(u_i,u_{(i+1) {\small \mbox{mod}}\; d}) \vert i = 0, \ldots, d-1\}$, 
each of weight $0$, and for $i=0, \ldots, d-1$, replace edges $(v_i,u)$ with edge $(v_i,u_i)$
of the same weight. Such a transformation can be done in $O(V/M)$ passes.

\subsubsection{Graph Decomposition} 
Next we decompose $G$ into $O(V/M)$ regions.
With respect to these regions, each vertex is categorized as either 
an interior or a boundary vertex
depending upon whether it  belongs to exactly
one region or it is shared among at least two regions.
There is no edge between two interior vertices belonging to two different
regions. The decomposition has the following properties.
\begin{enumerate}
\item Each region has $O(M)$ vertices.
\item The total number of regions is $O(V / M)$.
\item Each region has $O(\sqrt{M})$ boundary vertices.
\item Each boundary vertex is contained in at most three regions.
\end{enumerate}

For the decomposition of $G$, we use an adaptation of 
Frederickson's in-core algorithm \cite{F87} that 
decomposes a graph into a number of regions in $O(V \log V)$ time.
Frederickson's decomposition has been used in shortest paths 
algorithms before \cite{ArBrTo04, KRR+94, TZ00}.
It recursively applies the planar separator theorem of
Lipton and Tarjan \cite{LT79}, which we state now: 
\begin{theorem}
{\em Planar Separator Theorem [Lipton and Tarjan]:} Let $G = (V,E)$
be an $N$ vertex planar graph with nonnegative costs on its vertices
summing upto one, then there exists a separator $S$ of $G$ which
partitions $V$ into two sets $V_1, V_2$ such that $\vert S \vert =
O(\sqrt N)$ and each of $V_1, V_2$ has total cost of at most 2/3.
\end{theorem}
Such separator is called a $2/3$-separator and $(V1, V2, S)$ will be deemed the output of Separation.

Gazit and Miller \cite{GM87a} give a work optimal parallel
    algorithm for finding a $2/3$-separator. 
Their algorithm runs in $O(\sqrt{V}\log V)$ time using $O(\sqrt{V} / \log V)$
processors on a CRCW PRAM. Their algorithm can be simulated on the W-Stream
model. Hereafter, we call this simulation the {\it separator procedure}.

We decompose the graph into
    $O(V / M)$ regions of size $O(M) $ each, with
    $\sqrt{M}$ boundary vertices.
The decomposition runs in two phases, phase $1$ and phase $2$.
\paragraph{phase $1$} 
In this phase, the following loop is executed.
Initially $G$ is the only region and it is marked {\em unfinished}.

\noindent
While there is an unfinished region $R$ do the following:
\begin{enumerate}
\item if $\vert V(R) \vert \leq c_1 M$ and the size of $R$'s boundary is at most $c_2 \sqrt{M}$, then
mark $R$ as {\em finished}; continue to the next iteration;
\item else if $\vert V(R) \vert > c_1 M$ then run the separator procedure
on $R$, after giving each vertex a weight of $1 / \vert V(R) \vert $;
let $A = V_1$, $B = V_2$ and $C = S$ be the output;
\item else (that is, the size of $R$'s boundary is greater than $c_2 \sqrt{M}$), 
run the separator procedure
on $R$, after giving each of its $N'$ boundary vertices a weight of $1/N'$
and each of its interior vertices a weight of $0$;
let $A = V_1$, $B = V_2$ and $C = S$ be the output;

\item compute $C'\subseteq C$ such that no vertex in $C'$ is adjacent to $A
\cup B$;

\item let $C'' = C \setminus C'$; compute the connected
components $A_1, \ldots, A_q$ in $A \cup B \cup C'$;

\item while there exists a vertex $v$ in $C''$ such that it is
adjacent to some $A_i$ and not adjacent to a vertex in $A_j$ for
$j \neq i$ do
\begin{itemize}
\item remove $v$ from $C''$ and insert into $A_i$;
\end{itemize}

\item For $1\leq i\leq q$, let $R_i = A_i \cup \{v\in C'' \;|\; v$ has a neighbour in $A_i \}$;  
mark each $R_i$ {\em unfinished};
\end{enumerate}

We have the following lemma.
\begin{lemma}
(i) At the end of phase $1$, the graph is decomposed into connected subgraphs of size $O(M)$ each,
and the boundary vertices of each connected subgraph has at most $\sqrt{M}$ vertices.
(ii) Phase $1$ can be computed in $O((\sqrt{V}+ \frac{V}{M})\log V)$ passes.
\end{lemma}
\begin{proof}
The correctness (Statement (i) of the Lemma) follows from \cite{F87}.
Here we prove Statement (ii) of the Lemma. 

The algorithm decomposes the given region recursively, with each level of recursion
reducing the size of the region by a factor of at least $2/3$.
We handle all recursive calls at the same level concurrently.

Gazit and Miller \cite{GM87a} give a work optimal parallel
	algorithm for finding a $2/3$-separator. 
Their algorithm runs in $O(\sqrt{N}\log N)$ time using $O(\sqrt{N} / \log N)$
processors on a CRCW PRAM.
A self-simulation of their algorithm on a $p$ processor CRCW PRAM would run
in $O(\sqrt{N}\log N+N/p)$ time.
Suppose we have $x$ regions of sizes $N_1,\ldots,N_x$ to decompose, so that
$N_1+\ldots+N_x=V$. Say $M$ processors are available. 
Let the $i$-th region be allocated $MN_i/V$ processors.
The decompositions of all regions can be done simultaneously in
$O(\sqrt{n}\log n+V/M)$ steps, where $n=\max_{i}\{N_i\}$.
If $n=O(V/d^k)$, then $M\leq\frac{d^k\sqrt{V/d^k}}{\log (V/d^k)}$ implies that
the decompositions of all regions can be done simultaneously in
$O(V/M)$ steps.

In the following we consider a simulation of the above on the W-Stream model.

Theorem 1 in \cite{DEM+07} states that any
PRAM algorithm that uses $p\geq M$ processors and runs in time $T$ using space
        $S = \mbox{poly}(p)$ can be simulated in W-Stream in 
$O((Tp\log S)/$ $(M \log p)) $ passes using $M \log p$ bits of working memory and intermediate
        streams of size $O(M+p)$.
In the simulation, each step of the PRAM algorithm is simulated in W-Stream model in $O(p \log S / (M \log p))$ 
	passes, where at each pass we simulate the execution of $M \log p / \log S$ processors using $M \log p$ bits
	of working memory.
The state of each processor and the content of the memory accessed by the algorithm are maintained in the intermediate
	steps.
we simulate the execution of a processor as follows. We first read from the input stream the state 
	and then read the content of the memory cell used by each processor, and then execute the step of the algorithm.
Finally we write to the output stream the new state and the modified content. Unmodified content is written as it is.
Note that this simulation works on our models also.

We handle the $k$-th level of recursion using a simulation of the following PRAM algorithm
designed on an $M$-processor CRCW PRAM:
let $d=3/2$;
if $M>\frac{d^k\sqrt{V/d^k}}{\log (V/d^k)}$,
decompose the regions obtained from the next higher level of recursion in 
$O(\sqrt{V/d^k}\;{\log (V/d^k)})$ time using $\frac{d^k\sqrt{V/d^k}}{\log (V/d^k)}$ processors;
else, decompose them in $O(V/M)$ time using $M$ processors as shown above.
The total time taken by the PRAM algorithm is $O((\sqrt{V}+ \frac{V}{M})\log V)$.
The W-Stream simulation, therefore, requires $O((\sqrt{V}+ \frac{V}{M})\log V)$ passes.

The remaining steps can all be implemented in $O(V/M)$ passes per level of recursion. 
In particular, for connected components, we adopt a hook and contract strategy \cite{Ja92}.
Every vertex can be made to hook in a single pass, if the graph is in adjacency list
representation. The contraction can be achieved by $O(1)$ invocations to 
sorts, scans, Euler tours and list ranks, as shown in Chapter~\ref{mst:chapt}. 
Each hook-and-contract reduces a planar graph by a constant factor in size. 
Thus, connected component can be computed in $O(V/M)$ passes.

Therefore, the total number of passes is $O((\sqrt{V}+ \frac{V}{M})\log V)$.
\end{proof}

The total number of regions (connected subgraphs) produced by Phase 1 can be greater than 
$O(V/M)$.
In the second phase, some of these regions are combined to reduce the total number of regions
to $O(V/M)$.

\paragraph{Phase $2$}

The phase begins by calling each connected subgraph a region.
It executes the following steps for {\em small} regions; that is, regions
of size at most $c_1 M/2$ and boundary size
    at most $c_2 \sqrt M /2$.
\begin{quote}
(1) while there exist two small regions that share boundary vertices, combine them; \\
(2) while there exist two small regions that are adjacent to the same set of either
    one or two regions, combine them.
\end{quote}

For the correctness of the above, we depend on \cite{F87}, where it is shown that
the number of the regions in the decomposition is $O(V/M)$.
Now we show how to compute the above two steps in $O(V / M)$ passes.
Suppose, the boundary vertices of each region are included and distinguished in the vertex list
for the region; since the total number of boundary vertices is $O(V /\sqrt{M})$ and
each boundary vertex is in at most three regions, this supposition increases
the size of the stream by at most a constant factor.

In a few sorts, each region knows its adjacent regions and boundary vertices
        of those regions.
In one pass $O(M)$ regions can be grouped together.
Therefore Step (1) can be
        done in $O(V/M)$ passes.

Similarly, $O(V/M)$ passes are required to perform step $2$.
Hence we give the following lemma.

\begin{lemma}
After Phase 2, the graph is decomposed into a set of regions that satisfy all the required properties.
Phase $2$ is computed in $O(V/M)$ passes.
\end{lemma}

\subsubsection{Computing the shortest Paths from Source Vertex $s$}
The graph is now decomposed into $\Theta(V/M)$ regions each with a boundary of
at most $\Theta(\sqrt{M})$ vertices. Also, each boundary vertex is shared by at most
three regions. Since each region fits in the memory, in one pass we can compute 
the shortest distance between each pair of boundary vertices in each region;
we do this by invoking an in-core  all pairs shortest paths algorithm. A new
graph $G^R$ is constructed by replacing each region with a
complete graph on its boundary vertices. The weight of each edge in
the complete graph is the shortest distance within the region between its endpoints. 
If the source vertex $s$ is not a boundary vertex, then we
include it in $G^R$ and connect it to the boundary vertices of the
region containing it, and all these new edges are also weighted by the respective
within-the-region distances. This graph has $O(V/\sqrt{M})$ vertices and $O(V)$
edges. The degree of each vertex is $O(\sqrt{M})$. This graph need not be planar.

Now we show how to compute the shortest paths in $G^R$ from the source vertex $s$
using Dijkstra's algorithm, but without using a priority queue. Like
in Dijkstra's algorithm, our algorithm maintains a label $d(v)$ for
each vertex of $G^R$; $d(v)$ gives the length of a path from the source vertex
$s$ to $v$. We say an edge $(v,w)$ is relaxed if $d(w) \leq d(v) + \mbox{weight}(v,w)$.

Initially, $d(v) = \infty$, $\forall v \in V^R$, $v \neq s$.
Assume that $d(v)$ is stored with each vertex in the vertex list and edge list. 
$d(s) = 0$. Execute
the following steps until all edges are relaxed and all vertices
in the vertex list are marked.

\begin{enumerate}
\item Select the vertex $v$ with the smallest label among
all unmarked vertices in the vertex list. Also select the edges incident on $v$
from the edge list. Mark $v$.  
\item Relax the selected edges, if necessary. Labels of some
vertices may need to be changed. Update the labels of these vertices
in the vertex list and edge list.
\end{enumerate}

The above steps are from Dijkstra's algorithm. Therefore the correctness follows.
Each iteration above can be executed in $O(1)$ passes.
In each iteration a vertex is marked. This vertex will not be
selected again. Therefore $\vert V^R \vert$ iterations are sufficient.
As $V^R=O(V/\sqrt{M})$, the number of passes needed is $O(V/\sqrt{M})$.
Thus, we obtain the following lemma.
\begin{lemma}
SSSP and BFS problems can be computed in $O((\sqrt{V}+ \frac{V}{M})\log V+\frac{V}{\sqrt{M}})$ passes on
models ${\cal M}1$ and ${\cal M}2$.
\end{lemma}
If we run the randomized algorithm given in \cite{DFR06} on graph $G^R$ then 
with high probability we can compute the shortest paths from the source vertex
in $O(C V \log V /M)$ passes, which is better than the above when $C \log V < \sqrt{M}$.

\subsubsection{All pair shortest paths}

We can compute the shortest distances from
    $O(\sqrt{M})$ source vertices in $O(V/ \sqrt{M})$
    passes with $O(V \sqrt{M})$ space.
For this, we maintain $O(\sqrt{M})$ distances at each vertex, and obtain the following lemma.
\begin{lemma}
For $V$ source vertices, we can compute the shortest paths in
    $O(V^2 / M)$ passes using $O(V^2)$ space. 
\end{lemma}
\section{Conclusions from this Chapter}
\label{mod:wstream:chapt}
We presented two variants of the W-Stream model on which we found it easy to design
	faster algorithms for a few problems. 
We believe that the lower bounds for these problems on our models would match those on the W-Stream model.
Lower bounds in communication complexity have not proved helpful for our models.
New techniques may have to be developed for proving lower bounds on them.  




\bibliographystyle{plain}
\bibliography{ref}



\end{document}